\documentclass[aps,onecolumn,a4paper,twoside,superscriptaddress,longbibliography,nofootinbib,notitlepage]{revtex4-1}
\usepackage{graphicx}
\usepackage{amsmath}
\usepackage{amssymb}
\usepackage{amsthm}
\usepackage[utf8]{inputenc}
\usepackage{comment}

\usepackage{color}

\newcommand{\abs}[1]{\left\lvert#1\right\rvert} 
\newcommand{\norm}[1]{\left\lVert#1\right\rVert} 

\usepackage[a4paper,total={6.5in,9in}]{geometry}
\usepackage[english]{babel}
\usepackage[T1]{fontenc}
\usepackage{mathrsfs}
\usepackage{braket}
\usepackage{physics}
\usepackage{enumerate}
\usepackage{graphicx}
\usepackage{mathtools}
\usepackage[export]{adjustbox}
\usepackage{microtype}

\theoremstyle{plain}

\usepackage[colorlinks = true,
            linkcolor = red,
            urlcolor  = blue,
            citecolor = blue,
            anchorcolor = red]{hyperref}

\usepackage{orcidlink}

\newtheorem{theorem}{Theorem}
\newtheorem{corollary}[theorem]{Corollary}

\newtheorem{lemma}[theorem]{Lemma}
 
\newtheorem{definition}[theorem]{Definition}  
\newtheorem{remark}[theorem]{Remark}  
\newtheorem*{remark*}{Remark}   
\renewcommand\qedsymbol{$\blacksquare$}
\newenvironment{proof-of}[1][{\hspace{-\blank}}]{{\medskip\noindent\textit{Proof~{#1}.\ }}}{\hfill\qedsymbol}

\renewcommand{\Tr}{{\operatorname{Tr}\,}}
\renewcommand{\Pr}{{\operatorname{Pr}}}
\renewcommand{\dim}{{\operatorname{dim}}}
\newcommand{\id}{{\operatorname{id}}}

\DeclarePairedDelimiter\floor{\lfloor}{\rfloor}

\newcommand{\1}{\openone}
\newcommand{\proj}[1]{|#1\rangle\!\langle #1|}

\newcommand{\cT}{{\mathcal{T}}}

\newcommand{\und}[1]{\underline{#1}}

\newcommand{\nc}{\newcommand}
\nc{\rnc}{\renewcommand}
\nc{\avg}[1]{\langle#1\rangle}
\nc{\Rank}{\operatorname{Rank}}
\nc{\smfrac}[2]{\mbox{$\frac{#1}{#2}$}}
\renewcommand{\tr}{\operatorname{Tr}}
\nc{\ox}{\otimes}
\nc{\dg}{\dagger}
\nc{\dn}{\downarrow}
\nc{\cA}{{\cal A}}
\nc{\cB}{{\cal B}}
\nc{\cC}{{\cal C}}
\nc{\cF}{{\cal F}}
\nc{\cG}{{\cal G}}
\nc{\cH}{{\cal H}}
\nc{\cI}{{\cal I}}
\nc{\cJ}{{\cal J}}
\nc{\cK}{{\cal K}}
\nc{\cL}{{\cal L}}
\nc{\cM}{{\cal M}}
\nc{\cN}{{\cal N}}
\nc{\cO}{{\cal O}}
\nc{\cP}{{\cal P}}
\nc{\cQ}{{\cal Q}}
\nc{\cR}{{\cal R}}
\nc{\cS}{{\cal S}}
\nc{\cX}{{\cal X}}
\nc{\cY}{{\cal Y}}
\nc{\cZ}{{\cal Z}}
\nc{\csupp}{{\operatorname{csupp}}}
\nc{\qsupp}{{\operatorname{qsupp}}}
\nc{\rar}{\rightarrow}
\nc{\lrar}{\longrightarrow}
\nc{\polylog}{{\operatorname{polylog}}}
\nc{\wt}{{\operatorname{wt}}}

\nc{\RR}{{{\mathbb R}}}
\nc{\CC}{{{\mathbb C}}}
\nc{\FF}{{{\mathbb F}}}
\nc{\NN}{{{\mathbb N}}}
\nc{\ZZ}{{{\mathbb Z}}}
\nc{\PP}{{{\mathbb P}}}
\nc{\QQ}{{{\mathbb Q}}}
\nc{\UU}{{{\mathbb U}}}
\nc{\EE}{{{\mathbb E}}}

\nc{\Hom}[2]{\mbox{Hom}(\CC^{#1},\CC^{#2})}
\nc{\rU}{\mbox{U}}

\nc{\ob}[1]{#1}

\nc{\SEP}{{\text{SEP}}}
\nc{\NS}{{\text{NS}}}
\nc{\LOCC}{{\text{LOCC}}}
\nc{\PPT}{{\text{PPT}}}
\nc{\EXT}{{\text{EXT}}}
\nc{\Sym}{{\operatorname{Sym}}}

\nc{\ERLO}{{E_{\text{r,LO}}}}
\nc{\ERLOCC}{{E_{\text{r,LOCC}}}}
\nc{\ERPPT}{{E_{\text{r,PPT}}}}
\nc{\ERLOCCinfty}{{E^{\infty}_{\text{r,LOCC}}}}
\nc{\Aram}{{\operatorname{\sf A}}}

\newcommand{\aw}[1]{{#1}}
\newcommand{\zk}[1]{{#1}}

\begin{document}

\title{Resource theory of heat and work with non-commuting charges}

\author{Zahra Baghali Khanian\,\orcidlink{0000-0002-0892-7519}}
\affiliation{Grup d'Informaci\'{o} Qu\`{a}ntica, Departament de F\'{\i}sica, Universitat Aut\`{o}noma de Barcelona, 08193 Bellaterra (Barcelona), Spain}
\affiliation{ICFO--Institut de Ci\`encies Fot\`oniques, Barcelona Institute of Science and Technology, 08860 Castelldefels (Barcelona), Spain}

\author{Manabendra Nath Bera\,\orcidlink{0000-0002-8329-2656}}
\affiliation{Department of Physical Sciences, Indian Institute of Science Education and Research (IISER), Mohali, Punjab 140306, India}
\affiliation{ICFO--Institut de Ci\`encies Fot\`oniques, Barcelona Institute of Science and Technology, 08860 Castelldefels (Barcelona), Spain}

\author{\\Arnau Riera\,\orcidlink{0000-0002-3271-7802}}
\affiliation{Institut el Sui, Carrer Sant Ramon de Penyafort, s/n, 08440 Cardedeu (Barcelona), Spain}
\affiliation{ICFO--Institut de Ci\`encies Fot\`oniques, Barcelona Institute of Science and Technology, 08860 Castelldefels (Barcelona), Spain}

\author{Maciej Lewenstein\,\orcidlink{0000-0002-0210-7800}}
\affiliation{ICFO--Institut de Ci\`encies Fot\`oniques, Barcelona Institute of Science and Technology, 08860 Castelldefels (Barcelona), Spain}
\affiliation{ICREA--Instituci\'o Catalana de Recerca i Estudis Avan\c{c}ats, 08010 Barcelona, Spain}

\author{Andreas Winter\,\orcidlink{0000-0001-6344-4870}}
\affiliation{Grup d'Informaci\'{o} Qu\`{a}ntica, Departament de F\'{\i}sica, Universitat Aut\`{o}noma de Barcelona, 08193 Bellaterra (Barcelona), Spain}
\affiliation{ICREA--Instituci\'o Catalana de Recerca i Estudis Avan\c{c}ats, 08010 Barcelona, Spain}

\begin{abstract}
We consider a theory of quantum thermodynamics with multiple conserved quantities (or charges). To this end, we generalize the seminal results of Sparaciari \emph{et al.} [\emph{Phys. Rev. A} 96:052112, 2017] to the case of multiple, in general non-commuting charges, for which we formulate a resource theory of thermodynamics of asymptotically many non-interacting systems.
To every state we associate the vector of its expected charge values and its entropy, forming the \emph{phase diagram} of the system. Our fundamental result is the Asymptotic Equivalence Theorem (AET), which allows us to identify the equivalence classes of states under asymptotic approximately charge-conserving unitaries with the points of the phase diagram.

Using the phase diagram of a system and its bath, we analyze the first and the second laws of thermodynamics. In particular, we show that to attain the second law, an asymptotically large bath is necessary. In the case that the bath is composed of several identical copies of the same elementary bath, we quantify exactly how large the bath has to be to permit a specified work transformation of a given system, in terms of the number of copies of the ``elementary bath'' systems per work system (bath rate). If the bath is relatively small, we show that the analysis requires an extended phase diagram exhibiting negative entropies. This corresponds to the purely quantum effect that at the end of the process, system and bath are entangled, thus permitting classically impossible transformations (unless the bath is enlarged). For a large bath, or many copies of the same elementary bath, system and bath may be left uncorrelated and we show that the optimal bath rate, as a function of how tightly the second law is attained, can be expressed in terms of the heat capacity of the bath.

Our approach, solves a problem from earlier investigations about how to store the different charges under optimal work extraction protocols in physically separate batteries.
\end{abstract}


\date{18 November 2022}

\maketitle

\section{Introduction}
\label{sec:intro}
Thermodynamics is one of the most successful physical theories and a pillar of modern science and technology. 
It was initially developed empirically to describe heat engines, such as the steam engine and 
internal combustion engines that powered the industrial revolution of the 18th and 19th century. 
Later on, it has been founded on statistical mechanics with the assumption that the systems 
are composed of a large number of classical particles. 
The thermal baths, which the system interacts with, are even larger in size so that the temperature of 
the bath effectively does not alter after the interaction. The laws of thermodynamics find their 
applications in almost all branches of the exact sciences. The emergence of quantum mechanics 
in the last century, and the subsequent achievements in controlling and tuning of an individual or a 
finite number of quantum systems, led to the exploration of thermodynamics in the quantum regime. 
There, the system is made up of a single or moderate number of quantum particles 
interacting with a thermal bath. This regime is often termed the \emph{finite-size} regime. 
The system may possess nontrivial quantum correlations, such as entanglement among the particles, 
and the bath can be finite or comparable in size with the system. In the quantum domain, another 
layer of difficulties arises when one considers more than one conserved quantities (charges) 
that do not commute with each other, as the simultaneous conservation of all the charges 
cannot be guaranteed. 

Recent studies of quantum thermodynamics \cite{Gemmer2009,book2} focus on systems of finite 
size and the cases where measurements are allowed only once. 
In addition to thermodynamic averages, there one is interested in 
values and bounds on fluctuations of thermodynamic quantities. One way to handle these 
problems is by the use of various {\it fluctuations theorems} \cite{Jar97,ro99,cht11}. 
Another way to deal with these regimes is exactly via the resource theory of thermodynamics 
that allows for rigorous treatment of second laws, optimal work extraction problem, 
etc (cf. \cite{bho13,h&p13,Brandao2015}, see also \cite{ssp14,abe13,brl17,Bera2017}). 
The resource theory of quantum thermodynamics was recently extended to deal with quantum 
and nano-scale engines made up of a finite or a small number of quantum particles, and 
two baths at different temperatures \cite{blb19}. 

Resource theory is a rigorous mathematical framework initially developed to characterize the 
role of entanglement in quantum information processing tasks. Later the framework was extended 
to characterize coherence, non-locality, asymmetry and many more, including quantum Shannon theory itself, see \cite{bcp14,WinterRT,c&g16,m&s16,
V&S17,msz16,sap16,srb17,g&w19,cpv18,
sha19,Vic14,d&a18,thermo_multiple_charge1,
thermo_multiple_charge2,thermo_multiple_charge3,
DevetakHarrowWinter:Shannon,l&w19}. 
The resource theory approach applies also to classical theories. In general, the resource 
theories have the following common features: (1) a well-defined set of resource-free states, 
and any states that do not belong to this set has a non-vanishing amount of resource; 
(2) a well-defined set of resource-free operations (allowed operations), 
that cannot create or increases resource in a state. These allow one to quantify the 
resources present in the states or operations and characterize their roles in the transformations 
between the states or the operations. In particular, it enables the definition and rigorous 
calculation or bounding of resource measures; to determine which states can be 
transformed to others using allowed operation; how the resource content of states may be 
changed, and how these changes are bounded under the allowed operations, etc.

In the present paper, we formulate a resource theory of quantum thermodynamics with 
multiple conserved quantities, where the system and bath a priori are arbitrary in size. 
We adhere to the asymptotic regime where a system of many non-interacting particles with 
multiple conserved quantities or charges interacts with a bath. 
\zk{It is discussed in \cite{Lostaglio2017} that in the resource theory of thermodynamics 
with multiple non-commuting conserved quantities, complete passivity and maximum entropy 
principle lead to incompatible sets of resource free states.  
Here we choose the maximum entropy principle, that is the resource-free states 
are the generalized Gibbs states (GGS),
and allowed operations are the (average) entropy and (average) charge preserving operations.} 
The thermodynamic resource is quantified by the Helmholtz free entropy. Clearly, the 
entropy and charge preserving operations cannot create thermodynamic resource in the resource-free GGSs. 
For any quantum state, we associate a vector with entries of the average charge 
values and entropy of that state. We call the set of all these vectors the phase diagram 
of a system. The concept of phase diagram in the present sense was originally pioneered 
in \cite{Sparaciari2016} for a system with energy as the only conserved quantity of the 
system where it has been shown that the phase diagram is a convex set. 
\aw{This terminology is motivated by traditional thermodynamics, where the phase 
diagram is a multi-dimensional map of the equilibrium states of a system 
according to the temperature and other relevant intensive or extensive 
quantities (such as pressure, volume, particle concentrations, etc). The difference 
here is that we also allow non-equilibrium states, effectively decoupling the 
entropy from the dynamical parameters. 
The seminal results of \cite{Sparaciari2016} were generalized to multiple 
\emph{pairwise commuting} conserved quantities by the present authors \cite{Bera2017}, 
and the further generalization to the case of multiple, in general non-commuting 
charges is the subject of the present paper.} For an individual system with multiple 
charges the phase diagram is not necessarily convex. Interestingly, however, for a composition 
of two or more systems, the phase diagram becomes convex. Moreover, for a composition of large 
enough systems, for any point in the phase diagram, there is a state with tensor 
product structure that realises it. 
This implies that from the macroscopic point of view it is enough 
to consider states of a composite system with tensor product structure. This is an 
important feature when we study a traditional thermodynamics set-up considering only 
tensor product states, and it does not affect the generality of the laws of thermodynamics 
which only depend on the macroscopic properties of a state rather than the state itself.
We find that given the entropy and charge preserving operations as the allowed operations, 
the (generalized) 
phase diagram fully characterizes the thermodynamic transformations of the states and 
the role of thermodynamic resources in such processes. We further extend our study to 
situations where the system and bath become correlated after initially being independent. 
In such a case we use the conditional entropy instead of the entropy, to express the 
phase diagram and derive the laws of quantum thermodynamics when the final state exhibits 
\aw{possible} system-bath correlations.

\medskip
The rest of the paper is organized as follows. 
In Section \ref{sec:resource-theory}, we specify our resource theory, 
considering a quantum system $Q$ with a finite-dimensional Hilbert space, together with a
Hamiltonian $H=A_1$ and other quantities (``charges'') $A_2, \ldots, A_c$.
We introduce here the concept of \emph{phase diagram} and prove the fundamental 
\emph{Asymptotic Equivalence Theorem} \ref{Asymptotic equivalence theorem} (AET), 
which shows that the points in the phase diagram label asymptotic equivalence classes
of sequences of states. 
This allows us to study asymptotic thermodynamics of systems with multiple conserved 
quantities in Section \ref{sec:asymptotic}. 
We start by describing the system model, comprising a work system, baths and batteries, 
which permits us to formulate and prove the first law in 
Subsection \ref{subsec:model}; 
the second law is discussed in Subsection \ref{subsec:secondlaw};
in Subsection \ref{subsec:finitebath} we characterize precisely which 
work transformations are possible on a system with a given bath, 
in terms of the extended phase diagram, which features negative 
entropies corresponding to the purely quantum effect of entanglement 
between system and bath; 
in Subsection \ref{subsec:bath-rate} we introduce the thermal bath rate, 
and discuss the tradeoff between the bath rate and work extraction.
We conclude in Section \ref{sec:discussion}
with a discussion of our theory and an outlook. 
The paper also includes \aw{three} appendices: 
Appendix \ref{section: Miscellaneous definitions and facts} introduces 
technical notation and some auxiliary results; 
Appendix \ref{section: Approximate microcanonical (a.m.c.) subspace} gives an 
explicit construction of so-called approximate 
microcanonical subspaces (a.m.c.) for non-commuting observables \cite{Halpern2016};
Appendix \ref{proof-AET} provides the full proof of the 
AET Theorem \ref{Asymptotic equivalence theorem}.

\section{Resource theory of charges and entropy}
\label{sec:resource-theory}
A system in our resource theory is a quantum system $Q$ with a finite-dimensional Hilbert space
(denoted $Q$, too, without danger of confusion), together with a
Hamiltonian $H=A_1$ and other quantities (``charges'') $A_2, \ldots, A_c$, all of which are 
Hermitian operators that do not necessarily commute with each other. We consider composition of 
$n$ non-interacting systems, where the Hilbert space of the \emph{composite} system $Q^n$ is 
the tensor product $Q^{\otimes n} = Q_1 \otimes \cdots \otimes Q_n$ of the Hilbert spaces of 
the \emph{individual} systems, and the $j$-th charge of the composite system is the sum of 
charges of individual systems as follows,
\begin{equation}
  A^{(n)}_j = \sum_{i=1}^{n} \1^{\otimes (i-1)} \otimes A_j \otimes \1^{\otimes (n-i)}, 
                \quad j=1,2,\ldots,c.
\end{equation}
For ease of notation, we will write throughout
$A_j^{(Q_i)} = \1^{\otimes (i-1)} \otimes A_j \otimes \1^{\otimes (n-i)}$. We note that throughout the paper we label various subsystems or individual  systems with subscripts whereas here subscript $j$ denotes different charges of each subsystem. To avoid this confusion, note that various charges are always labeled by $j$. 

We wish to build a resource theory where the objects are states on a quantum system, 
which are transformed under thermodynamically meaningful operations.
To any quantum state $\rho$ is assigned the point 
$(\und{a},s) = (a_1,\ldots,a_c,s) 
= \bigl( \Tr \rho A_1, \ldots, \Tr \rho A_c, S(\rho) \bigr) \in \mathbb{R}^{c+1}$,
which is an element in the \emph{phase diagram} that has been originally introduced,
for $c=1$, as energy-entropy diagram in \cite{Sparaciari2016}; there it is shown,
for a system where energy is the only conserved quantity, that the diagram is a convex set.
In the case of commuting multiple conserved quantities, the charge-entropy diagram has been 
generalised and further investigated in \cite{Bera2017}. 
Note that the set of all these vectors, denoted $\mathcal{P}^{(1)}$, is not in 
general convex (unless the quantities commute pairwise). 
An example is a qubit system with charges $\sigma_x$, $\sigma_y$ and $\sigma_z$ where 
charge values uniquely determine the state as a linear function of the $\tr \rho\sigma_i$, 
hence the entropy, while the von Neumann entropy itself is well-known to be strictly concave.

Moreover, the set of these points for a composite system with charges $A_1^{(n)}, \ldots, A_c^{(n)}$, 
which we denote $\mathcal{P}^{(n)}$ contains, but is not necessarily equal to $n\mathcal{P}^{(1)}$
(which however is true for commuting charges). Namely, consider the point 
$g=\left(\frac{1}{2}\Tr (\rho_1+\rho_2) A_1, \ldots, \frac{1}{2}\Tr (\rho_1+\rho_2) A_c,
\frac{1}{2}S(\rho_1)+\frac{1}{2}S(\rho_2)\right)$, which does not necessarily 
belong to $\mathcal{P}^{(1)}$ but belongs to its convex hull; 
however, $2g \in \mathcal{P}^{(2)}$ due to the state $\rho_1 \otimes \rho_2$.
Therefore, we consider the convex hull of the set $\mathcal{P}^{(1)}$ and call it the 
\emph{phase diagram} of the system, denoted
\begin{equation}
  \overline{\mathcal{P}} 
     \equiv \overline{\mathcal{P}}^{(1)} 
     := \left\{ \left(\sum_i p_i \Tr \rho_i A_1, \ldots, \sum_i p_i\Tr \rho_i A_c, \sum_i p_i S(\rho_i)\right) : 
                0 \leq p_i \leq 1,\,\sum_i p_i = 1 \right\}.
\end{equation} 
The interpretation is that the objects of our resource theory are ensembles of 
states $\{p_i,\rho_i\}$, rather than single states. 

We define the \emph{zero-entropy diagram} and \emph{max-entropy diagram}, 
respectively, as the sets
\begin{align*}
  \overline{\mathcal{P}}_0^{(1)} 
     &= \{(\und{a},0): \Tr \rho A_j = a_j \text{ for a state } \rho \}, \\
  \overline{\mathcal{P}}_{\max}^{(1)} 
     &= \left\{\bigl(\und{a},S(\tau(\und{a}))\bigr): \Tr \rho A_j = a_j \text{ for a state } \rho \right\}, 
\end{align*}
where $\tau(\und{a})$ is the unique state maximising the entropy among all states 
with charge values $\Tr \rho A_j = a_j$ for all $j$, which is called generalized 
thermal state, or generalized Gibbs state, or also generalized grand canonical state \cite{Liu2007}. 
Note that, as a linear image of the compact convex set of states, the zero-entropy diagram is 
compact and convex.
We similarly define the set $\mathcal{P}^{(n)}$, the phase diagram $\overline{\mathcal{P}}^{(n)}$, 
zero-entropy diagram $\overline{\mathcal{P}}_0^{(n)}$ and max-entropy diagram 
$\overline{\mathcal{P}}_{\max}^{(n)}$ 
for the composition of $n$ systems with charges $A_1^{(n)}, \ldots, A_c^{(n)}$. 
Fig.~\ref{fig:phase-diagrams} illustrates these concepts. 

\begin{figure}[ht]
  \includegraphics[width=15cm,height=6cm]{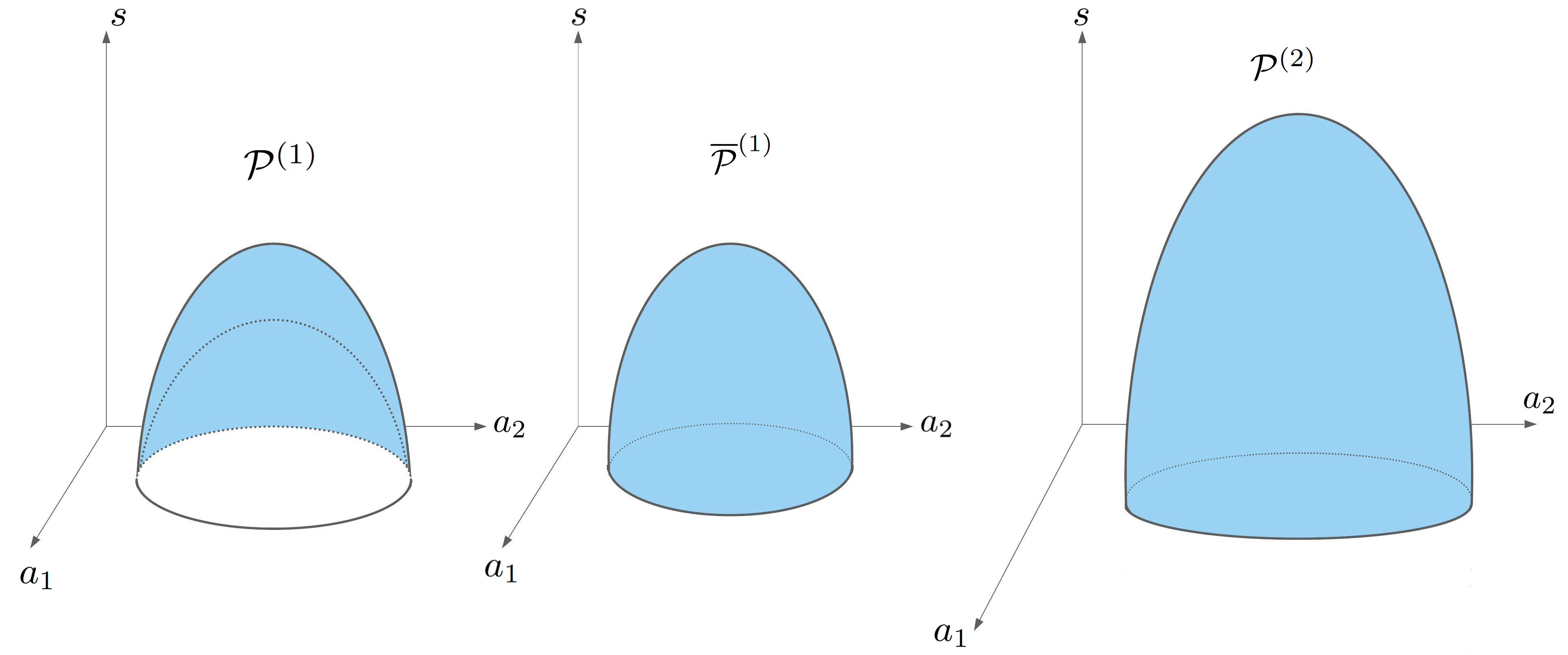}
  \caption{Schematic of the phase diagrams $\mathcal{P}^{(1)}$, $\mathcal{P}^{(2)}$
           and $\overline{\mathcal{P}}$. As seen, $\mathcal{P}^{(1)}$ is not convex, 
           having a hollow on the underside.}
  \label{fig:phase-diagrams}
\end{figure}

\begin{lemma}
\label{lemma:phase diagram properties}
For an individual system $Q$ and composite system $Q^{\otimes n}$  with charges $A_j$ and $A^{(n)}_j$, respectively, 
the following holds: 
\begin{enumerate}
  \item $\overline{\mathcal{P}}^{(n)}$, for $n \geq 1$, is a compact and convex 
        subset of $\mathbb{R}^{c+1}$.

  \item $\overline{\mathcal{P}}^{(n)}$, for $n \geq 1$, is the convex hull of the union 
        $\overline{\mathcal{P}}_{0}^{(n)} \cup \overline{\mathcal{P}}_{\max}^{(n)}$, 
        of the zero-entropy diagram and the max-entropy diagram.
     
    
  \item $\overline{\mathcal{P}}^{(n)} = n \overline{\mathcal{P}}^{(1)}$ for all $n \geq 1$. 
    
  \item $\mathcal{P}^{(n)}$ is convex for all $n \geq 2$, and indeed 
        $\mathcal{P}^{(n)} = \overline{\mathcal{P}}^{(n)} = n \overline{\mathcal{P}}^{(1)}$. 
    
    
    
  \item Every point of $\mathcal{P}^{(n)}$ is realised by a suitable tensor product state 
        $\rho_1 \otimes \cdots \otimes \rho_n$, for all $n \geq |Q|$ where $|Q|$ is the dimension of system $Q$.

  \item All points $\bigl(\und{a},S(\tau(\und{a}))\bigr) \in \overline{\mathcal{P}}_{\max}$
        are extreme points of $\overline{\mathcal{P}}$.
\end{enumerate}
\end{lemma}

\begin{proof}
1. The phase diagram is convex by definition. Further, $\Tr \rho A_j^{(n)}$ and $S(\rho)$ 
are continuous functions defined on the set of quantum states which is a compact set; 
hence, the set $\mathcal{P}^{(n)}$ is also a compact set. The cxonvex hull of a 
finite-dimensional compact set is compact, so the phase diagram is a compact set.  

2. Any point in the phase diagram according to the definition is a convex combination of the form 
\[(a_1,\ldots,a_c,s)
= \left(\sum_i p_i \Tr(\rho_i A_1), \ldots, \sum_i p_i\Tr(\rho_i A_c),\sum_i p_i S(\rho_i)\right).\]
The point $(a_1,\ldots,a_c,0)$ belongs to $\overline{\mathcal{P}}_{0}^{(1)}$ 
because the state $\rho=\sum_i p_i \rho_i$ has charge values $a_1, \ldots, a_c$. 
Moreover, the state with charge values $a_1, \ldots, a_c$ of maximum entropy is 
the generalized thermal state $\tau(\und{a})$, so we have 
\begin{align*}
    S(\tau(\und{a})) \geq S(\rho) \geq \sum_i p_i S(\rho_i),
\end{align*}
where the second inequality is due to concavity of the entropy. Therefore, any point
$(\und{a},s)$ can be written as the convex combination of the points $(\und{a},0)$ and 
$\bigl(\und{a},S(\tau(\und{a}))\bigr)$.

3. Due to item 2, it is enough to show that 
$\overline{\mathcal{P}}_{0}^{(n)} = n\overline{\mathcal{P}}_{0}^{(1)}$, and 
$\overline{\mathcal{P}}_{\max}^{(n)}= n\overline{\mathcal{P}}_{\max}^{(1)}$. 
The former follows from the definition. The latter is due to the fact that the 
thermal state for a composite system is the tensor power of the thermal state of 
the individual system. 

4. Let $\tau(\und{a}) = \sum_i p_i \ketbra{i}{i}$ be the diagonalization of the 
generalized thermal state.
For $n \geq 2$, define $\ket{v} = \sum_i \sqrt{p_i} \ket{i}^{\ox n}$. Obviously, the charge 
values of the states $\tau(\und{a})^{\ox n}$ and $\ketbra{v}{v}$ are the same, since they
have the same reduced states on the individual systems;
thus, there is a pure state for any point in the zero-entropy diagram of the composite system.
Now, consider the state $\lambda \ketbra{v}{v}+(1-\lambda)\tau(\und{a})^{\ox n}$, which has the 
same charge values as $\tau(\und{a})^{\ox n}$ and $\ketbra{v}{v}$. 
The entropy $S\bigl(\lambda \ketbra{v}{v}+(1-\lambda)\tau(\und{a})^{\ox n}\bigr)$ is a continuous function 
of $\lambda$; hence, for any value $s$ between $0$ and $S(\tau(\und{a})^{\ox n})$, there is a state 
with the given values and entropy $s$. 

5. For $n\geq |Q|$, it is easy to see that any state $\rho$ can be decomposed 
into a uniform convex combination of $n$ pure states, i.e. 
$\rho=\frac{1}{n} \sum_{i=1}^n \ketbra{\psi_i}{\psi_i}$. 
\aw{For instance, consider the diagonalization of $\rho = \sum_{t=1}^{|Q|} q_t \proj{t}$,
and define $\ket{\psi_\ell} := \sum_t \sqrt{q_t} e^{2\pi i t\ell/n}\ket{t}$. Due to 
the cyclotomic properties of the primitive $n$-root of unit, this satisfies the claim.}
Now observe that the state $\psi^{n} = \ketbra{\psi_1} \ox \cdots \ox \ketbra{\psi_n}$ has the same 
charge values as the state $\rho^{\otimes n}$, but as it is pure it has entropy $0$.
Further, consider the thermal state $\tau$ with the same charge values as $\rho$, but 
the maximum entropy consistent with them. Now let
$\rho_i := \lambda \ketbra{\psi_i}{\psi_i}+(1-\lambda)\tau$, and 
observe that $\rho_\lambda^{n} = \rho_1 \ox \cdots \rho_n$ has the same charge values as 
$\psi^{n}$, $\rho^{n}$ and $\tau^{\ox n}$.
Since the entropy $S(\rho_\lambda^{n})$ is a continuous function of $\lambda$, 
thus interpolating smoothly between $0$ and $n S(\tau)$, there is a 
tensor product state with the same given charge values and prescribed 
entropy $s$ in the said interval. 
%

6. This follows from the strict concavity of the von Neumann entropy $S(\rho)$ as a
function of the state, which imparts the strict concavity on $\und{a} \mapsto S(\tau(\und{a}))$.
\end{proof}

\medskip
The penultimate point of Lemma \ref{lemma:phase diagram properties} motivates us to define 
a resource theory where the objects are sequences of states on composite systems
of $n\rightarrow\infty$ parts. 
Inspired by \cite{Sparaciari2016}, the allowed operations in this resource theory 
are those that respect basic principles of physics, namely entropy and charge 
conservation. We point out right here, that ``physics'' in the present context 
does not necessarily refer to the fundamental physical laws of nature, but to 
any rule that the system under consideration obeys. 
It is well-known that quantum operations that preserve entropy for all states are 
unitaries. The class of unitaries that conserve charges of a system are precisely 
those that commute with all charges of that system. 
However, it turns out that these constraints are too strong if imposed literally, 
when many charges are to be conserved, as it could easily happen that only
trivial unitaries are allowed.
Our way out is to consider the thermodynamic limit and at the same time relax the 
allowed operations to approximately entropy and charge conserving ones.    
As for the former, we couple the composite system to an ancillary system with corresponding 
Hilbert space $\mathcal{K}$ of dimension $2^{o(n)}$, where restricting the dimension of 
the ancilla ensures that the entropy rate per system, that is the entropy of the composite 
system divided by $n$, does not change in the limit of large $n$. 
Here and in the following, we use standard little-oh notation, by which $o(n)$
denotes a function such that $\lim_{n\to\infty} \frac{o(n)}{n} = 0$.
Moreover, as for charge conservation, we consider unitaries that are \emph{almost} commuting 
with the total charges of the composite system and the ancilla. The precise definition 
is as follows. 

\begin{definition}
\label{almost-commuting unitaries}
A unitary operation $U$ acting on a composite system coupled to an ancillary system with 
Hilbert spaces $\mathcal{H}^{\otimes n}$ and $\mathcal{K}$ of dimension $2^{o(n)}$, respectively,
is called an \emph{almost-commuting unitary} with the total charges of a composite system and an 
ancillary system if the operator norm of the normalised commutator for all total charges 
vanishes asymptotically for large $n$:
\begin{align*}
  \lim_{n \to \infty} \frac{1}{n} \norm{ [U,A_j^{(n)}+A_j']}_{\infty}
     =\lim_{n \to \infty} \frac{1}{n}\norm{ U (A_j^{(n)}+A_j')-(A_j^{(n)}+A_j') U }_{\infty} 
     = 0 
  \qquad j=1,\ldots,c.
\end{align*}
where $A_j^{(n)}$ and $A_j'$ are respectively the charges of the composite system 
and the ancilla, such that $\norm{A_j'}_{\infty} \leq o(n)$.  
\end{definition}

We stress that the definition of almost-commuting unitaries automatically implies that 
the ancillary system has a relatively small dimension and charges with small operator 
norm compared to a composite system.    
The first step in the development of our resource theory is a precise characterisation 
of which transformations between sequences of product state are possible using almost 
commuting unitaries.
To do so, we define \emph{asymptotically equivalent} states as follows:

\begin{definition}
\label{Asymptotic equivalence definition}
Two sequences of product states 
\aw{$\rho^n = \rho_{Q^n} = (\rho_1)_{Q_1} \otimes \cdots \otimes (\rho_n)_{Q_n}$ and 
$\sigma^n = \sigma_{Q^n} = (\sigma_1)_{Q_1} \otimes \cdots \otimes (\sigma_n)_{Q_n}$ }
of a composite system with 
charges $A_j^{(n)}$ for $j=1,\ldots,c$, are called \emph{asymptotically equivalent} if  
\begin{align*}
  \lim_{n \to \infty} \frac{1}{n} \abs{S(\rho^n) - S(\sigma^n)}                       &= 0, \\ 
  \lim_{n \to \infty} \frac{1}{n} \abs{\Tr \rho^n A_j^{(n)} - \Tr \sigma^n A_j^{(n)}} &= 0 \text{ for all } j=1,\ldots,c.
\end{align*}
In other words, two sequences of product states are considered equivalent if their
associated points in the normalised phase diagrams $\frac1n \mathcal{P}^{(n)}$ differ
by a sequence converging to $0$. 
\end{definition}
\aw{We note that in the above definition, $\rho^n$ and $\sigma^n$ are tensor 
products of possibly different states; a tensor power state is denoted $\rho^{\otimes n}$}.

The \emph{asymptotic equivalence theorem} of \cite{Sparaciari2016} characterizes 
feasible state transformations via \emph{exactly} commuting unitaries where energy 
is the only conserved quantity of a system, showing that it is precisely 
given by asymptotic equivalence. 
We prove an extension of this theorem for systems with multiple, possibly non-commuting
conserved quantities, by allowing almost-commuting unitaries. 

\begin{theorem}[Asymptotic (approximate) Equivalence Theorem -- AET]
\label{Asymptotic equivalence theorem} 
Let $\rho^n=\rho_1 \otimes \cdots \otimes \rho_n$ and 
$\sigma^n=\sigma_1 \otimes \cdots \otimes \sigma_n$ 
be two sequences of product states of a composite system with 
charges $A_j^{(n)}$ for $j=1,\ldots,c$.
These two states are asymptotically equivalent if and only if 
there exist ancillary quantum systems with corresponding Hilbert space $\mathcal{K}$ 
of dimension $2^{o(n)}$ and an almost-commuting unitary $U$ acting on 
$\mathcal{H}^{\otimes n} \otimes \mathcal{K}$ such that 
\begin{align*}
  \lim_{n \to \infty} \norm{U (\rho^n \otimes \omega') U^{\dagger} - \sigma^n \otimes \omega}_1 &= 0, \\
\end{align*}
where $\omega$ and $\omega'$  are states of the ancillary system, and charges of the ancillary 
system are trivial, $A_j' = 0$.
\end{theorem}

\medskip
The \emph{proof} of this theorem is given in Appendix \ref{proof-AET}, as it relies on a
number of technical lemmas, among them the concept of an 
\emph{approximate microcanonical subspace (a.m.c.)} \cite{Halpern2016}, of which 
we give a novel construction in Appendix \ref{section: Approximate microcanonical (a.m.c.) subspace}. 

\medskip

\section{Asymptotic thermodynamics of multiple conserved quantities}
\label{sec:asymptotic} 
As a thermodynamic theory, or even as a resource theory in general, transformations
by almost-commuting unitaries do not appear to be the most fruitful: they are reversible
and induce an equivalence relation among the sequences of product states. 
In particular, every point $(\und{a},s)$ of the phase diagram $\overline{\mathcal{P}}^{(1)}$
defines an equivalence class, namely of all state sequences with charges and 
entropy converging to $\und{a}$ and $s$, respectively. 

To make the theory more interesting, and more resembling of ordinary thermodynamics, 
as expressed in its first and second laws (including irreversibility),
we now specialise to a setting considered in many previous papers in the resource theory 
of thermodynamics, both with single or multiple conserved quantities. 
Specifically, we consider an asymptotic analogue of the setting proposed 
in \cite{Guryanova2016} concerning the interaction of thermal baths with a 
quantum system and batteries, where it was shown that the second law constrains
the combination of extractable charge quantities. 
In \cite{Guryanova2016}, explicit protocols for state transformations to saturate 
the second law are presented, that store each of several commuting charges in its 
corresponding explicit battery. However, for the case of non-commuting charges, one battery, 
or a so-called reference frame, stores all different types of charges \cite{Halpern2016,Popescu2018}.
Only recently it was shown that reference frames for non-commuting charges
can be constructed, at least under certain conditions, which store the different 
charge types in physically separated subsystems \cite{Popescu2019}.
Moreover, the size of the bath required to perform the transformations is not 
addressed in these works, as only the limit of asymptotically large bath was
considered. 
We will address these questions in a similar setting but in the asymptotic regime, 
where Theorem~\ref{Asymptotic equivalence theorem} provides the necessary and sufficient
condition for physically possible state transformations.
In this new setting, the \emph{asymptotic} second law constrains the combination of 
extractable charges; we provide explicit protocols for realising transformations 
satisfying the second law, where each explicit battery can store its corresponding type of
work in the general case of non-commuting charges. Furthermore, we determine the 
minimum number of thermal baths of a given type that is required to perform a transformation.

\subsection{System model, batteries and the first law}
\label{subsec:model}
We consider a system being in contact with a bath and suitable batteries,
with a total Hilbert space 
$Q=S\otimes B\otimes W_1\otimes \cdots \otimes W_c$,
consisting of many non-interacting subsystems; namely, the work system, the thermal bath and 
$c$ battery systems with Hilbert spaces ${S}$, ${B}$ and ${W}_j$ for 
$j=1,\ldots,c$, respectively.  
We call the $j$-th battery system the $j$-type battery as it is designed to absorb
$j$-type work.
The work system and the thermal bath have respectively the charges $A_{S_j}$ and $A_{B_j}$ 
for all $j$, but $j$-type battery has only one nontrivial charge $A_{W_j}$, and all 
its other charges are zero because it is meant to store only the $j$-th charge. 
We note that $S$, $B$ and $W_j$s are different 
Hilbert spaces, so the charges of their corresponding systems can be different as well.
The total charge is the sum of the charges of the sub-systems $A_j=A_{S_j}+A_{B_j}+A_{W_j}$ 
for all $j$. Furthermore, for a charge $A$, let $\Sigma(A)=\lambda_{\max}(A)-\lambda_{\min}(A)$ 
denote the spectral diameter, where $\lambda_{\max}(A)$ and $\lambda_{\min}(A)$ are 
the largest and smallest eigenvalues of the charge $A$, respectively. 
We assume that the total spectral 
diameter of the work system and the thermal bath is bounded by the spectral diameter of the 
battery, that is $\Sigma(A_{S_j})+\Sigma(A_{B_j}) \leq \Sigma(A_{W_j})$ for all $j$; this 
assumption ensures that the batteries can absorb or release charges for transformations.   

As we discussed in the previous section, the generalized thermal state $\tau(\und{a})$
is the state that maximizes the entropy subject to the
constraint that the charges $A_j$ have the values $a_j$. 
This state is equal to $\frac{1}{Z}e^{-\sum_{j=1}^c \beta_j A_{j}}$ for real
numbers $\beta_j$ called inverse temperatures and chemical potentials; 
each of them is a smooth function of charge values $a_1,\ldots,a_c$, and 
$Z=\Tr e^{-\sum_{j=1}^c \beta_j A_{j}}$ is the generalized partition function. 
Therefore, the generalized thermal state can be equivalently denoted $\tau(\und{\beta})$ 
as a function of the inverse temperatures, associated uniquely with the charge
values $\und{a}$. 
We assume that the thermal bath is initially in a generalized thermal state
$\tau_b(\und{\beta})$, for globally fixed $\und{\beta}$. 
This is because in \cite{Halpern2016} it was argued that these are precisely the
\emph{completely passive} states, from which no energy can be extracted into 
a battery storing energy, while not changing any of the other conserved quantities,  
by means of almost-commuting unitaries and even when unlimited copies of the state are available. 
We assume that the work system with state $\rho_s$ and the thermal bath are initially 
uncorrelated, and furthermore that the battery systems can acquire only pure states.  

Therefore, the initial state of an \emph{individual} global system $Q$ 
is assumed to be of the following form,
\begin{equation}
  \label{eq:initial composite global}
  \rho_{SBW_1\ldots W_c} 
     = \rho_S \ox \tau(\und{\beta})_B \ox \proj{w_1}_{W_1} \ox \cdots \ox \proj{w_c}_{W_c},
\end{equation}
and the final states we consider are of the form 
\begin{equation}
  \label{eq:final composite global}
  \sigma_{SBW_1\ldots W_c} 
     = \sigma_{SB} \ox \proj{w_1'}_{W_1} \ox \cdots \ox \proj{w_c'}_{W_c},
\end{equation}
where $\rho_S$ and $\sigma_{SB}$ are states of the system and system-plus-bath, respectively, 
and $w_j$ and $w_j'$ label pure states of the $j$-type battery before and after the transformation. 
The notation is meant to convey the expectation value of the $j$-type work, i.e. $w_j^{(\prime)}$
is a real number and $\Tr \proj{w_j^{(\prime)}}A_{W_j} = w_j^{(\prime)}$.

The established resource theory of thermodynamics treats the batteries and the bath as 
`enablers' of transformations of the system $S$, and we will show first and second laws 
that express the essential constraints that any such transformation has to obey. 
We start with the batteries. With the notations $\und{W} = W_1\ldots W_c$, 
$\ket{\und{w}} = \ket{w_1}\cdots\ket{w_c}$, and $\ket{\und{w}'} = \ket{w_1'}\cdots\ket{w_c'}$,
let us look at a sequence $\rho^n = \rho_{S^n} = \rho_{S_1} \ox\cdots\ox \rho_{S_n}$ of initial
system states, and a sequence $\proj{\und{w}}^n = \proj{\und{w}_1}_{\und{W}_1} \ox\cdots\ox \proj{\und{w}_n}_{\und{W}_n}$
of initial battery states, recalling that the baths are initially all in the same thermal
state, $\tau_{B^n} = \tau(\und{\beta})^{\ox n}$; furthermore a sequence of target states 
$\sigma^n = \sigma_{S^nB^n} = \sigma_{S_1B_1} \ox\cdots\ox \sigma_{S_nB_n}$ of the system and bath, and a 
sequence $\proj{\und{w}'}^n = \proj{\und{w}_1'}_{\und{W}_1} \ox\cdots\ox \proj{\und{w}_n'}_{\und{W}_n}$
of target states of the batteries. 

\begin{definition}
  \label{definition:regular}
  A sequence of states $\rho^n$ on any system $Q^n$ is called \emph{regular} if 
  its charge and entropy rates converge, i.e. if
  \begin{align*}
    a_j &= \lim_{n\rightarrow\infty} \frac1n \Tr \rho^n A_j^{(n)},\ j=1,\ldots,c, \text{ and} \\
    s   &= \lim_{n\rightarrow\infty} \frac1n S(\rho^n) 
  \end{align*}
  exist. To indicate the dependence on the state sequence, 
  we write $a_j(\{\rho^n\})$ and $s(\{\rho^n\})$.
\end{definition}
In the rest of the chapter we will essentially focus on regular sequences, so that
we can simply identify them, up to asymptotic equivalence, with a point in the phase 
diagram. However, it should be noted that at the expense of clumsier expressions, 
most of our expositions can be extended to arbitrary sequences of product states or 
block-product states. 

According to the AET and the other results of the previous section, every point $(\und{a},s)$
in the phase diagram $\overline{\mathcal{P}}^{(1)}$ labels an equivalence class of 
regular sequences of product states under transformations by almost-commuting unitaries. 

\medskip

We emphasize that in AET by grouping the $Q$-systems into blocks of $k$, we do not of course change the 
physics of our system, except that now in the asymptotic limit we only consider
$n = k\nu$ copies of $Q$, but the state $\rho^n$ is asymptotically equivalent to
$\rho^{n+O(1)}$ via almost-commuting unitaries according to Definition \ref{almost-commuting unitaries} 
and Theorem \ref{Asymptotic equivalence theorem}. 
But now we consider $Q^k$ with its charge observables $A_j^{(k)}$ as elementary 
systems, which have many more states than the $k$-fold product states we began with. 
Yet, Lemma \ref{lemma:phase diagram properties} shows that the phase diagram for 
the $k$-copy system is simply the rescaled single-copy phase diagram,
$\overline{\mathcal{P}}^{(k)} = k \overline{\mathcal{P}}^{(1)}$, and indeed 
for $k\geq d$, $\mathcal{P}^{(k)} = k \overline{\mathcal{P}}^{(1)}$. This means 
that we can extend the relation of asymptotic equivalence and the
concomitant Asymptotic Equivalence Theorem (AET) \ref{Asymptotic equivalence theorem}
to any sequences of states that factor into product states of blocks $Q^k$, for 
any integer $k$, which freedom we exploit in this thermodynamics setup.

\medskip
Now, for regular sequences $\rho_{S^n}$ of initial states of the system and 
final states of the system plus bath, $\sigma_{S^nB^n}$, as well as regular
sequences of initial and final battery states, $\proj{\und{w}}^n$ and $\proj{\und{w}'}^n$, 
respectively, define the asymptotic rate of $j$-th charge change of the $j$-type battery as 
\begin{equation}
  \label{eq: W_j definition}
  \Delta A_{W_j} := a_j(\{\proj{w_j'}^n\})-a_j(\{\proj{w_j}^n\})
                  = \lim_{n\rightarrow\infty} \frac1n  \Tr (\proj{w_j'}^n-\proj{w_j}^n)A_{W_j}^{(n)}.
\end{equation}
Where there is no danger of confusion, we denote this number also as $W_j$, 
the $j$-type work extracted (if $W_j < 0$, this means that the work $-W_j$ is done 
on system $S$ and bath $B$). 

Similarly, we define the asymptotic rate of $j$-th charge change of the work system
and the bath as 
\begin{align*}
  \Delta A_{S_j} &:= a_j(\{\sigma_{S^n}\})-a_j(\{\rho_{S^n}\})
                   = \lim_{n\rightarrow\infty} \frac1n \Tr (\sigma_{S^n}-\rho_{S^n})A_{S_j}^{(n)}, \\ 
  \Delta A_{B_j} &:= a_j(\{\sigma_{B^n}\})-a_j(\{\tau(\und{\beta})_{B^n}\})
                   = \lim_{n\rightarrow\infty} \frac1n \Tr (\sigma_{B^n}-\tau(\und{\beta})_B^{\ox n})A_{B_j}^{(n)}, 
\end{align*}
where we denote $\sigma_{S^n} = \tr_{B^n} \sigma_{S^nB^n}$ and likewise 
$\sigma_{B^n} = \tr_{S^n} \sigma_{S^nB^n}$.

\begin{theorem}[First Law]
\label{thm:first-law}
Under the above notations, if the regular sequences of the initial state
$\rho_{S^nB^n\und{W}^n} = \rho_{S^n} \ox \tau(\und{\beta})_B^{\ox n} \ox \proj{\und{w}}^n$ 
and the final state $\sigma_{S^nB^n\und{W}^n} = \sigma_{S^nB^n} \ox \proj{\und{w}'}^n$
are equivalent under almost-commuting unitaries, then 
\begin{align*}
  s(\{\sigma_{S^nB^n}\}) &= s(\{\rho_{S^n}\}) + S(\tau(\und{\beta})) \text{ and} \\ 
  W_j                    &= -\Delta A_{S_j}-\Delta A_{B_j} \text{ for all } j=1,\ldots,c. 
\end{align*}

Conversely, given regular sequences $\rho_{S^n}$ and $\sigma_{S^nB^n}$ of product
states such that 
\[
  s(\{\sigma_{S^nB^n}\}) = s(\{\rho_{S^n}\}) + S(\tau(\und{\beta})), 
\]
and assuming that the spectral radius of the battery observables $W_{A_j}$ is large enough 
(see the discussion at the start of this section), then there exist regular sequences of
product states of the $j$-type battery, $\proj{w_j}^n$ and $\proj{w_j'}^n$, for all
$j=1,\ldots,c$, such that 
\begin{align}
  \label{eq:initial}
  \rho_{S^nB^n\und{W}^n}   &= \rho_{S^n} \ox \tau(\und{\beta})_B^{\ox n} \ox \proj{\und{w}}^n \text{ and} \\
  \label{eq:final}
  \sigma_{S^nB^n\und{W}^n} &= \sigma_{S^nB^n} \ox \proj{\und{w}'}^n 
\end{align} 
can be transformed into each other by almost-commuting unitaries.
\end{theorem}

\begin{proof}
The first part is by definition, since the almost-commuting unitaries 
asymptotically preserve the entropy rate and the work rate of all 
charges. 

In the other direction, all we have to do is find states $\proj{w_j}$ 
and $\proj{w_j'}$ of the $j$-type battery $W_j$, such that
$W_j = \Delta A_{W_j} = -\Delta A_{S_j}-\Delta A_{B_j}$, for all $j=1,\ldots,c$. 
This is clearly possible if the spectral radius of $W_{A_j}$ is large enough. 
With this, the states in Eqs. (\ref{eq:initial}) and (\ref{eq:final}) have the 
same asymptotic entropy and charge rates. 
Hence, the claim follows from the AET, Theorem~\ref{Asymptotic equivalence theorem}.
\end{proof}

\begin{remark}\normalfont
The second part of Theorem~\ref{thm:first-law} says that for regular product state 
sequences, as long as the initial and final states of the work system and the thermal bath 
have asymptotically the same entropy, they can be transformed one into the another 
because there are always batteries that can absorb or release the necessary charge difference. 
Furthermore, we can even fix the initial (or final) state of the batteries and 
design the matching final (initial) battery state, assuming that the charge 
expectation value of the initial (final) state is far enough from the edge of the
spectrum of $A_{W_j}$.
\end{remark}

For any such states, we say that there is a \emph{work transformation} 
taking one to the other, denoted 
$\rho_{S^n} \ox \tau(\und{\beta})_B^{\ox n} \rightarrow \sigma_{S^nB^n}$.
This transformation is always feasible, implicitly assuming the 
presence of suitable batteries for all $j$-type works to balance to books
explicitly. 

\begin{remark}\normalfont
As a consequence of the previous remark, we now change our point of view of what a 
transformation is. Of our complicated $S$-$B$-$\und{W}$ compound, we only 
focus on $SB$ and its state, and treat the batteries as implicit. Since we insist
that batteries need to remain in a pure state, which thus factors off and 
does not contribute to the entropy, and due to the above first law Theorem \ref{thm:first-law}, 
we can indeed understand everything that is going on by looking at how 
$\rho_{S^nB^n}$ transforms into $\sigma_{S^nB^n}$. 
\end{remark}

Note that in this context, it is in a certain sense enough that the initial states 
$\rho_{S^n}$ form a regular sequence of product states and that the target states 
$\sigma_{S^nB^n}$ form a regular sequence. This is because the first part of 
the first law, Theorem \ref{thm:first-law}, only requires regularity, and 
since the target state defines a unique point $(\und{a}',s')$ in the phase 
diagram, we can find a sequence of product states $\widetilde{\sigma}_{S^nB^n}$ in 
its equivalence class, and use the second part of Theorem \ref{thm:first-law}
to realise the work transformation 
$\rho_{S^n} \ox \tau(\und{\beta})_B^{\ox n} \rightarrow \widetilde{\sigma}_{S^nB^n}$.

\subsection{The second law}
\label{subsec:secondlaw}
If the first law in our framework arises from focusing on the system-plus-bath 
compound $SB$, while making the batteries implicit, the second law comes about 
from trying to understand the action on the work system $S$ alone, through the
concomitant back-action on the bath $B$. 
Following \cite{Guryanova2016,Halpern2016}, the second law constrains the different 
combinations of commuting conserved quantities that can be extracted from the work 
system. We show here that in the asymptotic regime, the second law similarly bounds 
the extractable work rate via the rate of free entropy of the system. 

The \emph{free entropy} for a system with state $\rho$, charges $A_j$ and inverse temperatures 
$\beta_j$ is defined in \cite{Guryanova2016} as
\begin{align}
  \label{free entropy}
  \widetilde{F}(\rho) = \sum_{j=1}^c \beta_j \Tr \rho A_j - S(\rho).
\end{align}
It is shown in \cite{Guryanova2016} that the generalized thermal state 
$\tau(\und{\beta})$ is the state that minimizes the free entropy for 
fixed $\beta_j$.

For any work transformation 
$\rho_{S^n} \ox \tau(\und{\beta})_B^{\ox n} \rightarrow \sigma_{S^nB^n}$
between regular sequences of states, 
we define the asymptotic rate of free entropy change for the work system and the 
thermal bath respectively as follows:
\begin{equation}\begin{split}
  \label{eq:free entropy rates}
  \Delta\widetilde{F}_S 
    &:= \lim_{n \to \infty} \frac{1}{n}\left(\widetilde{F}(\sigma_{S^n})
                                             -\widetilde{F}(\rho_{S^n}) \right), \\
  \Delta\widetilde{F}_B
    &:= \lim_{n \to \infty} \frac{1}{n}\left(\widetilde{F}(\sigma_{B^n})
                                             -n \widetilde{F}(\tau_B)\right),
\end{split}\end{equation}
where the free entropy is with respect to the charges of the work system and the thermal 
bath with fixed inverse temperatures $\beta_j$.

\begin{figure}[ht]
\begin{center}
  \includegraphics[width=10cm,height=8cm]{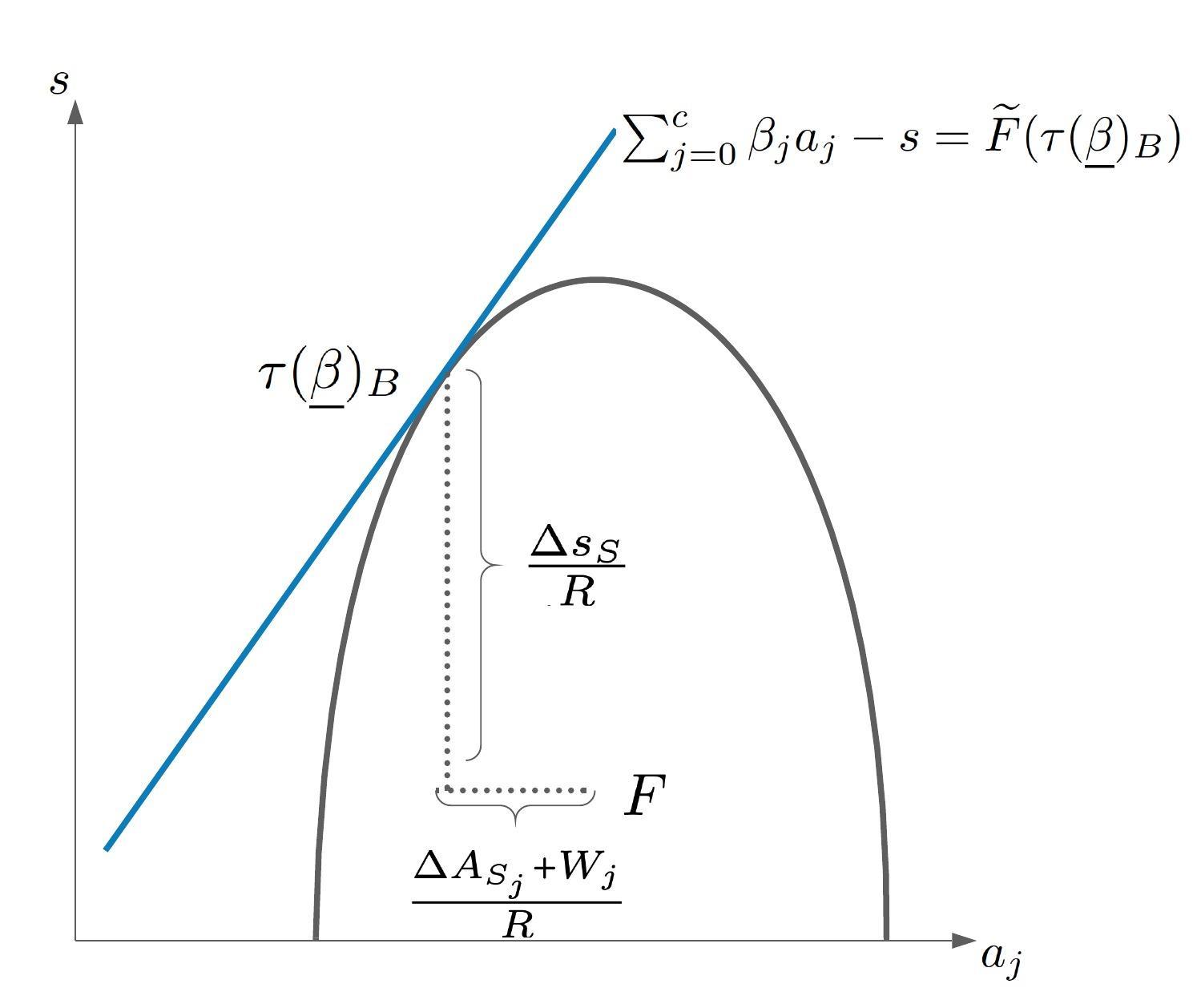}
 \end{center}
  \caption{State change of the bath for a given work transformation under extraction of 
           $j$-type work $W_j$, viewed in the phase diagram of the bath $\overline{\cP}_B$. 
           The blue line represents the tangent hyperplane at the corresponding point 
           of the generalized thermal state $\tau(\und{\beta})_B$, $R$ is the number of copies 
           of the elementary baths in the proof of Theorem \ref{asymptotic second law}, 
           and $F$ is the point corresponding to the final state of the bath.}
  \label{fig:second-law}
\end{figure}

\begin{theorem}[Second Law]
\label{asymptotic second law}
For any work transformation 
$\rho_{S^n} \ox \tau(\und{\beta})_B^{\ox n} \rightarrow \sigma_{S^nB^n}$
between regular sequences of states, the $j$-type works $W_j$ that are extracted 
(and they are necessarily $W_j = -\Delta A_{S_j}-\Delta A_{B_j}$ according 
to the first law) are constrained by the rate of free entropy change of the system:
\[
   \sum_{j=1}^c \beta_j W_j  \leq -\Delta\widetilde{F}_S.
\]

Conversely, for arbitrary regular sequences of product states, 
$\rho_{S^n}$ and $\sigma_{S^n}$, and any real numbers $W_j$ with  
$\sum_{j=1}^c \beta_j W_j < -\Delta\widetilde{F}_S$, 
there exists a bath system $B$ and a regular sequence of product states 
$\sigma_{S^nB^n}$ with $\Tr_{B^n}\sigma_{S^nB^n} = \sigma_{S^n}$, such that 
there is a work transformation 
$\rho_{S^n} \ox \tau(\und{\beta})_B^{\ox n} \rightarrow \sigma_{S^nB^n}$
with accompanying extraction of $j$-type work at rate $W_j$. 
This is illustrated in Fig.~\ref{fig:second-law}.
\end{theorem}

\begin{proof}
We start with the first statement of the theorem. Consider the global system transformation 
$\rho_{S^n} \ox \tau(\und{\beta})_B^{\ox n} \rightarrow \sigma_{S^nB^n}$ by 
almost-commuting unitaries. 
We use the definition of work (\ref{eq: W_j definition}) and free 
entropy (\ref{free entropy}), as well as the first law, Theorem \ref{thm:first-law}, 
to get
\begin{equation}
\label{work expansion formula}
\begin{split}
    \sum_j \beta_j W_j &= -\sum_j \beta_j(\Delta A_{S_j}+\Delta A_{B_j})\\
                       &= -\Delta\widetilde{F}_S-\Delta\widetilde{F}_B-\Delta s_S -\Delta s_B .
\end{split}
\end{equation}
The second line is due to the definition in Eq. (\ref{eq:free entropy rates}). 
Now observe that 
\begin{align}\label{eq: positive Delta_SB}
  \Delta s_S +\Delta s_B 
     &=     \lim_{n \to \infty} \frac1n \bigl(S(\sigma_{S^n})-S(\rho_{S^n})\bigr) 
                                + \frac1n \bigl(S(\sigma_{B^n})-nS(\tau(\und{\beta})_B)\bigr) \nonumber\\
     &\geq  \lim_{n \to \infty} \frac1n \bigl(S(\sigma_{{SB}^n})-S(\rho_{S^n})-S(\tau(\und{\beta})_B^{\ox n})\bigr)
            = 0, 
\end{align}
where the inequality is due to sub-additivity of von Neumann entropy, and the 
final equality is due to asymptotic entropy conservation. 
Further, the generalized thermal state $\tau(\und{\beta})_B$ has 
the minimum free entropy \cite{Guryanova2016}, hence 
\begin{align}\label{eq:positive Delta_F_B}
\Delta\widetilde{F}_B \geq 0.
\end{align}

For the second statement of the theorem, the achievability part of the second law, 
we aim to show that there is a work transformation 
$\rho_{S^n} \ox \tau(\und{\beta})_B^{\ox n} \rightarrow \sigma_{S^n} \ox \xi_{B^n}$, 
with a suitable regular sequences of product states, 
and works $W_1,\ldots,W_c$ are extracted. 
This will be guaranteed, by the first law  (Theorem \ref{thm:first-law}) 
and the AET, Theorem \ref{Asymptotic equivalence theorem}, if 
\begin{equation}\begin{split}
  s(\{\xi_{B^n}\})   &= S(\tau(\und{\beta})_B)  - \Delta s_S, \\
  a_j(\{\xi_{B^n}\}) &= \Tr \tau(\und{\beta})_B A_{B_j} - \Delta A_{S_j} - W_j \quad \text{for all } j=1,\ldots,c.
  \label{eq:state assumptions}
\end{split}\end{equation}
The left hand side here defines a point $(\und{a},s)$ in the charges-entropy 
space of the bath, and our task is to show that it lies in the phase diagram,
for which purpose we have to define the bath characteristics suitably. 
On the right hand side, 
$\bigl(\Tr\tau(\und{\beta})_B A_{B_1}, \ldots, \Tr\tau(\und{\beta})_B A_{B_c}, S(\tau(\und{\beta})_B)\bigr)$
is the point corresponding to the initial state of the bath, which due to 
its thermal nature is situated on the upper boundary of the region. 
At that point, the region has a unique tangent hyperplane, which has the equation 
$\sum_j \beta_j a_j-s = \widetilde{F}(\tau(\und{\beta})_B)$, and the phase diagram 
is contained in the half space $\sum_j \beta_j a_j-s \geq \widetilde{F}(\tau(\und{\beta})_B)$, 
corresponding to the fact that their free entropy is larger than that of the
thermal state. In fact, due to the strict concavity of the entropy, and hence 
of the upper boundary of the phase diagram, the phase diagram, with the exception 
of the thermal point $\bigl(\Tr \tau(\und{\beta})_B \underline{A_{B}}, S(\tau(\und{\beta})_B)\bigr)$ 
is contained in the open half space $\sum_j \beta_j a_j-s > \widetilde{F}(\tau(\und{\beta})_B)$. 

One of many ways to construct a suitable bath $B$ is as several ($R\gg 1$) non-interacting 
copies of an ``elementary bath'' $b$: $B=b^R$ and charges $A_{B_j}=A^{(R)}_{b_j}$, so that 
the GGS of $B$ is $\tau(\und{\beta})_B = \tau(\und{\beta})_b^{\otimes R}$. 
We claim that for large enough $R$, the left hand side of Eq. (\ref{eq:state assumptions}) 
defines a point in the phase diagram of $B$. Indeed, we can express the conditions 
in terms of $b$, assuming that we aim for a regular sequence of product states
$\xi_{b^{nR}}$:
\begin{equation}\begin{split}
  s(\{\xi_{b^{nR}}\})   &= S(\tau(\und{\beta})_b) - \frac1R \Delta s_S, \\
  a_j(\{\xi_{b^{nR}}\}) &= \Tr \tau(\und{\beta})_b A_{b_j} - \frac1R (\Delta A_{S_j} + W_j) 
                                                             \quad \text{for all } j=1,\ldots,c.
  \label{eq:R-bath-assumptions}
\end{split}\end{equation}
For all sufficiently large $R$, these points $(\und{a},s)$ are arbitrarily close to
where the bath starts off, at 
$(\und{a}_{\und{\beta}},s_{\und{\beta}}) 
 = \bigl(\Tr\tau(\und{\beta})_b A_{b_1}, \ldots, \Tr\tau(\und{\beta})_b A_{b_c}, S(\tau(\und{\beta})_b)\bigr)$,
while they always remains in the open half plane 
$\sum_j \beta_j a_j-s > \widetilde{F}(\tau(\und{\beta})_b)$. Indeed, 
they all lie on a straight line pointing from $(\und{a}_{\und{\beta}},s_{\und{\beta}})$
into the interior of that half plane. Hence, for sufficiently large $R$,
$(\und{a},s) \in \overline{\cP}$, the phase diagram of $b$, and by 
point 5 of Lemma~\ref{lemma:phase diagram properties} there does indeed exist
a regular sequence of product states corresponding to it.
\end{proof}

\zk{In the next two subsections we study the achievability of the second law in a setting where the thermal bath is given. Namely, given a bath system with fixed charges and the thermal states $\tau(\und{\beta})_B^{\otimes n}$, we aim to understand whether a specific work transformation is feasible and if so what is the minimum size of the thermal bath to implement such a work transformation? We rigorously state these questions as Q1 and Q2 in subsections \ref{subsec:finitebath} and \ref{subsec:bath-rate}, respectively, and answer them in their corresponding subsections.}

\subsection{Finiteness of the bath: tighter constraints and negative entropy}
\label{subsec:finitebath}
In the previous two subsections we have elucidated the traditional statements of 
the first and second law of thermodynamics, as emerging in our resource theory. 
In particular, the second law is tight, if sufficiently large baths are allowed 
to be used. 

Here, we specifically look at the the second statement (achievability) 
of the second law in the presence of an explicitly given, finite bath $B$. It will 
turn out that usually, equality in the second law cannot be attained, only 
up to a certain loss due to the finiteness of the bath. We also discover 
a purely quantum effect whereby the system and the bath remain entangled after 
effecting a certain state transformation, allowing quantum engines to perform 
tasks impossible classically (i.e. with separable correlations). 
The question we want to address is the following refinement of the one answered 
in the previous subsection:

\begin{quote}
\textit{
%
\zk{
Q1: For a given bath $B$, regular sequences $\{\rho_{S^n}\}$ and $\{\sigma_{S^n}\}$ of the initial and final states of the product form, respectively, as well as real numbers $W_1,\ldots,W_c$ satisfying the second law, 
are there extensions $\sigma_{S^nB^n}$ of $\sigma_{S^n}$ 
forming a regular sequence of product states, such that the work 
transformation $\rho_{S^n} \ox \tau(\und{\beta})_B^{\ox n} \rightarrow \sigma_{S^nB^n}$
is feasible with the extracted works at rates $W_1,\ldots,W_c$?
}
%
}
\end{quote}

To answer it, we need the following \emph{extended phase diagram}.
For a give state $\sigma_S$ of the system $S$, and a bath $B$, 
\aw{it is defined as the following set:}
\begin{equation}
  \cP^{(1)}_{|\sigma_S} := \left\{ \bigl(\Tr\xi_B A_1^{(B)},\ldots,\Tr\xi_B A_c^{(B)},S(B|S)_\xi\bigr) : 
                                   \xi_{SB} \text{ state with } \Tr_B\xi_{SB}=\sigma_S \right\}.
\end{equation}
\aw{Furthermore its $n$-copy version, for a given product state}
$\sigma_{S^n}=(\sigma_{1})_{S_1}\ox\cdots\ox(\sigma_n)_{S_n}$,
\begin{align}
 \cP^{(n)}_{|\sigma_{S^n}} 
    :=  \left\{\! \bigl(\Tr\xi_{B^n} A_1^{(B^n)}\!,\!\ldots\!,\!\Tr\xi_{B^n} A_c^{(B^n)}\!\!,S(B^n|S^n)_\xi\bigr)\! : 
               \xi_{S^nB^n} \text{ state with } \Tr_{B^n}\xi_{S^nB^n}\!=\sigma_{S^n} \!\right\}\!.
\end{align}
\aw{These sets capture which combinations of charge value of the bath and conditional 
von Neumann entropy $S(B|S)$ of the bath conditional on the system are consistent 
with quantum mechanics. Note that extended phase diagram contains the previously 
discussed phase diagram of the bath, since we can choose $\xi_{SB}=\sigma_S\ox\xi_B$
as a product state, and then $S(B|S)_\xi = S(\xi_B)$, but correlations between 
the system and the bath can reduce the conditional entropy below this quantity,
in some cases not only to zero but to negative values.} 
Finally, define the \emph{conditional entropy phase diagram} as 
\begin{align}
  \overline{\cP}_{|s_0} := \overline{\cP}^{(1)}_{|s_0}
    := \left\{ \bigl(\und{a},s\bigr) : a_j = \Tr\xi_B A_j^{(B)},\,
                                       -\min\{s_0,S(\tau(\und{a}))\} \leq s \leq S(\tau(\und{a})) 
                                       \text{ for a state } \xi_B \right\},
\end{align}
and likewise its $n$-copy version $\overline{\cP}^{(n)}_{|ns_0}$, 
for a number $s_0$ (intended to be an entropy or entropy rate). 
These concepts are illustrated in Fig.~\ref{fig:extended-phase-diagram}.
The relation between the sets, and the name of the latter, are explained in the 
following lemma.

\begin{figure}[ht]
\begin{center}   
  \includegraphics[scale=0.4]{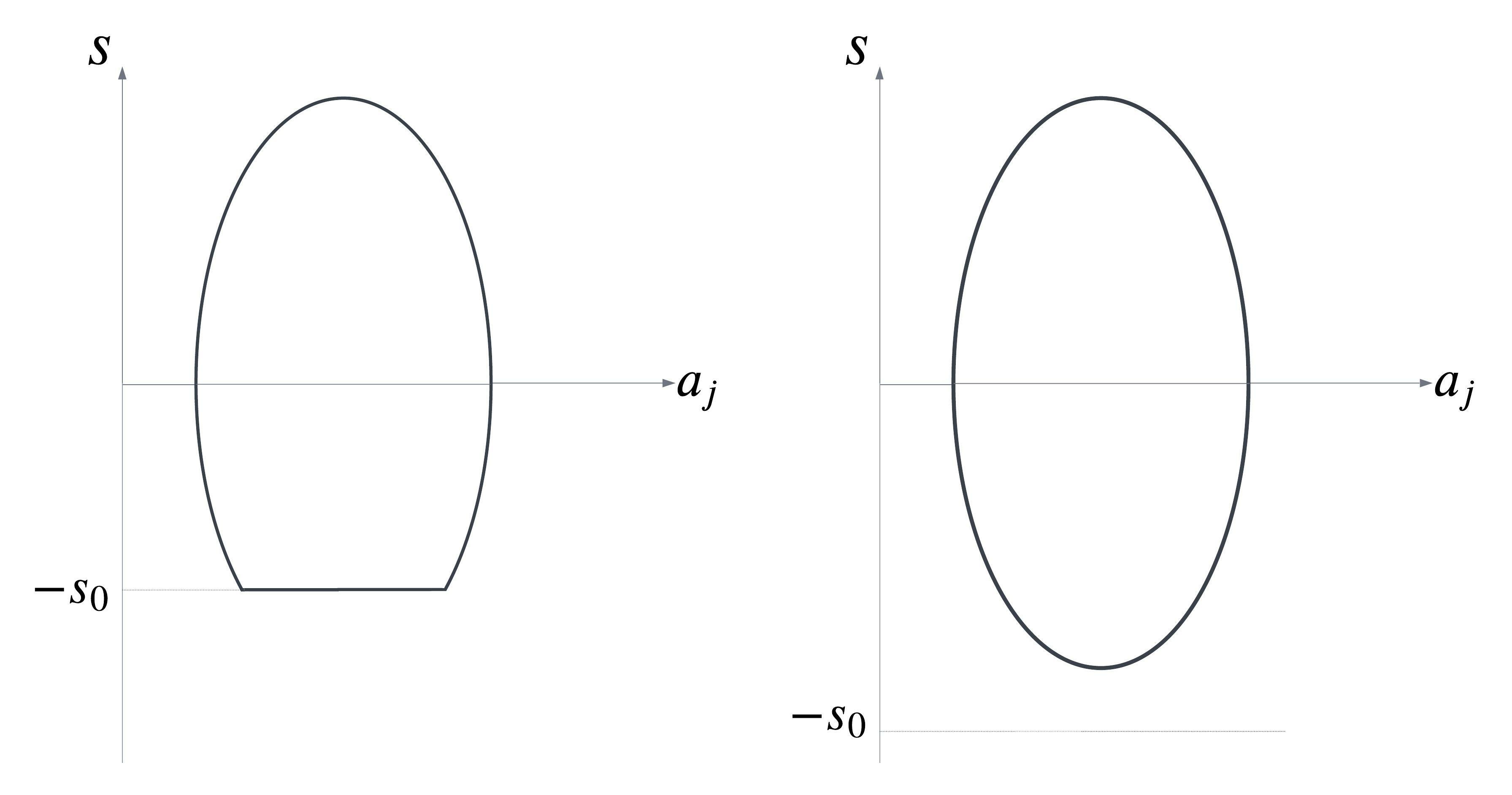}
  \end{center}
  \caption{Schematic of the extended phase diagram $\overline{\cP}_{|s_0}$. 
           Depending on the value of $s_0$, whether it is smaller or larger than 
           $\log|B|$, the diagram acquires either the left hand or the right hand 
           one of the above shapes.}
  \label{fig:extended-phase-diagram}
\end{figure}

\begin{lemma}
\label{lemma:extended-phase-diagram}
With the previous notation, we have:
\begin{enumerate}
  \item For all $k$, $\cP^{(k)}_{|\sigma_{S^k}} \subset \overline{\cP}^{(k)}_{|S(\sigma_{S^k})}$, 
        and the latter is a closed convex set.
  \item For all $k$, $\overline{\cP}^{(k)}_{|ks_0} = k \overline{\cP}^{(1)}_{|s_0}$.
  \item For a regular sequence $\{\sigma_{S^k}\}$ of product states with entropy rate
        $s_0=s(\{\sigma_{S^k}\})$, every point in $\overline{\cP}_{|s}$ is arbitrarily well 
        approximated by points in $\frac1k \cP^{(k)}_{|\sigma_{S^k}}$ for all sufficiently large $k$. 
        I.e., $\displaystyle{\overline{\cP}_{|s_0} = \lim_{k\to\infty} \frac1k \cP^{(k)}_{|S(\sigma_{S^k})}}$.        
\end{enumerate}
\end{lemma}

\begin{proof}
1. We only have to convince ourselves that for a state
$\xi_{S^kB^k}$ with $\Tr_{B^k}\xi_{S^kB^k}=\sigma_{S^k}$, 
\[
  - \min\{S(\sigma_{S^k}),kS(\tau(\und{a}))\} \leq S(B^k|S^k)_\xi \leq k S(\tau(\und{a})), 
\]
where $\und{a}=(a_1,\ldots,a_c)$ with $a_i = \frac1k \Tr\xi_{B^k} A_i^{(B^k)}$.
The upper bound follows from subadditivity, since 
$S(B^k|S^k)_\xi \leq S(B^k)_\xi \leq k S(\tau(\und{a}))$.
The lower bound consists of two inequalities: first, by purifying $\xi$ to 
a state $\ket{\phi} \in S^kB^kR$ and strong subadditivity, 
$S(B^k|S^k)_\xi \geq S(B^k|S^kR)_\phi = - S(B^k)_\xi \geq - k S(\tau(\und{a}))$.
Secondly, $S(B^k|S^k)_\xi \geq -S(S^k)_\xi = -S(\sigma_{S^k})$.

2. Follows easily from the definition. 

3. It is enough to show that the points of the minimum entropy diagram 
\[
  \overline{\cP}_{\min|s} 
    := \left\{\bigl(\und{a},-\min\{s_0,S(\tau(\und{a}))\}\bigr) 
              : \Tr \xi_B A_j^{(B)} = a_j \text{ for a state } \xi_B \right\}
\]
can be approximated as claimed by an admissible $k$-copy state $\xi_{S^kB^k}$. This 
is because the maximum entropy diagram $\overline{\cP}_{\max}^{(k)}$ is realized 
by states $\vartheta_{S^kB^k} := \sigma_{S^k}\ox\tau(\und{a})_B^{\ox k}$, and by 
interpolating the states, i.e. $\lambda\xi + (1-\lambda)\vartheta$ for $0\leq\lambda\leq 1$,
we can realize the same charge values $\und{a}$ with entropies in the 
whole interval $[S(B^k|S^k)_\xi;k S(\tau(\und{a}))]$.

The approximation of $\overline{\cP}_{\min|s}$ can be proved invoking 
results from quantum Shannon theory, specifically \emph{quantum state merging}, 
the form of which that we need here is stated below as Lemma \ref{lemma:marge}. 
For this, consider a tuple $\und{a} \in \overline{\cP}_0$ and a purification 
$\ket{\Psi} \in S^kB^kR^k$ of the state $\vartheta_{S^kB^k} = \sigma_{S^k}\ox\tau(\und{a})_B^{\ox k}$, 
which can be chosen in such a way as to be a product state itself: 
$\ket{\Psi} = \ket{\Psi_1}_{S_1B_1R_1}\ox\cdots\ox\ket{\Psi_k}_{S_kB_kR_k}$. 
\aw{Our strategy is to find $\xi_{S^kB^k}$ as correlated as possible, 
in the sense that we would like to minimize its entropy, subject to the constraint 
that its marginal on $S^k$ is $\sigma_{S^k}$ and that on $B^k$ shares the 
charge values with $\tau(\und{a})_B^{\ox k}$. 
As we do not know an explicit construction that achieves this, we 
resort to a random one that succeeds with high probability, which is 
what quantum state merging facilitates.}

We distinguish two cases, depending on which of the entropies 
$S(\sigma_{S^k})$ and $kS\bigl(\tau(\und{a})_B\bigr)$ is the smaller.

\begin{enumerate}[{(i)}]
  \item $\mathbf{S(\sigma_{S^k}) \geq S\bigl(\tau(\und{a})_B\bigr)}$: We shall 
    construct $\xi_{S^kB^k}$ in such a way that $\xi_{S^k} = \sigma_{S^k}$ and 
    $\xi_{B^k} \approx \tau\bigl(\und{a}\bigr)_B^{\ox k}$. To this end, choose 
    a pure state $\phi_{CR'}$ with entanglement entropy 
    $S(\phi_C) = \frac1k S(\sigma_{S^k})-S\bigl(\tau(\und{a})_B\bigr) + \frac12 \epsilon$,
    and consider the state 
    $\widetilde{\Psi}^{S^kB^kC^kR^k{R'}^k} = \Psi_{S^kB^kR^k}\ox\phi_{CR'}^{\ox k}$.
    Now we apply state merging (Lemma \ref{lemma:marge}) twice to 
    this state (which is a tensor product of $k$ systems), with a random 
    rank-one projector $P$ on the combined system $R^k{R'}^k$: 
    first, by splitting the remaining parties $S^k : B^kC^k$, 
    and second by splitting them $B^k : S^kC^k$. 
    By construction, in both bipartitions it is the solitary system 
    ($S^k$ and $B^k$, resp.) that has the smaller entropy by at least $\frac12\epsilon k$, 
    showing that the post-measurement state $\widetilde{\xi}(P)_{S^kB^kC^k}$ with 
    high probability approximates the marginals of $\vartheta_{S^kB^k}$ on $S^k$ 
    and on $B^k$ simultaneously.
    Choose a typical subspace projector $\Pi$ of $\phi_C^{\ox k}$ with 
    $\log \rank \Pi \leq S(\sigma_{S^k})-k S\bigl(\tau(\und{a})_B\bigr) + \epsilon k$,
    and let
    \[
      \ket{\xi(P)}_{S^kB^kC^k} := \frac1c (\1_{S^kB^k}\Pi_{C^k})\ket{\widetilde{\xi}(P)},
    \]
    with a normalization constant $c$. 
    Merging and properties of the typical subspace imply that for sufficiently large $k$, 
    \begin{align}
      \label{eq:xiP-S}
      \frac12 \left\| \xi(P)_{S^k} - \sigma_{S^k} \right\|_1            &\leq \epsilon, \\
      \label{eq:xiP-B}
      \frac12 \left\| \xi(P)_{B^k} - \tau(\und{a})_B^{\ox k} \right\|_1 &\leq \epsilon.
    \end{align}
    Now, we invoke Uhlmann's theorem applied to purifications of $\sigma_{S^k}$ and  
    of $\xi(P)_{S^kB^k}$, together with the well-known relations between fidelity 
    and trace norm applied to Eq.~(\ref{eq:xiP-S}), 
    to obtain a state $\xi_{S^kB^k}$ with $\xi_{S^k} = \sigma_{S^k}$ and 
    $\frac12 \left\| \xi(P)_{S^kB^k} - \xi_{S^kB^k} \right\|_1 \leq \sqrt{\epsilon(2-\epsilon)}$,
    thus by Eq.~(\ref{eq:xiP-B})
    \[
      \frac12 \left\| \xi_{B^k} - \tau(\und{a})_B^{\ox k} \right\|_1 \leq \epsilon + \sqrt{\epsilon(2-\epsilon)}.
    \]
    From the latter bound it follows that
    \[
      \left| \frac1k \tr \xi_{B^k}A_j^{(B^k)} - a_j \right| 
               \leq \|A_{B_j}\| \left(\epsilon + \sqrt{\epsilon(2-\epsilon)}\right).
    \]
    It remains to bound the conditional entropy: 
    \[\begin{split}
      \frac1k S(B^k|S^k)_\xi &=    \frac1k S\bigl(\xi_{S^kB^k}\bigr) - \frac1k S(\xi_{S^k})           \\
                         &\leq \frac1k S\bigl(\xi(P)_{S^kB^k}\bigr) - \frac1k S(\sigma_{S^k}) 
                                  + \left( \epsilon + \sqrt{\epsilon(2-\epsilon)} \right)\log(|S||B|) 
                                  + h\left( \epsilon + \sqrt{\epsilon(2-\epsilon)} \right)           \\
                         &\leq \frac1k \log \rank\Pi - \frac1k S(\sigma_{S^k}) 
                                  + \left( \epsilon + \sqrt{\epsilon(2-\epsilon)} \right)\log(|S||B|) 
                                  + h\left( \epsilon + \sqrt{\epsilon(2-\epsilon)} \right)           \\
                         &\leq \frac1k \left( S(\sigma_{S^k}) - kS\bigl(\tau(\und{a})\bigr) \right) 
                                  - \frac1k S(\sigma_{S^k}) 
                                  + \left( 2\epsilon + \sqrt{\epsilon(2-\epsilon)} \right)\log(|S||B|) 
                                  + h\left( \epsilon + \sqrt{\epsilon(2-\epsilon)} \right)            \\
                         &=    -S\bigl(\tau(\und{a})\bigr) 
                                  + \left( 2\epsilon + \sqrt{\epsilon(2-\epsilon)} \right)\log(|S||B|) 
                                  + h\left( \epsilon + \sqrt{\epsilon(2-\epsilon)} \right) ,
    \end{split}\]
    where in the second line we have used the Fannes inequality on the continuity 
    of the entropy \cite{Fannes1973,Audenaert2007}, with the binary entropy
    $h(x)=-x\log x-(1-x)\log(1-x)$;
    in the third line that $\xi(P)_{S^kB^k}$ has rank at most $\rank \Pi$;
    and in the fourth line the upper bound on the latter rank by construction. 

  \item $\mathbf{S(\sigma_{S^k}) < S\bigl(\tau(\und{a})_B\bigr)}$: We shall 
    construct $\xi_{S^kB^k}$ such that $\xi_{S^k} = \sigma_{S^k}$ and 
    $\tr\xi_{B^k}A_j^{(B^k)} \approx \tr\tau\bigl(\und{a}\bigr)_B A_{B_j}$ for 
    all $j=1,\ldots,c$. Here, choose a pure state $\phi_{CR'}$ with entanglement entropy 
    $S(\phi_C) = \epsilon$, and define
    $\widetilde{\Psi}^{S^kB^kC^kR^k{R'}^k} = \Psi_{S^kB^kR^k}\ox\phi_{CR'}^{\ox k}$.
    Now we apply state merging (Lemma \ref{lemma:marge}) to
    this state (which is a tensor product of $k$ systems), with a random 
    rank-one projector $P$ on the combined system $R^k{R'}^k$,
    by splitting the remaining parties $S^k : B^kC^k$, which ensures that
    $S^k$ has the smaller entropy by at least $\epsilon k$, 
    showing that the post-measurement state $\widetilde{\xi}(P)_{S^kB^kC^k}$ with 
    high probability approximates the marginal of $\vartheta_{S^kB^k}$ on $S^k$.
    Proceed as before with a typical subspace projector $\Pi$ of $\phi_C^{\ox k}$ 
    such that  
    $\log \rank \Pi \leq S(\sigma_{S^k})-k S\bigl(\tau(\und{a})_B\bigr) + \epsilon k$,
    and let
    \(
      \ket{\xi(P)}_{S^kB^kC^k} := \frac1c (\1_{S^kB^k}\Pi_{C^k})\ket{\widetilde{\xi}(P)},
    \)
    with a normalization constant $c$. 
    Merging and properties of the typical subspace thus imply that for sufficiently large $k$, 
    \begin{equation}
      \label{eq:xiPt-S}
      \frac12 \left\| \xi(P)_{S^k} - \sigma_{S^k} \right\|_1 \leq \epsilon. 
    \end{equation}
    Next we need to look at the charge values of $\xi(P)_{B^k}$. 
    Note that the expectation $\EE_P \xi(P)_{B^k}$ is approximately 
    equal to $\EE_P \widetilde{\xi}(P)_{B^k} = \tau(\und{a})_B^{\ox k}$.
    It follows from \cite[{Lemma III.5}]{Hayden2006}, that if $k$ is sufficiently 
    large, then with high probability 
    \begin{equation}
      \label{eq:xiPt-B-A}
      \left| \tr \bigl(\xi(P)_{B^k} - \tau(\und{a})_B^{\ox k}\bigr) A_j^{(B^k)} \right| \leq \|A_{B_j}\| \epsilon 
                      \quad \text{for all } j=1,\ldots,c. 
    \end{equation}
    So we just focus on a good instance of $P$, where both 
    Eqs.~(\ref{eq:xiPt-S}) and (\ref{eq:xiPt-B-A}) hold. Now we proceed as
    in the first case to find a state $\xi_{S^kB^k}$ with $\xi_{S^k} = \sigma_{S^k}$ and 
    $\frac12 \left\| \xi(P)_{S^kB^k} - \xi_{S^kB^k} \right\|_1 \leq \sqrt{\epsilon(2-\epsilon)}$, 
    using Uhlmann's theorem. Thus, as before we find
    \[
      \left| \frac1k \tr \xi_{B^k}A_j^{(B^k)} - a_j \right| 
                 \leq \|A_{B_j}\| \left(\epsilon + \sqrt{\epsilon(2-\epsilon)}\right).
    \]
    Regarding the conditional entropy, we have quite similarly as before,
    \[\begin{split}
      \frac1k S(B^k|S^k)_\xi &=    \frac1k S\bigl(\xi_{S^kB^k}\bigr) - \frac1k S(\xi_{S^k})           \\
                             &\leq \frac1k S\bigl(\xi(P)_{S^kB^k}\bigr) - \frac1k S(\sigma_{S^k}) 
                                  + \left( \epsilon + \sqrt{\epsilon(2-\epsilon)} \right)\log(|S||B|) 
                                  + h\left( \epsilon + \sqrt{\epsilon(2-\epsilon)} \right)           \\
                             &\leq \frac1k \log 2^{\epsilon k} - \frac1k S(\sigma_{S^k}) 
                                  + \left( \epsilon + \sqrt{\epsilon(2-\epsilon)} \right)\log(|S||B|) 
                                  + h\left( \epsilon + \sqrt{\epsilon(2-\epsilon)} \right)           \\
                             &\leq -\frac1k S(\sigma_{S^k}) 
                                  + \left( 2\epsilon + \sqrt{\epsilon(2-\epsilon)} \right)\log(|S||B|) 
                                  + h\left( \epsilon + \sqrt{\epsilon(2-\epsilon)} \right).
    \end{split}\]
\end{enumerate}

Since in both cases we knew the conditional entropy to be always
$\geq - \frac1k \min\left\{ S(\sigma_{S^k}),k S\bigl(\tau(\und{a})\bigr) \right\}$,
this concludes the proof.
\end{proof}

\medskip
\begin{lemma}[Quantum state merging \cite{HOW:negative-Nature,HOW:negative-CMP}]
\label{lemma:marge}
Given a pure product state 
$\Psi_{A^nB^nC^n}=(\Psi_1)_{A_1B_1C_1}\ox\cdots\ox(\Psi_n)_{A_nB_nC_n}$,
such that $S(\Psi_{A^n})-S(\Psi_{B^n}) \geq \epsilon n$, consider a Haar 
random rank-one projector $P$ on $C^n$. Then, it holds that the post-measurement state 
\[
  \psi(P)_{A^nB^n} = \frac{1}{\tr\Psi_{C^n}P}\tr_{C^n}\Psi(\1_{A^nB^n}\ox P)
\]
satisfies $\frac12\| \psi(P)-\Psi_{A^nB^n}\|_1 \leq \epsilon$, 
except with arbitrarily small probability for sufficiently large $n$.
\hfill$\blacksquare$
\end{lemma}

\medskip
\begin{remark}\normalfont
While we have seen that the upper boundary of the extended phase diagram 
$\overline{\cP}^{(k)}_{|S(\sigma_{S^k})}$ is exactly realized by points 
in $\cP^{(k)}_{|\sigma_{S^k}}$, namely those corresponding to the tensor product 
states $\sigma_{S^k} \ox \tau(\und{a})_B^{\ox k}$, 
it seems unlikely that we can achieve the analogous thing for the lower boundary: 
this would entail finding, for every (sufficiently large) $k$ a tensor product
state, or a block tensor product state, $\xi_{S^kB^k}$ with prescribed charge vector 
$\und{a}$ on $B^k$, and $S(B^k|S^k)_\xi = -\min\{kS\bigl(\tau(\und{a})\bigr),S(\sigma_{S^k})\}$. 

Now, for concreteness, consider the case that 
$kS\bigl(\tau(\und{a})\bigr) \leq S(\sigma_{S^k})$, so that the conditional entropy 
aimed for is $S(B^k|S^k)_\xi = -k S\bigl(\tau(\und{a})_B\bigr)$, which is the value 
of a purification of $\tau(\und{a})_B^{\ox k}$. In particular, it would mean that 
$S(\xi_{B^k}) = k S\bigl(\tau(\und{a})_B\bigr)$, and so -- recalling the
charge values and the maximum entropy principle -- it would follow that 
$\xi_{B^k} = \tau(\und{a})_B^{\ox k}$. 
However, from the equality conditions in strong subadditivity \cite{SSA-eq}, 
this in turn would imply that $\xi_{S^kB^k}$ is a probabilistic mixture of 
purifications of $\tau(\und{a})_B^{\ox k}$ whose restrictions to $S^k$ 
are pairwise orthogonal. This would clearly put constraints on the spectrum 
of $\sigma_{S^k}$ that are not generally met. 

In the other case that $kS\bigl(\tau(\und{a})\bigr) > S(\sigma_{S^k})$, the
conditional entropy should be $S(B^k|S^k)_\xi = - S(\sigma_{S^k})$, and 
since $\xi_{S^k} = \sigma_{S^k}$, this would necessitate a pure state 
$\xi_{S^kB^k}$. Looking at the proof of Lemma \ref{lemma:extended-phase-diagram}, 
however, we see that it leaves quite a bit of manoeuvring space, so it
may or may not be possible to satisfy all charge constraints
$\tr \xi_{B^k}A_j^{(B^k)} = a_j$ ($j=1,\ldots,c$).
\end{remark}

\medskip
Coming back to our question, if a work transformation 
$\rho_{S^n} \ox \tau(\und{\beta})_B^{\ox n} \rightarrow \sigma_{S^nB^n}$ is
feasible for regular sequences on the left hand side, 
by the first law this implies that 
\begin{align*}
  s(\{\sigma_{S^nB^n}\}) &= s(\{\rho_{S^n}\}) + S(\tau(\und{\beta})) \text{ and} \\ 
  W_j                    &= -\Delta A_{S_j}-\Delta A_{B_j}  \\
                         &= a_j(\{\rho_{S^n}\}) - a_j(\{\sigma_{S^n}\})
                            + a_j(\{\tau(\und{\beta})_{B^n}\}) - a_j(\{\sigma_{B^n}\}). 
\end{align*}
When $\sigma_{S^n}$ and the $W_j$ are given, this constrains the possible 
states $\sigma_{S^nB^n}$ as follows: for each $n$, 
\begin{align*}
  \frac1n S(B^n|S^n)_\sigma              &\approx S(\tau(\und{\beta})) - \Delta s_S, \\
  \frac1n \Tr \sigma_{B^n} A_{B_j}^{(n)} &\approx \Tr \tau(\und{\beta})_B A_{B_j} 
                                                  - \Delta A_{S_j} - W_j,\quad \text{for all } j=1,\ldots,c.
\end{align*}
Since by Lemma \ref{lemma:extended-phase-diagram} the left hand sides converge to the 
components of a point in $\overline{\cP}_{|s(\{\sigma_{S^n}\})}$, meaning that a necessary 
condition for the feasibility of the work transformation in question is that 
\begin{equation}\begin{split}
  \label{eq:conditional-point}
  (\und{a},t) \in \overline{\cP}_{|s(\{\sigma_{S^n}\})}, \text{ with } 
  a_j         &:=  \Tr \tau(\und{\beta})_B A_{B_j} - \Delta A_{S_j} - W_j, \\
  t           &:=  S(\tau(\und{\beta})) - \Delta s_S.
\end{split}\end{equation}
Again by Lemma \ref{lemma:extended-phase-diagram}, this is equivalent to 
all $a_j$ to be contained in the set of joint quantum expectations of the 
observables $A_{B_j}$, and 
\[
  -\min\left\{ s(\{\sigma_{S^n}\}),S\bigl(\tau(\und{a})\bigr) \right\} \leq t \leq S\bigl(\tau(\und{a})\bigr).
\]
The following theorem shows that this is also sufficient, when we allow blockings 
of the asymptotically many systems. 

\begin{figure}[ht]
  \begin{center}
    \includegraphics[width=10cm,height=8cm]{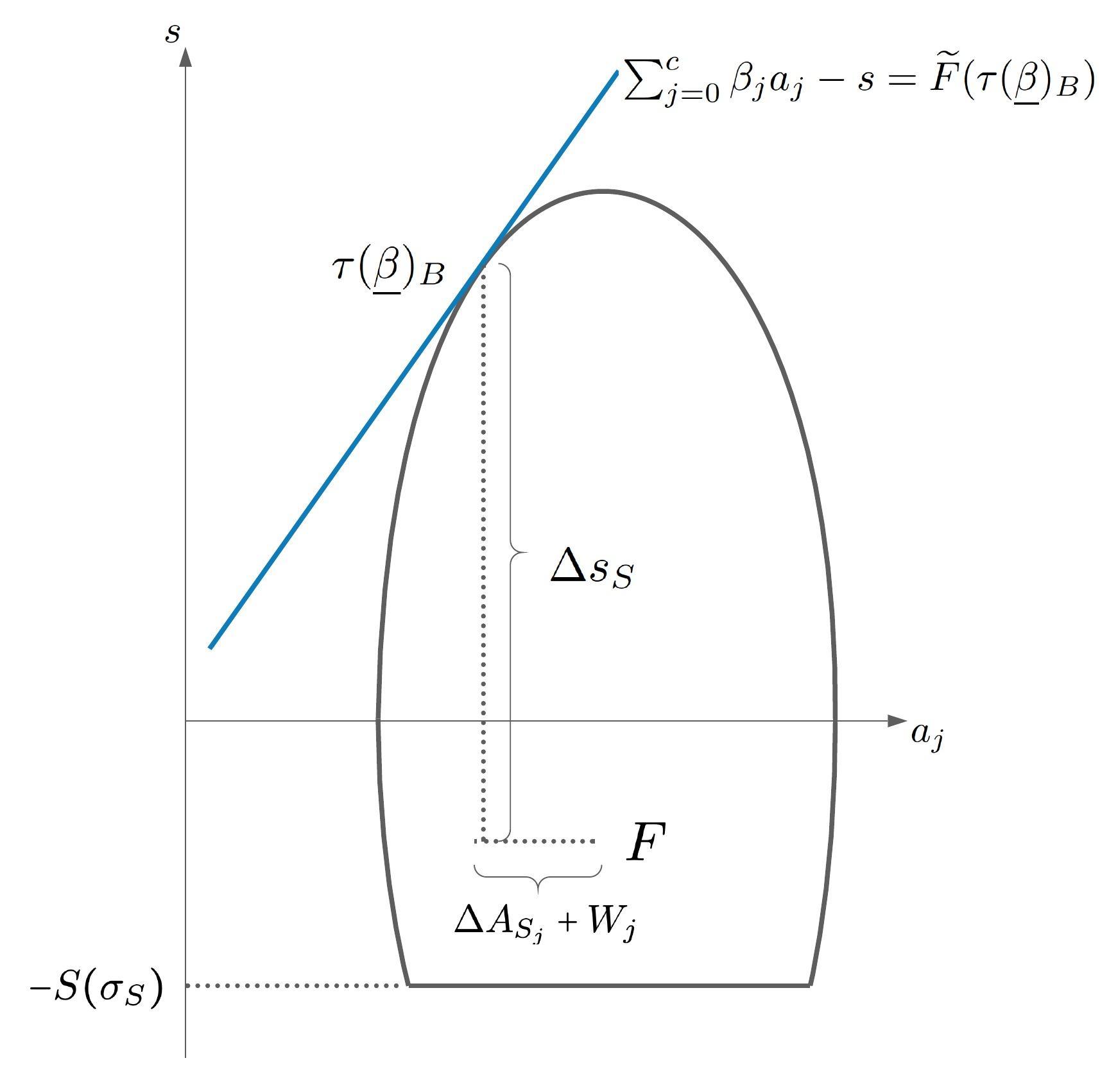}
  \end{center}
  \caption{State change of the bath for a given work transformation under the extraction of
           $j$-type work $W_j$, viewed in the extended phase diagram of the bath, which 
           initially is in the thermal state $\tau(\und{\beta})_B$, the blue line at the 
           corresponding point in the diagram representing the tangent hyperplane of the 
           diagram. The final states $\{\sigma_{S^nB^n}\}$ give rise to the point $F$ in 
           the extended diagram, whose charge values are those of $\{\sigma_{B^n}\}$, 
           while the entropy is $\frac1n S(B^n|S^n)_\sigma$.}
  \label{fig:second-law-finite}
\end{figure} 

\begin{theorem}[Second Law with fixed bath]
\label{thm:second-law-finite-bath}
For arbitrary regular sequences $\rho_{S^n}$ and $\sigma_{S^n}$ of product states, 
a given bath $B$, and any real numbers $W_j$, 
if there exists a regular sequence of block product states 
$\sigma_{S^nB^n}$ with $\Tr_{B^n}\sigma_{S^nB^n} = \sigma_{S^n}$, such that 
there is a work transformation 
$\rho_{S^n} \ox \tau(\und{\beta})_B^{\ox n} \rightarrow \sigma_{S^nB^n}$
with accompanying extraction of $j$-type work at rate $W_j$,
then Eq. (\ref{eq:conditional-point}) defines a point
$(\und{a},t) \in \overline{\cP}_{|s(\{\sigma_{S^n}\})}$.

Conversely, assuming additionally that $\sigma_{S^n} = \sigma_S^{\otimes n}$ is an i.i.d.~state,
if Eq. (\ref{eq:conditional-point}) defines a point
$(\und{a},t) \in \overline{\cP}_{|S(\sigma_S)}^0$ in the interior of the 
extended phase diagram, then for every $\epsilon>0$
there is a work transformation 
$\rho_{S^n} \ox \tau(\und{\beta})_B^{\ox n} \rightarrow \sigma_{S^nB^n}$
with block product states $\sigma_{S^nB^n}$ such that 
$\Tr_{B^n}\sigma_{S^nB^n} = \sigma_{S^n}$, and with accompanying 
extraction of $j$-type work at rate $W_j\pm\epsilon$.
This is illustrated in Fig.~\ref{fig:second-law-finite}.
\end{theorem}

\begin{proof}
We have already argued the necessity of the condition. It remains to show its 
sufficiency. 
Using Lemma \ref{lemma:extended-phase-diagram}, this is not hard: Namely, 
by its point 3, for sufficiently large $k$, $(\und{a},t) \in \overline{\cP}_{|s}$ 
is $\epsilon$-approximated by $\frac1k \cP^{(k)}_{|\sigma_S^{\ox k}}$, i.e. there 
exists a $\sigma_{S^kB^k}$ with $\tr_{B^k} \sigma_{S^kB^k} = \sigma_S^{\ox k}$ 
with $\frac1k S(B^k|S^k)_\sigma \leq t-\epsilon$ and $\frac1k \tr \sigma_{B^k} A_j^{(B^k)} \approx a_j$
for all $j=1,\ldots,c$. By mixing $\sigma$ with a small fraction of 
$\bigl(\tau(\und{a})_B\ox\sigma_S\bigr)^{\ox k}$, we can in fact assume that 
$\frac1k S(B^k|S^k)_\sigma = t$ while preserving $\frac1k \tr \sigma_{B^k} A_j^{(B^k)} \approx a_j$.
Now our target block product states will be 
$\sigma_{S^nB^n} := \bigl(\sigma_{S^kB^k}\bigr)^{\ox \frac{n}{k}}$ for $n$ a multiple of $k$.
By construction, this sequence has the same entropy rate as 
the initial regular sequence of product states $\rho_{S^n} \ox \tau(\und{\beta})_B^{\ox n}$, 
so by the first law, Theorem \ref{thm:first-law}, and the 
AET, Theorem \ref{Asymptotic equivalence theorem}, there is indeed a 
corresponding work transformation with $j$-type work extracted 
equal to $W_j\pm\epsilon$.
\end{proof}

\medskip
\begin{remark}\normalfont
One might object that tensor power target states are not general enough 
in Theorem \ref{thm:second-law-finite-bath}, as we 
had observed in Section \ref{sec:resource-theory} that such states do not
generate the full phase diagram $\overline{\mathcal{P}}$ of the system $S$. 
However, by considering blocks of $\ell$ systems $S^\ell$, we can 
apply the theorem to block tensor power target states 
$\sigma_{S^n} = \bigl(\sigma_1\ox\cdots\ox\sigma_\ell\bigr)^{\ox \frac{n}{\ell}}$, 
and these latter are in fact a rich enough class to exhaust the entire 
phase diagram $\overline{P}$, when $\ell \geq \dim S$ 
(point 5 of Lemma \ref{lemma:phase diagram properties}).

More generally, we can allow as target \emph{uniformly regular} sequences of 
product states $\sigma_{S^n}$, by which we mean the following strengthening 
of the condition in Definition \ref{definition:regular}. 
Denoting $B_{N+1}^{N+n} := B_{N+1}\ldots B_{N+n}$, we require that for all 
$\epsilon > 0$ and uniformly for all $N$, it holds that for sufficiently large $n$, 
\[
  \left| a_j - \frac1n \Tr \sigma_{B_{N+1}^{N+n}} A_j^{(n)} \right| 
                                                       \leq \epsilon \text{ for all } j=1,\ldots,c, \text{ and } 
  \left| s - \frac1n S(\sigma_{B_{N+1}^{N+n}}) \right| \leq \epsilon. 
\]
\end{remark}

\medskip
\begin{remark}\normalfont
\aw{We conclude this subsection with a reflexion on the peculiar role 
of entanglement played in the quantum advantage implied by 
Theorem \ref{thm:second-law-finite-bath}. Indeed, whereas in many 
quantum tasks entanglement is the fuel requisite at the beginning to perform 
super-classically, here it is the possibility of leaving the system and 
bath in an entangled state which allows to reach points in the extended
phase diagram outside the usual phase diagram, i.e. with negative conditional 
entropy $S(B|S)$. Note that no separable state can achieve this, as by the 
result of \cite{HoroHoro1994,NielsenKempe} then $S(B|S)\geq 0$.}

\aw{Evidently, demonstrating such an effect would require phenomenal 
control of the quantum degrees of freedom of both $S$ and $B$,
so in a macroscopic system that would presumably be impossible. But
we believe it not completely beyond the bounds of the recent 
demonstrations of thermal machines in mesoscopic and nanoscopic systems.}
While we cannot indicate any concrete references, a well-designed experiment would 
be feasible with any of the contemporary platform for quantum simulations (QS), such as

\begin{itemize}
\item {\bf Superconducting qubits}, used by Google \cite{Google} or D-Wave \cite{D-Wave}, 
      are often employed as digital QSs (cf. \cite{Parao}) and/or in circuit QED 
      systems \cite{Wallraff};
\item {\bf Ultracold atoms}, which offer analog quantum simulation, can be realized 
      in the continuum or in optical lattices \cite{LSA2017}. They are very flexible 
      and they allow to simulate complex Hubbard models, as well as spin systems;
\item {\bf Trapped ions} can also be used as perfect analog or digital QSs \cite{Monroe53,MonroeRMP}.
      They typically simulate spin-$\frac12$ systems, but very recently a qudit quantum 
      computer/simulator was realized with ions \cite{Ringbauer};
\item {\bf Rydberg atoms} are atoms where the electron has been excited to a high 
      principal quantum number, and which are trapped in optical tweezers. They mimic 
      spin systems with long-range interactions \cite{Lukin51,LukinScars,Browaeys};
\item {\bf Light and cavity materials}: Quantum simulators based on cavity QED take 
      advantage of the coupling between quantum system and the coherent light field of 
      the cavity in which such system has been placed. Experiments are 
      mainly conducted in the scope of Jaynes-Cummings and Dicke models \cite{Schlawin_2022}.
      Recent studies concern also engineering materials entirely from light with 
      resulting photon-photon interactions \cite{Clark_2020,carusotto_2020,Ma_2019,Schine_2016};
\item {\bf Twistronics systems}: Twistronics deals with twisted bilayer graphene or 
      other two-dimensional materials \cite{Pablo,Dima}. For small ``magic'' angle, such 
      systems lead to periodic Moir\'e patterns at a length scale much larger than the 
      typical scale of condensed matter systems: in this sense, they can themselves be 
      considered as condensed matter quantum simulators of condensed matter \cite{Kennes21}.
      Twisted bilayer materials can, however, also be mimicked by ultracold atoms in a 
      two-dimensional lattice with synthetic dimensions \cite{tymek};
\item {\bf Polaritons} are especially useful for non-equilibrium systems and quantum 
      hydrodynamics simulation, as well as relativistic effects thanks to dual (half light 
      half particle) nature of the polaritonic quasi-particles \cite{Basov2021,hubener2021,JBloch}.
\end{itemize}
\end{remark}

\subsection{Tradeoff between thermal bath rate and work extraction}  
\label{subsec:bath-rate}
Here we consider a different take on the question of the work deficit due
to finiteness of the bath. Namely, we still consider a given fixed finite 
bath system $B$, but now as which state transformations and associated 
generalized works are possible when for each copy of the subsystem $S$,
$R\geq 0$ copies of $B$ are present. It is clear what that means when 
$R$ is an integer, but below we shall give a meaning to this rate as a real number. 
We start off with the observation that ``large enough bath'' in 
Theorem \ref{asymptotic second law} can be taken to mean $B^R$, for the given 
elementary bath $B$ and sufficiently large integer $R$.

\begin{theorem}
\label{thm:large-R-2nd-law}
For arbitrary regular sequences of product states, 
$\rho_{S^n}$ and $\sigma_{S^n}$, and any real numbers $W_j$ with  
$\sum_{j=1}^c \beta_j W_j < -\Delta\widetilde{F}_S$, 
there exists an integer $R\geq 0$ and a regular sequence of product states 
$\sigma_{S^nB^{nR}}$ with $\Tr_{B^{nR}}\sigma_{S^nB^{nR}} = \sigma_{S^n}$, such that 
there is a work transformation 
$\rho_{S^n} \ox \tau(\und{\beta})_B^{\ox nR} \rightarrow \sigma_{S^nB^{nR}}$
with accompanying extraction of $j$-type work at rate $W_j$.
\end{theorem}

\begin{proof}
This was already shown in the achievability part of Theorem \ref{asymptotic second law}.
\end{proof}

\medskip
To give meaning to a rational rate $R = \frac{\ell}{k}$, group the systems of $S^n$,
for $n=\nu k$, into blocks of $k$, which we denote $\widetilde{S}=S^k$, 
and consider $\rho_{S^n} \equiv \rho_{\widetilde{S}^\nu}$ as a $\nu$-party state,
and likewise $\sigma_{S^n} \equiv \sigma_{\widetilde{S}^\nu}$. For each 
$\widetilde{S}=S^k$ we assume $\ell$ copies of the thermal bath, 
$\tau(\und{\beta})_B^{\otimes \ell} = \tau(\und{\beta})_{\widetilde{B}}$,
with $\widetilde{B} = B^\ell$. If $\{\rho_{S^n}\}$ and $\{\sigma_{S^n}\}$ are
regular sequences of product states, then evidently so are 
$\{\rho_{\widetilde{S}^\nu}\}$ and $\{\sigma_{\widetilde{S}^\nu}\}$. With this definition of the rate, the question that we address in this subsections is:

\zk{
\begin{quote}
\textit{
Q2: For a given bath $B$, regular sequences $\{\rho_{S^n}\}$ and $\{\sigma_{S^n}\}$ of the initial and final states of the product form, respectively, as well as real numbers $W_1,\ldots,W_c$ satisfying 
$\sum_j\beta_j W_j = -\Delta\widetilde{F}_S-\delta$, $\delta \geq 0$, 
what is the infimum over all rates $R = \frac{\ell}{k}$
such that there is a work transformation 
\[
  \rho_{S^n} \ox \tau(\und{\beta})_{B^{nR}} 
  \equiv \rho_{\widetilde{S}^\nu} \ox \tau(\und{\beta})_{\widetilde{B}}^{\ox \nu\ell} 
             \rightarrow 
             \sigma_{\widetilde{S}^\nu \widetilde{B}^{\nu\ell}}
             \equiv \sigma_{S^nB^{nR}},
\]
with the extracted works at rates $W_1,\ldots,W_c$ and the final state satisfying 
$\Tr_{\widetilde{B}^{\nu\ell}} \sigma_{\widetilde{S}^\nu \widetilde{B}^{\nu\ell}} 
= \sigma_{\widetilde{S}^\nu}$. 
}
\end{quote}
}

\bigskip
%
%
%
%
%

\zk{We first observe that if $S(\rho_{S^n})=S(\sigma_{S^n})$, then $\sum_j\beta_j W_j = -\Delta\widetilde{F}_S$ can hold without using any thermal bath, which follows from Eq.~(\ref{work expansion formula}).}
That is, the thermal bath is not necessary  for
the work transformation and extracting work if the entropy of the work 
system does not change. 
Conversely, the role of the thermal bath is precisely to facilitate changes 
of entropy in the work system.

To answer the above question about the minimum bath rate $R^*$, 
we first show the following lemma.

\begin{lemma}
Consider regular sequences of product states, 
$\rho_{S^n}$ and $\sigma_{S^n}$, and real numbers $W_j$, and assume 
that for large enough rate $R$ there is a work transformation 
$\rho_{S^n} \otimes \tau(\und{\beta})_B^{\ox nR} \rightarrow \sigma_{S^nB^{nR}}$,
with $\sigma_{S^n}$ as the reduced final state on the work system, 
and works $W_1,\ldots,W_c$ are extracted. Then there is another work transformation 
$\rho_{S^n} \otimes \tau(\und{\beta})_B^{\ox nR} \rightarrow \sigma_{S^n}\ox\xi_{B^{nR}}$,
in which the final state of the work system and the thermal bath are uncorrelated,
$\xi_{B^{nR}}$ is a regular sequence of product states, 
and the same works $W_1,\ldots,W_c$ are extracted.
\end{lemma}

\begin{proof}
Assuming that $\rho_{S^n} \otimes \tau(\und{\beta})_B^{\ox nR} \rightarrow \sigma_{S^nB^{nR}}$ is a work transformation, the second law implies that $\sum_j\beta_j W_j = -\Delta\widetilde{F}_s-\delta$ for some $\delta \geq 0$, and \zk{we obtain the following coordinates for the bath system for $0 \leq \delta' \leq \delta$:}
\begin{equation}\label{eq: coordinates_P_1}
\begin{split}
  s(\{\sigma_{B^{nR}}\})   &= S(\tau(\und{\beta})_B) - \frac1R \Delta s_S+\frac{\delta'}{R}, \\
  a_j(\{\sigma_{B^{nR}}\}) &= \Tr \tau(\und{\beta})_B A_{B_j} - \frac1R (\Delta A_{S_j} + W_j) 
                                                             \quad \text{for all } j=1,\ldots,c.
\end{split}\end{equation}
\zk{
To obtain the the first equality, which is the expansion of $\Delta s_S +\Delta s_B=\delta'$, we use the fact that $\sum_j \beta_j W_j 
                       = -\Delta\widetilde{F}_S\underbrace{-\Delta\widetilde{F}_B-\Delta s_S -\Delta s_B }_{-\delta}$, i.e. $\Delta\widetilde{F}_B+\Delta s_S +\Delta s_B =\delta$ which follows from Eq.~(\ref{work expansion formula}). Due to positivity of the entropy rate change, i.e. $\Delta s_S +\Delta s_B \geq 0$  from Eq.~(\ref{eq: positive Delta_SB}) and 
                $\Delta \widetilde{F}_B \geq 0$  from Eq.~(\ref{eq:positive Delta_F_B}), we infer that $0 \leq \underbrace{\Delta s_S +\Delta s_B}_{\delta'} \leq \delta$.
The second equality, which is the expansion  of $\Delta A_{B_j}+\Delta A_{S_j}=-W_j$, follows from the first law, Theorem \ref{thm:first-law}, 
and the AET, Theorem \ref{Asymptotic equivalence theorem}.}
If $R$ is large enough, due to the convexity of the phase diagram of the thermal bath $\overline{\mathcal{P}}_B^{(1)}$,  the following coordinates belong  to the phase diagram as well
\begin{equation}\label{eq: coordinates_P_2}
\begin{split}
  s(\{\xi_{B^{nR}}\})   &= S(\tau(\und{\beta})_B) - \frac1R \Delta s_S, \\
  a_j(\{\xi_{B^{nR}}\}) &= \Tr \tau(\und{\beta})_B A_{B_j} - \frac1R (\Delta A_{S_j} + W_j) 
                                                             \quad \text{for all } j=1,\ldots,c.
\end{split}\end{equation}
\zk{We can observe this in Fig.~\ref{fig:optimal-rate}; the new coordinates have the same charge values, but the entropy is $\frac{\delta'}{R}$ smaller than the entropy of the coordinates in Eq.~(\ref{eq: coordinates_P_1}). Therefore, as long as  $S(\tau(\und{\beta})_B) - \frac1R \Delta s_S \geq 0$, the new coordinates are inside the phase diagram as well.}
Hence, due to points 3 and 5 of Lemma~\ref{lemma:phase diagram properties}, 
there is a tensor product state $\xi_{B^{nR}}$ with coordinate of 
Eq.~(\ref{eq: coordinates_P_2}) on $\overline{\cP}_B^{(1)}$. Hence the first law, 
Theorem~\ref{thm:first-law}, implies that the desired transformation exists, and
works $W_1,\ldots,W_c$ are extracted.
\end{proof}

\medskip
\begin{theorem}
\label{thm:optimal-rate}
For regular sequences of product states, $\rho_{S^n}$ and $\sigma_{S^n}$, 
and real numbers $W_j$ satisfying $\sum_j\beta_j W_j = -\Delta\widetilde{F}_s-\delta$, 
let $R^*$ be the infimum of rates such that there is a work transformation 
$\rho_{S^n} \otimes \tau(\und{\beta})_B^{\ox nR} \rightarrow \sigma_{S^n}\ox\xi_{B^{nR}}$
under which works $W_1,\ldots,W_c$ are extracted, and 
$\xi_{B^{nR}}$ is a regular sequence of product states. 

Then, this minimum $R^*$ is achieved for a state $\xi_{B^{nR}}$ on the boundary of the 
phase diagram $\overline{\mathcal{P}}_B$ of the thermal bath. Indeed, it is the point 
where the line given by Eq.~(\ref{eq:R-bath-assumptions}) intersects the boundary of 
the phase diagram; see Fig.~\ref{fig:optimal-rate}.
Equivalently, it is the smallest $R$ such that the point in 
Eq.~(\ref{eq:R-bath-assumptions}) is contained in $\overline{\mathcal{P}}_B$. 

For $\delta \ll 1$, the minimum rate can be written as
\begin{equation}
  \label{eq:rate-vs-heatcapacity}
  R \approx -\frac{1}{2\delta} \sum_{ij} \frac{\partial\beta_j}{\partial a_i}(\Delta A_{S_i}+W_i)(\Delta A_{S_j}+W_j),
\end{equation}
where $\Delta A_{S_j} = a(\{\sigma_{S^n}\}) - a(\{\rho_{S^n}\})$.
\end{theorem}

\begin{figure}[ht]
  \begin{center}
    \includegraphics[width=10cm,height=8cm]{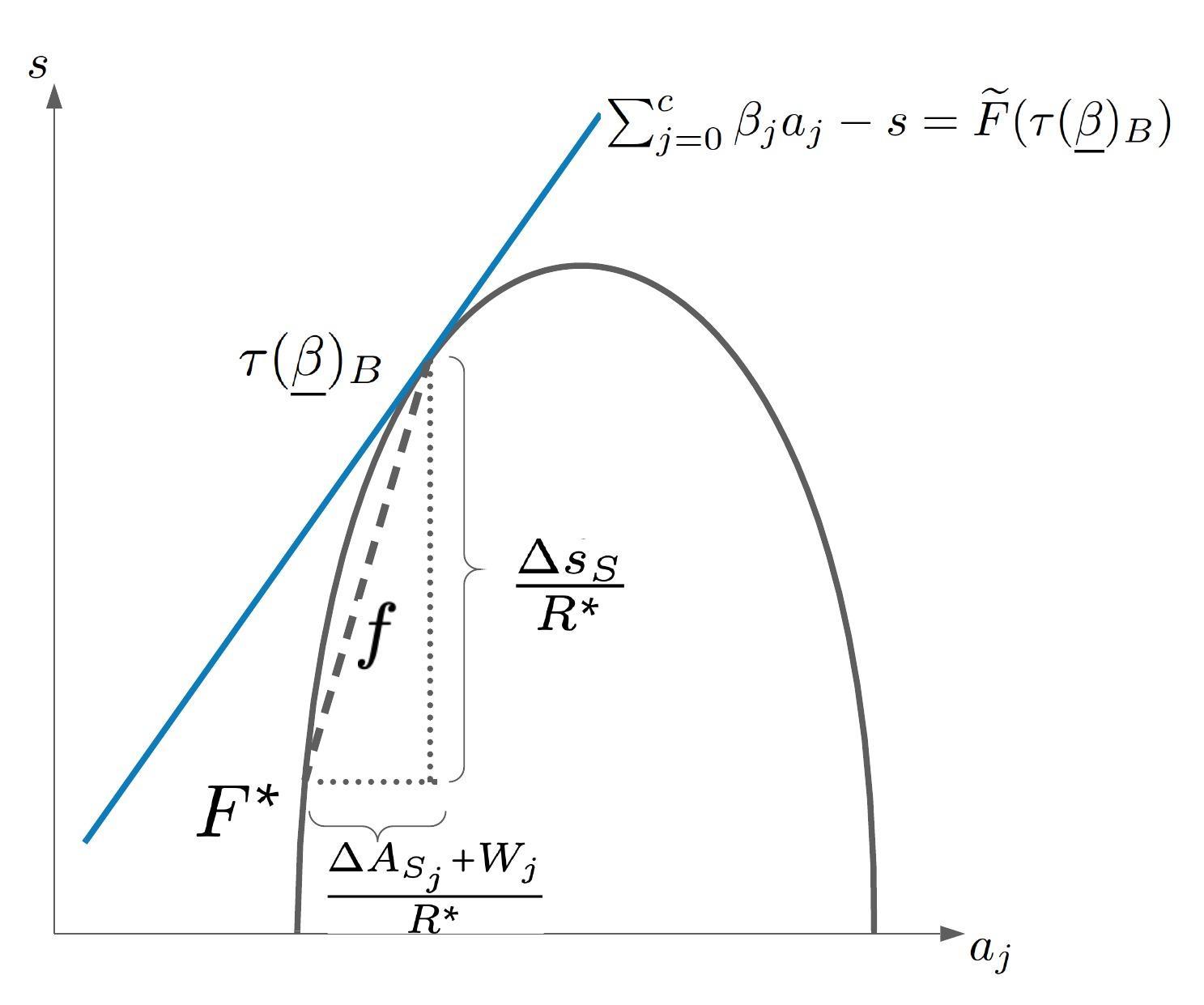}
  \end{center}
  \caption{Graphical illustration of $R^*$, the minimum bath rate for a 
           work transformation $\{\rho_{S^n}\} \rightarrow \{\sigma_{S^n}\}$ 
           satisfying the second law, according to Theorem~\ref{thm:optimal-rate}. 
           The initial state is the generalized thermal state $\tau(\und{\beta})$, its 
           corresponding point marked on the upper boundary of the phase diagram. The final 
           bath states correspond to points on the line denoted $f$, and they 
           are feasible if and only they fall into the phase diagram. 
           Consequently, $F^*$ is the point corresponding to the minimum rate.}
  \label{fig:optimal-rate}
\end{figure}

\begin{proof}
\zk{We notice that the initial and final states of the work system as well as the initial state of the bath and the extracted work rates are known.} Also, the final state of the thermal bath $\xi_{{B}^{ nR}}$ is a tensor product state, therefore,   
the first law (Theorem~\ref{thm:first-law}), 
and the AET( Theorem~\ref{Asymptotic equivalence theorem}) imply that \zk{the entropy and the charges rates of the global system are preserved, hence, we obtain the following entropy and charge rates for the final state of the bath:}
\begin{equation}\label{eq:coordinates} 
\begin{split}
  s(\{\xi_{B^{nR}}\})   &= S(\tau(\und{\beta})_B) - \frac1R \Delta s_S, \\
  a_j(\{\xi_{B^{nR}}\}) &= \Tr \tau(\und{\beta})_B A_{B_j} - \frac1R (\Delta A_{S_j} + W_j) 
                                                             \quad \text{for all } j=1,\ldots,c,
\end{split}\end{equation}
where $\Delta s_{S} = s(\{\sigma_{S^n}\}) - s(\{\rho_{S^n}\})$. 
\zk{The above quantities on the left member are rates of the entropy and charge changes, therefore, they must belong to the diagram $\frac{\overline{\mathcal{P}}_B^{(nR)}}{nR}$.}
Hence, due to point 3 of Lemma~\ref{lemma:phase diagram properties}, 
the above coordinates belong to $\overline{\mathcal{P}}_B^{(1)}=\frac{\overline{\mathcal{P}}_B^{(nR)}}{nR}$. 
%
Now, for $R=R^*$ assume that the above coordinates belong to the point $(\und{a},s)$ on the boundary 
of the phase diagram $\overline{\mathcal{P}}^{(1)}_B$. Then, for $R>R^*$ the point 
of Eq.~(\ref{eq:coordinates}) is a convex combination of the points 
$(\und{a},s)$ and the corresponding point of the state $\tau(\und{\beta})_B$, 
so it belongs to the phase diagram due to its convexity. Therefore, all points with $R>R^*$ 
are inside the diagram.

To approximate the minimum $R$ for small $\delta$, define the function 
$S(\und{a}):=S(\tau(\und{a})_B)$ for $\und{a}=(a_1,\ldots,a_c)$. 
Its Taylor expansion around the point corresponding to the
initial thermal state $\tau(\und{\beta})_B \equiv S\left(\tau(\und{a}^0)_B\right)$ of the bath
gives the approximation
\begin{equation}
  \label{eq:S-taylor}
  S(\und{a}) \approx S(\und{a}^0) + \sum_j \beta_j (a_j-a_j^0)
                                  + \frac12 \sum_{ij} \frac{\partial\beta_j}{\partial a_i}(a_j-a_j^0) (a_i-a_i^0),
\end{equation}
where we have used the well-know relation $\frac{\partial S}{\partial a_i}=\beta_i$. 
From Eq.~(\ref{eq:coordinates}), we obtain
\begin{align*}
  S(\und{a})-S(\und{a}^0) &= -\frac{\Delta s_S}{R},\\
  a_j-a_j^0               &=  \frac{1}{R} (-\Delta A_{S_j}-W_j),
\end{align*}
and by substituting these values in the Taylor approximation (\ref{eq:S-taylor}), 
using the definition of the free entropy and of the deficit $\delta$, 
we arrive at the claimed Eq. (\ref{eq:rate-vs-heatcapacity}).
\end{proof}

\medskip
\begin{remark}\normalfont
For a single charge, $c=1$, which we traditionally interpret as the internal 
energy $E$ of a system, Eq.~(\ref{eq:rate-vs-heatcapacity}) takes on the very 
simple form
\[
  R \approx -\frac{1}{2\delta} \frac{\partial\beta}{\partial E}(\Delta E_S+W)^2.
\]
Here we can use the usual thermodynamic definitions to rewrite 
$\frac{\partial\beta}{\partial E} = \frac{\partial\frac{1}{T}}{\partial E} = -\frac{1}{T^2}\frac{1}{C}$,
with the heat capacity $C = \frac{\partial E}{\partial T}$, all derivatives taken with
respect to corresponding Gibbs equilibrium states. Thus, 
\begin{equation}
  R \approx \frac{1}{T^2}\frac{1}{C}\cdot\frac{1}{2\delta}(\Delta E_S+W)^2,
\end{equation}
resulting in a clear operational interpretation of the heat capacity in terms 
of the rate of the bath to approach the second law tightly. 

For larger numbers of charges, the matrix 
$\bigl[ \frac{\partial\beta_j}{\partial a_i} \bigr]_{ij} 
  = \bigl[ \frac{\partial^2 S}{\partial a_i\partial a_j} \bigr]_{ij}$ is actually 
the Hessian of the entropy $S\bigl(\tau(\und{a})_B\bigr)$ with respect to the charges, 
and the r.h.s. side of Eq.~(\ref{eq:rate-vs-heatcapacity}) is $\frac{1}{2\delta}$
times the corresponding quadratic form evaluated on the vector
$(\Delta A_{S_1}+W_1,\ldots,\Delta A_{S_c}+W_c)$. Note that by the strict 
concavity of the generalized Gibbs entropy, this is a negative definite 
symmetric matrix, thus explaining the minus sign in Eq. (\ref{eq:rate-vs-heatcapacity}).
In the same vein as the single-parameter discussion before, the Hessian matrix can 
be read as being composed of generalized heat capacities, which likewise receive 
their operational interpretation in terms of the required rate of the bath. 
\end{remark}

\medskip
The heat capacity has made appearances in previous results in the resource approach 
to thermodynamics: Chubb \emph{et al.} \cite{Korzekwa:heatcap} have found it to 
show up in the optimal interconversion rate between states in a resource
theory of Gibbs-preserving transformations and with unlimited baths at temperature 
$T$. Their setting is the finite-copy regime, and in contrast to our result
of finite bath where the heat capacity affects the first order term (scaling 
linear with $n$), the heat capacity determines the so-called second order term, 
scaling with $\sqrt{n}$. While it is thus amusing to contemplate the separate 
appearance of the heat capacity in the two results, the settings seem too different 
to allow for a meaningful comparison.

\section{Discussion} 
\label{sec:discussion}
We have presented a resource theory in which the objects are sequences of tensor product states, 
and thermodynamically meaningful allowed transformations, 
namely operations which preserve the entropy and charges of a system asymptotically. 
The allowed operations classify the objects into equivalence classes of state sequences 
that are interconvertible under allowed operations. The basic result on which 
our theory is built is that the objects are interconvertible via allowed operations 
if and only if they have the same average entropy and average charge values in the 
asymptotic limit. 

The existence of the allowed operations between the objects of the same class is based on two pillars:
First, for objects with the same average entropy there are states with sublinear dimension which 
can be coupled to the objects to make their spectrum asymptotically identical.
Second, objects with the same average charge values project onto a common subspace of the 
charges of the system which has the property that any unitary acting on this subspace 
is an almost-commuting unitary with the corresponding charges. Therefore, the spectrum 
of the objects of the same class can be modified using small ancillary systems and then 
they are interconvertible via unitaries that asymptotically preserve the charges of the system.
The notion of a common subspace for different charges, which are Hermitian operators, 
is introduced in \cite{Halpern2016} as approximate microcanonical (a.m.c.) subspace.
In this paper, for given charges and parameters, we construct a permutation-symmetric 
a.m.c., something not guaranteed by the construction in \cite{Halpern2016}.

We then applied this resource theory to understand quantum thermodynamics with 
multiple conserved quantities. We specifically consider an asymptotic generalization of 
the setting proposed in \cite{Guryanova2016} where there are many copies of a global 
system consisting of a main system, called a work system, a thermal bath with fixed 
temperatures and various batteries to store the different charges of the system. 
Our approach allows us, in our setting, to resolve affirmatively a question 
from \cite{Guryanova2016,Halpern2016}, which asks about the possibility of constructing 
physically separate batteries for all the involved charge numbers, be they commuting 
or not (cf. \cite{Popescu2019}).
Therefore, the objects and allowed operations of the resource theory apply
quantum states of a thermodynamics system and thermodynamical transformations, 
respectively. 
It is evident that the allowed operations can transform a state with a tensor product 
structure to a state of a general form; however, we show that restricting the final 
states to the specific form of tensor product structure does not reduce the generality 
and tightness of the bounds that we obtain, which follows from the fact that for 
any point of the phase diagram there is a state with tensor product structure realizing it. 

As discussed in \cite{Guryanova2016}, for a system with multiple charges, the free entropy is
a conceptually more meaningful quantity than the free energy, which is originally defined 
when energy is the only conserved quantity of the system. Namely, the free energy bounds the 
amount of energy that can be extracted (while conserving the other charges); 
however, for a system with multiple charges there are not various quantities that 
bound the extraction of individual charges. 
Rather, there is only a bound on the trade-off between the charges that can be extracted 
which is precisely the free entropy defined with respect to the temperatures of the thermal bath.
We show that indeed this is the case in our scenario as well and formulate the second law: 
the amount of charge combination that is extracted is bounded by the free entropy change of 
the work system per number of copies of the work system, i.e. the free entropy rate change.
Conversely, we show that all transformations
with given extracted charge values,
with a combination strictly bounded by the free entropy rate change of the 
work system, are feasible.
In particular, any amount of a given charge, or the so-called work type, is extractable
providing that sufficient amounts of other charges are injected to the system.

This raises the following fundamental question: for given extractable charge values, with 
a combination saturating the second law up to a deficit $\delta$, what is the minimum number 
of the thermal baths per number of the copies of the work system.
We define this ratio as the thermal bath rate.
We find that for large thermal bath rates the optimal value is inversely 
proportional to the deficit $\delta$, and there is always a corresponding 
transformation where the final state of the work system and the thermal bath are uncorrelated.
However, in general this is not true: the minimum rate might be obtained where the final state correlates
the work system and the thermal bath.
This is a purely quantum mechanical effect, making certain work transformations possible 
with a smaller size of the thermal bath than would be possible classically; it 
relies on work system and bath becoming entangled.
In order to describe precisely the possible work transformations with a fixed bath, 
we define and analyze the extended phase diagram of the bath, which depends on a given 
conditional state of the work system and records the conditional, rather than plain, entropy. 

Our results paint a broad picture of thermodynamics as a resource theory,
which ultimately relies only on conservation laws, namely the conservation 
of information (i.e. entropy) -- cf.~\cite{Bera2017} --, and the conservation 
of extensive physical quantities. At the microscopic level, the former means that 
the allowed transformations are (approximate) unitaries, the latter that they 
(approximately) commute with the conserved quantities. Amazingly, after these
simple premises give rise to the phase diagram, the supposed deep distinction between 
entropy and the conserved charges disappears: they both are simply extensive conserved 
quantities. Following our development of thermodynamics, with its batteries for 
the distinct charges, we could augment this with an entropy battery, which 
carries no charges and is only there to absorb or release entropy. This is a
very general picture that in some respect includes as a special case our treatment 
of the second law: namely, the role of the bath is largely as such an entropy 
battery, although the fact that it also carries charges complicates things 
compared to this abstract vantage point.

\medskip
We leave several open questions to be addressed. Not to dwell on the overly 
technical ones, which will be evident to readers of the detailed claims and 
proofs, a fundamental problem is whether it is possible to prove the AET
Theorem~\ref{Asymptotic equivalence theorem} with unitaries that exactly commute 
with the conserved quantities, rather than approximately? This would require
the construction of a subexponential reference frame to take care of the 
conservation laws; this is known to be possible for a single conserved 
quantity (energy) \cite{Sparaciari2016}, and more generally for pairwise 
commuting charges \cite{Bera2017}. If it were possible in the non-commuting 
setting, it would give our theory a much stronger appeal, since at a fundamental 
level, conservation laws in nature are considered to hold strictly, rather than 
only approximately.

There is a whole plethora of open questions concerning practical and experimental 
implications of our results (similar experimental settings for thermodynamics with non-commuting charges have been characterized recently in \cite{Experiment1,Experiment2}). The most straightforward, and perhaps most interesting
is this one: can one design a system and a bath of small to moderate size, such 
that a concrete work transformation will necessarily leave the system and the bath in 
a final entangled state? The impossibility of such a transformation in a classical 
system could be interpreted as a thermal machine entanglement witness.

\acknowledgments
The authors thank Micha{\l} Horodecki for pointing out Ref. \cite{Korzekwa:heatcap}.
ZBK is grateful to Paul Skrzypczyk and Tony Short for valuable discussions. 
AW thanks Glen and Ella Runciter for invaluable hermetic and immanent remarks on 
the ubiquity of the second law. 

ZBK and AW acknowledge financial support by the Spanish MINECO 
(projects FIS2016-86681-P and PID2019-107609GB-I00) 
with the support of FEDER funds, and the Generalitat de Catalunya
(project CIRIT 2017-SGR-1127).
MNB acknowledges financial support from SERB-DST (CRG/2019/002199), Government of India.
ML acknowledges the support from ERC AdG NOQIA, Spanish Ministry of Economy and 
Competitiveness (Severo Ochoa program for Centres of Excellence in R\&{}D (CEX2019-000910-S), 
Plan National FISICATEAMO and FIDEUA PID2019-106901GB-I00/10.13039/501100011033, FPI), 
Fundaci\'{o} Privada Cellex, Fundaci\'{o} Mir-Puig, and from Generalitat de Catalunya 
(AGAUR grant no. 2017-SGR-1341, CERCA program, QuantumCAT U16-011424, co-funded by 
ERDF Operational Program of Catalonia 2014-2020), MINECO-EU QUANTERA MAQS 
(funded by State Research Agency (AEI) PCI2019-111828-2/10.13039/501100011033), 
EU Horizon 2020 FET-OPEN OPTOLogic (grant no. 899794), and the National Science Centre, 
Poland-Symfonia grant no. 2016/20/W/ST4/00314.


\appendix

\section*{Appendices}
Here we collect mathematically involved arguments that would detract from 
the exposition of the main article. These are the proof 
of Theorem~\ref{Asymptotic equivalence theorem} (in Appendix \ref{proof-AET}), which 
requires the construction of approximate microcanonical (a.m.c.) subspaces \cite{Halpern2016}, 
for which we give a self-contained proof in 
Appendix \ref{section: Approximate microcanonical (a.m.c.) subspace}.
We start with the collection of miscellaneous definitions and facts in
Appendix \ref{section: Miscellaneous definitions and facts}.

\section{Miscellaneous definitions and facts}
\label{section: Miscellaneous definitions and facts}

\begin{definition}\label{def:typicality}
Let $\rho_1,\ldots,\rho_n$ be quantum states on a $d$-dimensional Hilbert space $\mathcal{H}$ with diagonalizations $\rho_i=\sum_j p_{ij} \pi_{ij}$ and one-dimensional projectors $\pi_{ij}$. For $\alpha >0$ and $\rho^n=\rho_1 \otimes \cdots \otimes \rho_n$ define the set of entropy-typical sequences as
\begin{align} 
    \mathcal{T}_{\alpha,\rho^n }^n=\left\{j^n=j_1 j_2 \ldots j_n : \abs{\sum_{i=1}^n -\log p_{ i j_i}-S(\rho_i) } \leq \alpha \sqrt{n}\right\}. \nonumber
\end{align}
Define the entropy-typical projector of $\rho^n$ with constant $\alpha$ as
\begin{align*}
    \Pi^n_{\alpha ,\rho^n }=\sum_{j^n \in \mathcal{T}_{\alpha, \rho^n}^n} \pi_{1j_1} \otimes \cdots \otimes \pi_{nj_n}.
\end{align*}
\end{definition}

\begin{lemma}[{Cf.~\cite{csiszar_korner_2011}}]
\label{lemma:typicality properties}
There is a constant $0<\beta \leq \max \set{(\log 3)^2,(\log d)^2}$ 
such that the entropy-typical projector has the following properties for any $\alpha >0$, $n>0$ and arbitrary state $\rho^n=\rho_1 \otimes \cdots \otimes \rho_n$:
\begin{align*}
  \Tr \left(\rho^n \Pi^n_{\alpha ,\rho^n }\right) &\geq 1-\frac{\beta}{\alpha^2}, \\
  2^{-\sum_{i=1}^n S(\rho_i)-\alpha \sqrt{n}}  \Pi^n_{\alpha,\rho^n} 
                                &\leq \Pi^n_{\alpha,\rho^n}\rho^n \Pi^n_{\alpha,\rho^n} 
                                 \leq 2^{-\sum_{i=1}^n S(\rho_i)+\alpha \sqrt{n}}\Pi^n_{\alpha,\rho^n},\quad \text{and} \\
  \left(1-\frac{\beta}{\alpha^2}\right) 2^{ \sum_{i=1}^n S(\rho_i)-\alpha \sqrt{n}} 
                                &\leq  \Tr \left(\Pi^n_{\alpha ,\rho^n }\right) 
                                 \leq  2^{\sum_{i=1}^n S(\rho_i)+\alpha \sqrt{n}}.
\end{align*}
\end{lemma}

\begin{lemma}[Gentle operator lemma \cite{winter1999_2,Ogawa2007,wilde_2013}]
\label{Gentle Operator Lemma}
If a quantum state $\rho$ with diagonalization $\rho=\sum_j p_j \pi_j$ projects onto a
POVM element $\Lambda$ with probability  $1- \epsilon$, 
i.e. $\Tr(\rho \Lambda) \geq 1- \epsilon$ for $0 \leq \Lambda \leq \1$, then
\begin{align*}
  \sum_j p_j\norm{\pi_j-\sqrt{\Lambda}\pi_j \sqrt{\Lambda}}_1 \leq 2\sqrt{\epsilon}.
\end{align*}
\end{lemma}

\begin{lemma}[Cf.~Bhatia~\cite{bhatia97}]
\label{lemma:norm inequality}
For operators $A$, $B$ and $C$ and for any $p \in [1,\infty]$, the following holds 
\begin{align*}
    \norm{ABC}_p \leq \norm{A}_{\infty} \norm{B}_p \norm{C}_{\infty}.
\end{align*}
\end{lemma}

\begin{lemma}[Hoeffding's inequality, cf.~\cite{DemboZeitouni}]
\label{Hoeffding's inequality}
Let $X_1, X_2,\ldots,X_n$ be independent random variables with $a_i \leq X_i \leq b_i$,
and define the empirical mean of these variables as $\overline{X}=\frac{X_1+\ldots+X_n}{n}$.
Then, for any $t>0$,
\begin{align*}
  \Pr\left\{\overline{X}-\mathbb{E}(\overline{X}) \geq  t\right\} 
                  &\leq \exp\left(-\frac{2n^2t^2}{\sum_{i=1}^n(b_i-a_i)^2}\right),\\
  \Pr\left\{\overline{X}-\mathbb{E}(\overline{X}) \leq -t\right\} 
                  &\leq \exp\left(-\frac{2n^2t^2}{\sum_{i=1}^n(b_i-a_i)^2}\right).
\end{align*}
\end{lemma}

\section{Approximate microcanonical (a.m.c.) subspace}
\label{section: Approximate microcanonical (a.m.c.) subspace}
In this section, we recall the definition of the notion of 
approximate microcanonical (a.m.c.) and give a new proof that it exists 
for certain explicitly given parameters. 

\begin{definition}
\label{defi:microcanonical}
An \emph{approximate microcanonical (a.m.c.) subspace}, or more precisely
a \emph{$(\epsilon,\eta,\eta',\delta,\delta')$-approximate microcanonical subspace},
$\cM$ of $\cH^{\ox n}$, with projector $P$, 
for charges $A_j$ and values $v_j = \langle A_j \rangle$
is one that consists, in a certain precise sense, of exactly the
states with ``very sharp'' values of all the $A_j^{(n)}$. 
Mathematically, the following has to hold:
\begin{enumerate}
  \item Every state $\omega$ with support contained in $\cM$ satisfies
        $\tr \omega\Pi^\eta_j \geq 1-\delta$ for
        all $j$.
  \item Conversely, every state $\omega$ on $\cH^{\ox n}$ such that
        $\tr \omega\Pi^{\eta'}_j \geq 1-\delta'$ for
        all $j$, satisfies $\tr\omega P \geq 1-\epsilon$.
\end{enumerate}
Here, $\Pi^\eta_j := \bigl\{ nv_j-n\eta\Sigma(A_j) \leq A_j^{(n)} \leq nv_j+n\eta\Sigma(A_j) \bigr\}$
is the spectral projector of $A_j^{(n)}$ of values close to $n v_j$,
and $\Sigma(A) = \lambda_{\max}(A)-\lambda_{\min}(A)$ is the spectral diameter
of the Hermitian $A$, i.e.~the diameter of the smallest disc
covering the spectrum of $A$.
\end{definition}

\medskip

\begin{remark}\normalfont
It is shown in \cite[Thm.~3]{Halpern2016} that for every $\epsilon > c\delta' > 0$,
$\delta > 0$ and $\eta > \eta' > 0$, and for all sufficiently
large $n$, there exists a nontrivial 
$(\epsilon,\eta,\eta',\delta,\delta')$-a.m.c.~subspace.
However, there are two (related) reasons why one might be not completely
satisfied with the argument in~\cite{Halpern2016}: First, the proof uses
a difficult result of Ogata~\cite{Ogata} to reduce the non-commuting
case to the seemingly easier of commuting observables; while this
is conceptually nice, it makes it harder to perceive the nature
of the constructed subspace. Secondly, despite the fact that the
defining properties of an a.m.c.~subspace are manifestly permutation symmetric
(w.r.t.~permutations of the $n$ subsystems), the resulting construction
does not necessarily have this property.

Here we address both these concerns. Indeed, we shall show by 
essentially elementary means how to obtain an a.m.c.~subspace
that is by its definition permutation symmetric.
\end{remark}

\begin{theorem}
  \label{thm:symmetric-micro-exists}
  Under the  assumptions of Definition~\ref{defi:microcanonical}, for every $\epsilon > 2(n+1)^{3d^2}\delta' > 0$,
  $\eta >\eta' > 0$ and $\delta>0$, for all sufficiently large $n$
  there exists an approximate microcanonical subspace projector.
  In addition, the subspace can be chosen to be stable under permutations
  of the $n$ systems: $U^\pi\cM = \cM$, or equivalently $U^\pi P (U^\pi)^\dagger = P$,
  for any permutation $\pi\in S_n$ and its unitary action $U^\pi$.

  More precisely, given $\eta > \eta' > 0$ and $\epsilon > 0$, there exists a $\alpha > 0$ such
  that there is a non-trivial $(\epsilon,\eta,\eta',\delta,\delta')$-a.m.c. subspace
  with
  \begin{align*}
    \delta  &= (c+3)(5n)^{5d^2} e^{-\alpha n} \text{ and } \\ 
    \delta' &= \frac{\epsilon}{2(n+1)^{3d^2}} - (c+3)(5n)^{2d^2} e^{-\alpha n}.
  \end{align*}
  Furthermore, we may choose $\alpha = \frac{(\eta-\eta')^2}{8(cd+1)^2}$.
\end{theorem}
\begin{proof}
For $s>0$, partition the state space $\cS(\cH)$ on $\cH$ into 
\begin{align}
  \cC_s(\underline{v}) &= \bigl\{ \sigma : \forall j\ |\tr\sigma A_j - v_j| \leq s\Sigma(A_j) \bigr\}, \label{def:C_s} \\
  \cF_s(\underline{v}) &= \bigl\{ \sigma : \exists j\ |\tr\sigma A_j - v_j| >    s\Sigma(A_j) \bigr\}
                        = \cS(\cH) \setminus \cC_s(\underline{v}), \label{def:F_s}
\end{align}
which are the sets of states with $A_j$-expectation values ``close'' 
to and ``far'' from $\underline{v}$. 
Note that if $\rho \in \cC_s(\underline{v})$ and 
$\sigma \in \cF_t(\underline{v})$, $0 < s < t$, then $\|\rho-\sigma\|_1 \geq t-s$.

Choosing the precise values of $s>\eta'$ and $t<\eta$ later, 
we pick a universal distinguisher $(P,P^\perp)$
between $\cC_s(\underline{v})^{\ox n}$ and $\cF_t(\underline{v})^{\ox n}$,
according to Lemma~\ref{lemma:universal-test} below:
\begin{align}
  \label{eq:s-close}
  \forall \rho\in\cC_s(\underline{v})\   \tr\rho^{\ox n} P^\perp &\leq (c+2)(5n)^{2d^2} e^{-\zeta n}, \\
  \label{eq:t-far}
  \forall \sigma\in\cF_t(\underline{v})\ \tr\sigma^{\ox n} P     &\leq (c+2)(5n)^{2d^2} e^{-\zeta n},
\end{align}
with $\zeta = \frac{(t-s)^2}{2c^2(2d^2+1)}$.
Our a.m.c.~subspace will be $\cM := \operatorname{supp} P$;
by Lemma~\ref{lemma:universal-test}, $P$ and likewise $\cM$ are permutation symmetric.

It remains to check the properties of the definition. First, let $\omega$
be supported on $\cM$. Since we are interested in $\tr\omega \Pi_j^\eta$, we may
without loss of generality assume that $\omega$ is permutation symmetric.
Thus, by the ``constrained de Finetti reduction'' (aka ``Postselection 
Lemma'')~\cite[Lemma~18]{Duan2016},
\begin{align}
  \label{eq:dF}
  \omega \leq (n+1)^{3d^2} \int {\rm d}\sigma\,\sigma^{\ox n} F(\omega,\sigma^{\ox n})^2,
\end{align}
with a certain universal probability measure ${\rm d}\sigma$ on $\cS(\cH)$, 
and the fidelity $F(\rho,\sigma) = \|\sqrt{\rho}\sqrt{\sigma}\|_1$ between
states. We need the monotonicity of the fidelity under cptp maps, which we
apply to the test $(P,P^\perp)$:
\[
  F(\omega,\sigma^{\ox n})^2 \leq F\bigl( (\tr\sigma^{\ox n}P,1-\tr\sigma^{\ox n}P),(1,0) \bigr)^2
                             \leq \tr\sigma^{\ox n}P,
\]
which holds because $\tr\omega P = 1$.
Thus,
\begin{align}
  \label{eq:dF-P}
  \tr\omega(\Pi_j^\eta)^\perp \leq (n+1)^{3d^2} \int {\rm d}\sigma\,\bigl(\tr\sigma^{\ox n}(\Pi_j^\eta)^\perp\bigr) 
                                                                                               (\tr\sigma^{\ox n}P).
\end{align}

Now we split the integral on the right hand side of Eq.~(\ref{eq:dF-P}) into two parts,
$\sigma\in\cC_t(\underline{v})$ and $\sigma\not\in\cF_t(\underline{v})$:
If $\sigma\in\cF_t(\underline{v})$, then by Eq.~(\ref{eq:t-far}) 
we have
\[
  \tr\sigma^{\ox n}P \leq (c+2)(5n)^{2d^2} e^{-\zeta n}.
\]
On the other hand, if $\sigma\in\cC_t(\underline{v})$, then because of $t < \eta$ we have
\[
  \tr\sigma^{\ox n}(\Pi_j^\eta)^\perp \leq 2 e^{-2(\eta-t)^2 n},
\]
which follows from Hoeffding's inequality~\cite{DemboZeitouni}: 
Indeed, let $Z_\ell$ be the i.i.d.~random variables obtained by the
measurement of $A_j$ on the state $\sigma$. They take values in
the interval $[\lambda_{\min}(A_j),\lambda_{\max}(A_j)]$, their expectation 
values satisfy $\EE Z_j = \tr\sigma A_j \in [v_j \pm t\Sigma(A_j)]$, while
\[\begin{split}
  \tr\sigma^{\ox n}(\Pi_j^\eta)^\perp &= \Pr\left\{\frac1n\sum_\ell Z_\ell \not\in[v_j \pm \eta\Sigma(A_j)\right\} \\
                     &\leq \Pr\left\{\frac1n\sum_\ell Z_\ell \not\in[\tr\sigma A_j \pm (\eta-t)\Sigma(A_j)\right\},
\end{split}\]
so Hoeffding's inequality applies.
All taken together, we have
\[\begin{split}
  \tr\omega(\Pi_j^\eta)^\perp 
      &\leq (n+1)^{3d^2} \left( (c+2)(5n)^{2d^2} e^{-\zeta n} +  2 e^{-2(\eta-t)^2 n} \right) \\
      &\leq (c+3)(5n)^{5d^2} e^{-2(\eta-t)^2 n},
\end{split}\]
because we can choose $t$ such that
\begin{equation}
  \label{eq:t}
  \eta-t = \frac{t-s}{2c\sqrt{2d^2+1}} \geq \frac{t-s}{4cd}.
\end{equation}

Secondly, let $\omega$ be such that $\tr\omega \Pi_j^\eta \geq 1-\delta'$;
as we are interested in $\tr\omega P$, we may again assume
without loss of generality that $\omega$ is permutation symmetric,
and invoke the constrained de Finetti reduction~\cite[Lemma~18]{Duan2016},
Eq.~(\ref{eq:dF}). 
From that we get, much as before,
\[
  \tr\omega P^\perp \leq (n+1)^{3d^2} \int {\rm d}\sigma\, (\tr\sigma^{\ox n}P^\perp) F(\omega,\sigma^{\ox n})^2,
\]
and we split the integral on the right hand side into two parts,
depending on $\sigma\in\cF_s(\underline{v})$ or
$\sigma\in\cC_s(\underline{v})$: In the latter case, 
$\tr\sigma^{\ox n}P^\perp \leq (c+2)(5n)^{2d^2} e^{-\zeta n}$,
by Eq.~(\ref{eq:s-close}). In the former case, there exists a $j$ such
that $\tr\sigma A_j = w_j \not\in [v_j \pm s\Sigma(A_j)]$, and so
\[\begin{split}
  F(\omega,\sigma^{\ox n})^2 
      &\leq F\bigl( (1-\delta',\delta'), (\tr\sigma^{\ox n}\Pi_j^{\eta'},1-\tr\sigma^{\ox N}\Pi_j^{\eta'}) \bigr) \\
      &\leq \left( \sqrt{\delta'} + \sqrt{\tr\sigma^{\ox n}\Pi_j^{\eta'}} \right)^2                          \\
      &\leq 2 \delta' + 2 \tr\sigma^{\ox n}\Pi_j^{\eta'}                             \\
      &\leq 2 \delta' + 4 e^{-2(s-\eta')^2 n},
\end{split}\]
the last line again by Hoeffding's inequality; indeed, with the previous notation,
\[\begin{split}
  \tr\sigma^{\ox n}\Pi_j^{\eta'} &= \Pr\left\{ \frac1n \sum_\ell Z_\ell \in[v_j \pm \eta'\Sigma(A_j) \right\} \\
                      &\leq \Pr\left\{ \frac1n \sum_\ell Z_\ell \not\in[w_j \pm (s-\eta')\Sigma(A_j) \right\}.
\end{split}\]
All taken together, we get
\[\begin{split}
  \tr\omega P^\perp
      &\leq (n+1)^{3d^2} \left( (c+2)(5n)^{2d^2} e^{-\zeta n} +  4 e^{-2(s-\eta')^2 n} + 2 \delta' \right) \\
      &\leq (n+1)^{3d^2} (c+3)(5n)^{2d^2} e^{-2(s-\eta')^2 n} + 2(n+1)^{3d^2}\delta',
\end{split}\]
because we can choose $s$ such that
\begin{equation}
  \label{eq:s}
  s-\eta' = \frac{t-s}{2c\sqrt{2d^2+1}} \geq \frac{t-s}{4cd}.
\end{equation}

From eqs.~(\ref{eq:t}) and (\ref{eq:s}) we get by summation
\[
  \eta-\eta' = t-s + \frac{t-s}{c\sqrt{2d^2+1}} \leq (t-s)\left( 1+\frac{1}{cd} \right),
\]
from which we obtain
\[
  s-\eta' = \eta-t \geq \frac{\eta-\eta'}{4(cd+1)},
\]
concluding the proof.
\end{proof}

\medskip
\begin{lemma}
  \label{lemma:universal-test}
  For all $0 < s < t$ there exists $\zeta > 0$, such that
  for all $n$ there exists a permutation symmetric 
  projector $P$ on $\cH^{\ox n}$ with the properties
  \begin{align}
    \forall \rho\in\cC_s(\underline{v})\   \tr\rho^{\ox n} P^\perp &\leq (c+2)(5n)^{2d^2} e^{-\zeta n}, \\
    \forall \sigma\in\cF_t(\underline{v})\ \tr\sigma^{\ox n} P     &\leq (c+2)(5n)^{2d^2} e^{-\zeta n},
  \end{align}
  where $\cC_s(\underline{v})$ and $\cF_t(\underline{v})$ are defined in Eq.~(\ref{def:C_s}) and Eq.~(\ref{def:F_s}), respectively. 
  The constant $\zeta$ may be chosen as
  $\zeta = \frac{(t-s)^2}{2c^2(2d^2+1)}$.
\end{lemma}
\begin{proof}
We start by showing that there is a POVM $(M,\1-M)$ with 
\begin{align}
  \forall \rho\in\cC_s(\underline{v})\            \tr\rho^{\ox n} (\1-M) &\leq c e^{-\frac{(t-s)^2}{2c^2}n}, \\
  \forall \sigma\in\cF_t(\underline{v})\ \ \qquad \tr\sigma^{\ox n} M    &\leq       e^{-\frac{(t-s)^2}{2c^2}n}.
\end{align}
Namely, for each $\ell=0,\ldots,n$ choose $j_\ell \in \{1,\ldots,c\}$ uniformly
at random and measure $A_{j_\ell}$ on the $\ell$-th system. Denote the outcome
by the random variable $Z_\ell^{j_{\ell}}$ and let $Z_\ell^j = 0$ for $j\neq j_\ell$.
Thus, for all $j$, the random variables $Z_\ell^j$ are i.i.d.~with
mean $\EE Z_\ell^j = \frac{1}{c}\tr\rho A_j$, if the measured state is $\rho^{\ox n}$.

Outcome $M$ corresponds to the event
\[
  \forall j\ \frac1n \sum_\ell Z_\ell^j \in \frac{1}{c}\left[v_j \pm \frac{s+t}{2}\Sigma(A_j)\right];
\]
outcome $\1-M$ corresponds to the complementary event
\[
  \exists j\ \frac1n \sum_\ell Z_\ell^j \not\in \frac{1}{c}\left[v_j \pm \frac{s+t}{2}\Sigma(A_j)\right].
\]
We can use Hoeffding's inequality to bound the traces in question.\\
For $\rho\in \cC_s(\underline{v})$, we have $|\EE Z_\ell^j - v_j | \leq \frac{s}{c}\Sigma(A_j)$
for all $j$, and so:
\[\begin{split}
  \tr\rho^{\ox n}(\1-M) &=    \Pr\left\{ \exists j\ \frac1n \sum_\ell Z_\ell^j 
                                          \not\in \frac{1}{c}\left[v_j \pm \frac{s+t}{2}\Sigma(A_j)\right] \right\} \\
                        &\leq \sum_{j=1}^c \Pr\left\{ \frac1n \sum_\ell Z_\ell^j 
                                          \not\in \frac{1}{c}\left[v_j \pm \frac{s+t}{2}\Sigma(A_j)\right] \right\} \\
                        &\leq \sum_{j=1}^c \Pr\left\{ \frac1n \sum_\ell Z_\ell^j 
                                          \not\in \frac{1}{c}\left[v_j \pm \frac{s+t}{2}\Sigma(A_j)\right] \right\} \\
                        &\leq \sum_{j=1}^c \Pr\left\{ \left| \frac1n \sum_\ell Z_\ell^j - \EE Z_1^j \right|
                                                                     > \frac{t-s}{2c} \Sigma(A_j) \right\} \\
                        &\leq c e^{-\frac{(t-s)^2}{2c^2}n}.
\end{split}\]
For $\sigma\in\cF_t(\underline{v})$, there exists a $j$ such that
$|\EE Z_\ell^j - v_j | > \frac{t}{c}\Sigma(A_j)$. Thus,
\[\begin{split}
  \tr\sigma^{\ox n} M     &\leq \Pr\left\{ \frac1n \sum_\ell Z_\ell^j 
                                            \in \frac{1}{c}\left[v_j \pm \frac{s+t}{2}\Sigma(A_j)\right] \right\} \\
                          &\leq \Pr\left\{ \left| \frac1n \sum_\ell Z_\ell^j - \EE Z_1^j \right|
                                                                   > \frac{t-s}{2c} \Sigma(A_j) \right\} \\
                          &\leq e^{-\frac{(t-s)^2}{2c^2}n}.
\end{split}\]

This POVM is, by construction, permutation symmetric, but $M$ is not a 
projector. To fix this, choose $\lambda$-nets $\cN_C^\lambda$ in $\cC_s(\underline{v})$
and $\cN_F^\lambda$ in $\cF_t(\underline{v})$, with
$\lambda = e^{-\zeta n}$, with $\zeta = \frac{(t-s)^2}{2c^2(2d^2+1)}$.
This means that every state $\rho\in\cC_s(\underline{v})$
is no farther than $\lambda$ in trace distance from a $\rho'\in\cN_C^\lambda$,
and likewise for $\cF_t(\underline{v})$.
By~\cite[Lemma~III.6]{Hayden2006} (or rather, a minor variation of its proof), 
we can find such nets with 
$|\cN_C^\lambda|,\ |\cN_F^\lambda| \leq \left( \frac{5n}{\lambda} \right)^{2d^2}$
elements.
Form the two states
\begin{align*}
  \Gamma &:= \frac{1}{|\cN_C^\lambda|} \sum_{\rho\in\cN_C^\lambda} \rho^{\ox n}, \\
  \Phi   &:= \frac{1}{|\cN_F^\lambda|} \sum_{\sigma\in\cN_F^\lambda} \sigma^{\ox n},
\end{align*}
and let
\[
  P := \{ \Gamma-\Phi \geq 0 \}
\]
be the Helstrom projector which optimally distinguishes $\Gamma$ from $\Phi$.
But we know already a POVM that distinguishes the two states, hence
$(P,P^\perp=\1-P)$ cannot be worse:
\[
  \tr\Gamma P^\perp + \tr\Phi P \leq \tr\Gamma (\1-M) + \tr\Phi M
                                \leq (c+1) e^{-\frac{(t-s)^2}{2c^2}n},
\]
thus for all $\rho\in\cN_C^\lambda$ and $\sigma\in\cN_F^\lambda$,
\[
  \tr\rho^{\ox n}P^\perp,\ \tr\sigma^{\ox n}P 
      \leq (c+1)\left( \frac{5n}{\lambda} \right)^{2d^2} e^{-\frac{(t-s)^2}{2c^2}n}.
\]
So, by the $\lambda$-net property, we find 
for all $\rho\in\cC_s(\underline{v})$ and $\sigma\in\cF_t(\underline{v})$,
\[
  \tr\rho^{\ox n}P^\perp,\ \tr\sigma^{\ox n}P 
      \leq \lambda + (c+1)\left( \frac{5n}{\lambda} \right)^{2d^2} e^{-\frac{(t-s)^2}{2c^2}n}
      \leq (c+2)(5n)^{2d^2} e^{-\zeta n},
\]
by our choice of $\lambda$.
\end{proof}

\begin{corollary}
\label{corollary: a.m.c. projection}
For charges $A_j$, values $v_j = \langle A_j \rangle$ and $n>0$, 
Theorem~\ref{thm:symmetric-micro-exists} implies that there is an a.m.c.~subspace 
$\mathcal{M}$ of $\mathcal{H}^{\otimes n}$ for any $\eta' > 0$, with the following parameters:
\begin{align*}
  \eta     &=2 \eta',\\
  \delta'  &=\frac{c+3}{2}(5n)^{2d^2} e^{-\frac{n \eta'^2}{8c^2(d+1)^2}}, \\ 
  \delta   &=(c+3)(5n)^{2d^2} e^{-\frac{n \eta'^2}{8c^2(d+1)^2}},\\ 
  \epsilon &=2(c+3)(n+1)^{3d^2}(5n)^{2d^2} e^{-\frac{n \eta'^2}{8c^2(d+1)^2}}. 
\end{align*}
Moreover, let $\rho^n=\rho_1 \otimes \cdots \otimes \rho_n$ be a tensor product
state with $\frac{1}{n}\abs{\Tr(\rho^n A_j^{(n)})- v_j}\leq \frac{1}{2} \eta' \Sigma(A_j)$
for all $j$. Then, $\rho^n$ projects onto the a.m.c. subspace with probability $\epsilon$: 
$\Tr(\rho^n P) \geq 1- \epsilon$.
\end{corollary}

\begin{proof}
For simplicity of notation we drop the subscript $j$ from $A_j$, $v_j$ and $\Pi^{\eta'}_j$, 
so let $\sum_{\ell=1}^d E_\ell \proj{\ell}$ be the spectral decomposition of $A$. 
Define independent random variables $X_i$ for $i=1,\ldots,n$ taking values in the set 
$\{E_1,\ldots,E_d\}$ with probabilities $\Pr\{X_i=E_\ell\} = p_i(E_\ell)=\Tr \rho_i \proj{\ell}$.
Furthermore, define the random variable $\overline{X}=\frac1n (X_1+\ldots+X_n)$ which has the 
expectation value 
\begin{align*}
  \mathbb{E}(\overline{X})=\frac{1}{n} \Tr \rho^n A^{(n)}.
\end{align*}
Therefore, we obtain
\begin{align*}
  1-\Tr \rho^n \Pi^{\eta'}
    &=\sum_{\substack{\ell_1,\ldots,\ell_n: \\ 
            \abs{E_{\ell_1}+\ldots+E_{\ell_n}-n v} \geq n \eta' \Sigma(A)}} 
            \bra{\ell_1} \rho_1 \ket{\ell_1} \cdots \bra{\ell_n} \rho_n \ket{\ell_n} \\
    &=\Pr \left\{ \abs{\overline{X}- v} \geq \eta' \Sigma(A) \right\} \\
    &=\Pr \left\{ \overline{X}- \mathbb{E}(\overline{X})  \geq \eta'\Sigma(A)+v - \mathbb{E}(\overline{X})  
                  \text{ or } 
                  \overline{X}- \mathbb{E}(\overline{X})  \leq -\eta' \Sigma(A)+v - \mathbb{E}(\overline{X}) \right\} \\
    &\leq \exp \left(-\frac{2n (\eta'\Sigma(A)+v - \mathbb{E}(\overline{X}))^2 }{(\Sigma(A))^2} \right)+\exp \left(-\frac{2n (\eta'\Sigma(A)-v + \mathbb{E}(\overline{X}))^2 }{(\Sigma(A))^2} \right)\\
    &\leq 2 \exp (-\frac{n \eta'^2}{2}) 
     \leq \delta',
\end{align*}
where the second line follows because random the variables $X_1,\ldots,X_n$ are independent, and as a result 
$\Pr\left\{X_i=E_{\ell_i} \,\forall i=1,\ldots,n\right\} 
 = \bra{\ell_1}\rho_1\ket{\ell_1} \cdots \bra{\ell_n}\rho_n\ket{\ell_n}$;
the fourth line is due to Hoeffding's inequality, Lemma~\ref{Hoeffding's inequality}; 
the fifth line is due to assumption $\abs{\mathbb{E}(\overline{X})- v} \leq \frac{1}{2} \eta' \Sigma(A)$.
Thus, by the definition of the a.m.c. subspace, $\Tr \rho^n P \geq 1- \epsilon$.
\end{proof}

\section{Proof of the AET Theorem~\ref{Asymptotic equivalence theorem}}
\label{proof-AET}

\zk{
In this section, we first review the notion of the entropy-typical subspace defined in Definition \ref{def:typicality}, which we refer to it as the typical subspace for simplicity.  
Lemma~\ref{lemma:typicality properties} summarizes the properties of this subspace which we use in the proofs of this section. Intuitively, 
the  typical subspace of the support of a tensor product state $\rho^n=\rho_1 \otimes \cdots \otimes \rho_n$ with projector $\Pi^n_{\alpha,\rho^n}$, for a positive constant $\alpha$, is a high probability subspace for $\rho^n$ of dimension approximately equal to $2^{S(\rho^n)}=2^{\sum_{i=1}^n S(\rho_i)}$. Moreover, eigenvalues of $\rho^n$ inside this subspace belong to  a tight interval around $2^{-S(\rho^n)}$ with the radius of $2^{-\alpha \sqrt{n}}$.
We use these properties to prove Lemma~\ref{lemma: timmed state} of which we will use 
points 3 and 4 to prove the AET. 
In this lemma, we  show that if $\rho^n$ projects onto a subspace $\mathcal{M}$ with high probability, then 
we can find a state  $\widetilde{\rho}$ inside this subspace which is close to the state $\rho^n$ (in trace distance) and has useful properties. In particular, similar to the typicality properties, the eigenvalues of $\widetilde{\rho}$ belong to an interval around $2^{-S(\rho^n)}$ with a small radius. We use this to show that the state $\widetilde{\rho}$ can be decomposed as the tensor product of a maximally mixed state of dimension almost equal to $2^{-S(\rho^n)}$ and another state with significantly smaller dimension.

}
\begin{lemma}
\label{lemma: timmed state}
Let $\mathcal{M} \subset \mathcal{H}^{\otimes n}$ with projector $P$ be a 
high-probability subspace for the state $\rho^n=\rho_1 \otimes \cdots \otimes \rho_n$, 
i.e. $\Tr \rho^n P \geq 1- \epsilon$.
Then, for $\alpha>0$ and all sufficiently large $n$,
there exist a subspace $\widetilde{\mathcal{M}} \subseteq \mathcal{M}$ 
with projector $\widetilde{P}$, and a state $\widetilde{\rho}$ with support 
in $\widetilde{\mathcal{M}}$, such that the following holds:  
\begin{enumerate}
    \item $\Tr \Pi^n_{\alpha,\rho^n}\rho^n \Pi^n_{\alpha,\rho^n} \widetilde{P}  
             \geq 1- 2\sqrt{\epsilon}-O\left(\frac{1}{\alpha}\right)$.
    \item $2^{-\sum_{i=1}^n S(\rho_i)-2\alpha \sqrt{n}}\widetilde{P} 
             \leq \widetilde{P} \Pi^n_{\alpha,\rho^n}\rho^n \Pi^n_{\alpha,\rho^n} \widetilde{P} 
             \leq 2^{-\sum_{i=1}^n S(\rho_i)+ \alpha \sqrt{n}}\widetilde{P}$.
    \item There is a unitary $U$ such that $U\widetilde{\rho}U^{\dagger}=\tau \otimes \omega$,
          where $\tau$ is a maximally mixed state of rank $2^{\sum_{i=1}^n S(\rho_i) - O(\alpha \sqrt{n})}$, 
          and $\omega$ is a  state of dimension $2^{O(\alpha \sqrt{n})}$.
    \item $\norm{\widetilde{\rho}-\rho^n}_1 
             \leq 2\sqrt{\epsilon} + O\left(\frac{1}{\alpha}\right) 
                                   + 2\sqrt{2\sqrt{\epsilon}+O\left(\frac{1}{\alpha}\right)}$.
\end{enumerate}
\end{lemma}

\begin{proof} 
\zk{
In point 1 and 2 of the lemma, we first construct the subspace $\mathcal{\widetilde{M}}$ with projector $\widetilde{P}$. To this end, we project the typical subspace of $\rho^n$ with projector $\Pi^n_{\alpha,\rho^n}$ onto the space $\mathcal{M}$, i.e. $P\Pi^n_{\alpha,\rho^n} P$, and define $\widetilde{P}$ as a projector onto the support of  $P\Pi^n_{\alpha,\rho^n} P$ with corresponding eigenvalues bigger than $2^{-\alpha \sqrt{n}}$.
Since $\rho^n$ project onto $\mathcal{M}$ with high probability, 
 therefore the unnormalized state $\Pi^n_{\alpha,\rho^n}\rho^n  \Pi^n_{\alpha,\rho^n} \approx \rho^n$ projects onto $\mathcal{M}$ with high probability as well. 
We use this to show that  they both project onto $\mathcal{\widetilde{M}}$ with high probability.
Moreover, from Lemma~\ref{lemma:typicality properties}, we know that
the eigenvalues of the unnormalized state $\Pi^n_{\alpha,\rho^n}\rho^n \Pi^n_{\alpha,\rho^n} $ are inside a tight interval around $2^{-\sum_{i=1}^n S(\rho_i)}$. We use this to  show that the new unnormalized state $\widetilde{P} \Pi^n_{\alpha,\rho^n}\rho^n \Pi^n_{\alpha,\rho^n} \widetilde{P}$ has the same property.

In point 3 of the lemma, we further \textit{trim} the unnormalized state $\widetilde{P} \Pi^n_{\alpha,\rho^n}\rho^n \Pi^n_{\alpha,\rho^n} \widetilde{P}$ to obtain a new state which has a degeneracy of the order of multiples of $2^{-\sum_{i=1}^n S(\rho_i)-10\alpha\sqrt{n}}$. 
By trimming, we mean discarding some parts of an unnormalized  state in such a way that the trace of the new unnormalized  state is almost the same. 
We use this property to decompose the state into the tensor product of a maximally mixed state and another state of smaller dimension. Lastly, in point 4, we show that the new state is close to state $\rho^n$.

}

1. Let $E\geq 0$ and $F\geq 0$ be two positive operators such that 
$E+F=P \Pi^n_{\alpha,\rho^n} P$, where all eigenvalues of $F$ are smaller than $2^{-\alpha \sqrt{n}}$, 
and define $\widetilde{P}$ to be the projection onto the support of $E$. 
In other words, $\widetilde{P}$ is the projection onto the support of $P \Pi^n_{\alpha ,\rho^n } P$ 
with corresponding eigenvalues greater $2^{-\alpha\sqrt{n}}$. 
\zk{Also, notice that all eigenvalues of $E$ and $F$ are smaller than 1.
Thus, we  obtain
\begin{align*}
    \Tr(\Pi^n_{\alpha,\rho^n}\rho^n \Pi^n_{\alpha,\rho^n}  \widetilde{P})& \geq \Tr(\Pi^n_{\alpha,\rho^n}\rho^n \Pi^n_{\alpha,\rho^n}  E)\\
    &=\Tr(\Pi^n_{\alpha,\rho^n}\rho^n \Pi^n_{\alpha,\rho^n}  P\Pi^n_{\alpha,\rho^n}P)-\Tr(\Pi^n_{\alpha,\rho^n}\rho^n \Pi^n_{\alpha,\rho^n}F)\\
    &\geq\Tr(\Pi^n_{\alpha,\rho^n}\rho^n \Pi^n_{\alpha,\rho^n}  P\Pi^n_{\alpha,\rho^n}P)-2^{-\alpha \sqrt{n}}\\
    &\geq \Tr(\rho^n  P\Pi^n_{\alpha,\rho^n}P)-\norm{\Pi^n_{\alpha,\rho^n}\rho^n \Pi^n_{\alpha,\rho^n}-\rho^n}_1 -2^{-\alpha \sqrt{n}}\\ 
    &\geq \Tr(\rho^n \Pi^n_{\alpha,\rho^n})-\norm{P\rho^n P-\rho^n}_1-\norm{\Pi^n_{\alpha,\rho^n}\rho^n \Pi^n_{\alpha,\rho^n}-\rho^n}_1-2^{-\alpha \sqrt{n}}\\ 
    & \geq 1-\frac{\beta}{\alpha^2}-2\sqrt{\epsilon}-2\frac{\sqrt{\beta}}{\alpha}-2^{-\alpha \sqrt{n}},
\end{align*}}
where \zk{the first line follows from the definition of $E$ which implies $\widetilde{P} \geq E$. 
The third line follows from H\"{o}lder's inequality 
in the following form: $\Tr(\Pi^n_{\alpha,\rho^n}\rho^n \Pi^n_{\alpha,\rho^n}F) \leq \Tr(\Pi^n_{\alpha,\rho^n}\rho^n \Pi^n_{\alpha,\rho^n})\cdot \norm{F}_{\infty}\leq 2^{-\alpha\sqrt{n}}$.
The fourth and fifth 
lines are due to H\"{o}lder's inequality 
in the following form: for any two states $\rho$ and $\sigma$ and any operator $0 \leq \Lambda \leq \1$, $\Tr(\rho \Lambda) \geq \Tr(\sigma \Lambda)-\norm{\rho-\sigma}_1$
holds which is obtained by rearranging terms  in the following H\"{o}lder's inequality $\Tr((\rho-\sigma)\Lambda)\leq \norm{\rho-\sigma}_1 \cdot \norm{\Lambda}_{\infty}\leq \norm{\rho-\sigma}_1$
}.
The last line follows from Lemmas \ref{lemma:typicality properties}
and \ref{Gentle Operator Lemma}.

\medskip
2. By the fact that in the typical subspace the eigenvalues of $\rho^n$ are bounded 
(Lemma \ref{lemma:typicality properties}), we obtain 
\begin{align*}
      \widetilde{P} \Pi^n_{\alpha,\rho^n}\rho^n \Pi^n_{\alpha,\rho^n} \widetilde{P} &\leq 2^{-\sum_{i=1}^n S(\rho_i)+\alpha\sqrt{n}}\widetilde{P} \Pi^n_{\alpha ,\rho^n }  \widetilde{P} \\
      &\leq 2^{-\sum_{i=1}^n S(\rho_i)+\alpha\sqrt{n}}\widetilde{P}.
\end{align*}
For the lower bound notice that 
\begin{align*}
      \widetilde{P} \Pi^n_{\alpha,\rho^n}\rho^n \Pi^n_{\alpha,\rho^n} \widetilde{P}& \geq 2^{-\sum_{i=1}^n S(\rho_i)-\alpha\sqrt{n}}\widetilde{P} \Pi^n_{\alpha ,\rho^n }  \widetilde{P}\\ 
      &=2^{-\sum_{i=1}^n S(\rho_i)-\alpha\sqrt{n}}\widetilde{P}P \Pi^n_{\alpha ,\rho^n } P \widetilde{P} \\ 
      &\geq 2^{-\sum_{i=1}^n S(\rho_i)-2\alpha\sqrt{n}} \widetilde{P},
\end{align*}
where the equality holds because $\widetilde{P} \subseteq \mathcal{M}$, therefore $\widetilde{P} P=\widetilde{P}$. The last inequality follows because $\widetilde{P}$ is the projection onto support of $P \Pi^n_{\alpha ,\rho^n } P$ with eigenvalues greater $2^{-\alpha\sqrt{n}}$.


\medskip
3. \zk{In this point, we construct $\widetilde{\rho}$.} Consider the unnormalized state $\widetilde{P} \Pi^n_{\alpha,\rho^n}\rho^n \Pi^n_{\alpha,\rho^n} \widetilde{P}$ 
with support inside $\widetilde{\mathcal{M}}$.
From from point 2, we know that all the eigenvalues of this state belongs to the interval 
$\left[2^{-\sum_{i=1}^n S(\rho_i)-2\alpha \sqrt{n}},2^{-\sum_{i=1}^n S(\rho_i)+\alpha \sqrt{n}}\right] 
       := [p_{\min},p_{\max}]$. 
We divide this interval into $b=2^{\floor{5\alpha \sqrt{n}}}$ many intervals (bins) of equal length 
$\Delta p=\frac{p_{\max}-p_{\min}}{b}$. Now, we \emph{trim} the eigenvalues of this unnormalized 
state in three steps as follows.
\begin{enumerate}[(a)]
    \item Each eigenvalue belongs to a bin which is an interval $[p_k,p_{k+1})$ for some 
           $0 \leq k \leq b -1$ with $p_k=p_{\min}+\Delta p \times k$. For example, eigenvalue  
           $\lambda_l$ is equal to $p_k+q_l$ for some $k$ such that $0\leq q_l< \Delta p$. 
           We throw away the $q_l$ part of each eigenvalue $\lambda_l$. The sum of these parts 
           over all eigenvalues is very small,
           \begin{align*}
             \sum_{l=1}^{|\widetilde{\mathcal{M}}|} q_l 
                   \leq \Delta p |\widetilde{\mathcal{M}}| \leq 2^{-2\alpha \sqrt{n}+1},
           \end{align*}
           where the dimension of the subspace $\widetilde{\mathcal{M}}$ is bounded as 
           $|\widetilde{\mathcal{M}}|\leq 2^{\sum_{i=1}^n S(\rho_i)+2\alpha \sqrt{n}}$, which 
           follows from point 2 of the lemma.

    \item We throw away the bins which contain less than $2^{\sum_{i=1}^n S(\rho_i)-10\alpha \sqrt{n}}$ 
           many eigenvalues. The sum of all the eigenvalues that are thrown away is bounded by
           \begin{align*}
             2^{\sum_{i=1}^n S(\rho_i)-10\alpha \sqrt{n}} 
                  \times 2^{5\alpha \sqrt{n}} 
                  \times 2^{-\sum_{i=1}^n S(\rho_i)+\alpha \sqrt{n}} \leq 2^{-4\alpha \sqrt{n}},
           \end{align*}
           where the first number in the product is the number of eigenvalues 
           in such a bin, the second is the number of bins, and the third is the maximum eigenvalue. 
    
    \item If the $k$-th bin is not thrown away in the previous step, it contains 
           $M_k$ many equal eigenvalues, where $M_k$ is bounded as follows:
           \begin{align}\label{eq:bounds of bin size}
             2^{\sum_{i=1}^n S(\rho_i)-10\alpha  \sqrt{n}} 
                   \leq  M_k \leq 2^{\sum_{i=1}^n S(\rho_i)+2\alpha  \sqrt{n}}.
           \end{align}
           Let
           \begin{align}\label{eq:L}
             L= 2^{\floor{\sum_{i=1}^n S(\rho_i) -10 \alpha\sqrt{n}}}  
           \end{align}
           and for the $k$th bin, let $m_{k}$ be an integer number such that 
           \begin{align}\label{M_k_L}
             m_{k} L\leq M_k  \leq (m_{k}+1) L.     
           \end{align}
           Then, $m_{k}$ is bounded as follows
           \begin{align}\label{m_k}
             m_{k} \leq 2^{12\alpha\sqrt{n}}. 
           \end{align}
           From the $k$-th bin, we keep $m_{k} L$ number of eigenvalues and throw away the rest,
           where there are $M_k-m_{k} L \leq L$ many of them; the sum of the eigenvalues that are 
           thrown away in this step is bounded by
           \begin{align}
             \sum_{k=0}^{b-1}  p_{k}(M_k-m_k L) \leq L\sum_{k=0}^{b-1} p_k \leq  2^{-4\alpha\sqrt{n}}. \nonumber 
           \end{align}
\end{enumerate}
Hence, for sufficiently large $n$ the sum of the eigenvalues thrown away in 
the last three steps is bounded by
\begin{align}\label{eq:thrown away sum}
    2^{-2\alpha\sqrt{n}+1}+2^{-4\alpha\sqrt{n}}+2^{-4\alpha\sqrt{n}}\leq 2^{-\alpha\sqrt{n}}
\end{align}
\zk{Therefore, there are only left $b$ different eigenvalues where the $k$-th eigenvalue has degeneracy of $m_k L$ for $k=0,1,\cdots, b-1$.}
In other words, the eigenvalues of all bins
 that are not thrown away in these three steps, form an $L$-fold degenerate unnormalized state of dimension $\sum_{k=0}^{b-1} m_{k} L$ because each eigenvalue has at least degeneracy of the order of $L$. 
Thus, up to a unitary $U^{\dagger}$, it can be factorized into the tensor product of a maximally mixed state $\tau$ and  an unnormalized state $\omega'$ of dimensions $L$ and $\sum_{k=0}^{b-1} m_{k} $, respectively. 
From Eq. (\ref{m_k}), the dimension of $\omega'$ is bounded by
    \begin{align}
         \sum_{k=0}^{b-1} m_{k} \leq  2^{12\alpha\sqrt{n}} \times 2^{5 \alpha\sqrt{n}}=2^{17 \alpha\sqrt{n}}.    \nonumber 
    \end{align}
Then, let $\omega =\frac{\omega'}{\Tr(\omega')}$ and define 
\begin{align*}
   \widetilde{\rho}:=U \tau \otimes \omega U^{\dagger}.
\end{align*}    
   
\medskip
4. \zk{ In point 3 of the lemma, we trimmed $\widetilde{P} \Pi^n_{\alpha,\rho^n}\rho^n \Pi^n_{\alpha,\rho^n} \widetilde{P}$ to obtain the state $\widetilde{\rho}$, i.e.  $\widetilde{\rho} \approx \widetilde{P} \Pi^n_{\alpha,\rho^n}\rho^n \Pi^n_{\alpha,\rho^n} \widetilde{P}$.
Moreover, from point 1 of the lemma we know that the unnormalized state $\Pi^n_{\alpha,\rho^n}\rho^n \Pi^n_{\alpha,\rho^n}$ projects onto $\widetilde{P}$ with high probability. Therefore, 
Lemma~\ref{Gentle Operator Lemma} and Lemma~\ref{lemma:typicality properties} imply that the new unnormalized state $\widetilde{P} \Pi^n_{\alpha,\rho^n}\rho^n \Pi^n_{\alpha,\rho^n} \widetilde{P} \approx \Pi^n_{\alpha,\rho^n}\rho^n \Pi^n_{\alpha,\rho^n} \approx \rho^n$. Hence, we obtain that $\widetilde{\rho} \approx \rho^n$. In the following, we prove this in detail.}
From points 3 and 1, we obtain
\begin{align}\label{eq: Tr of omega'}
   \Tr(\omega') &= \Tr(\tau \otimes \omega') \\
   &\geq \Tr(\widetilde{P} \Pi^n_{\alpha,\rho^n}\rho^n \Pi^n_{\alpha,\rho^n} \widetilde{P})-2^{-\alpha\sqrt{n}}\\
   &\geq 1-2\sqrt{\epsilon}-2\frac{\sqrt{\beta}}{\alpha}-\frac{\beta}{\alpha^2}-2^{-\alpha \sqrt{n}+1}.
\end{align}
Thereby, we get the following
\begin{align*}
    \norm{\widetilde{\rho}-\rho^n}_1& \leq \norm{\widetilde{\rho}-U \tau \otimes \omega' U^{\dagger}}_1+\norm{U \tau \otimes \omega' U^{\dagger}-\widetilde{P} \Pi^n_{\alpha,\rho^n}\rho^n \Pi^n_{\alpha,\rho^n} \widetilde{P}}_1+ \norm{\widetilde{P} \Pi^n_{\alpha,\rho^n}\rho^n \Pi^n_{\alpha,\rho^n} \widetilde{P} -\rho^n}_1 \\
    &\leq 1-\Tr(\omega')+\norm{U \tau \otimes \omega' U^{\dagger}-\widetilde{P} \Pi^n_{\alpha,\rho^n}\rho^n \Pi^n_{\alpha,\rho^n} \widetilde{P}}_1+ \norm{\widetilde{P} \Pi^n_{\alpha,\rho^n}\rho^n \Pi^n_{\alpha,\rho^n} \widetilde{P} -\rho^n}_1\\
    &\leq 1-\Tr(\omega')+2^{-\alpha \sqrt{n}}+\norm{\widetilde{P} \Pi^n_{\alpha,\rho^n}\rho^n \Pi^n_{\alpha,\rho^n} \widetilde{P} -\rho^n}_1\\
    &\leq 1-\Tr(\omega')+2^{-\alpha \sqrt{n}} +2\sqrt{2\sqrt{\epsilon}+2\frac{\sqrt{\beta}}{\alpha}+\frac{\beta}{\alpha^2}+2^{-\alpha \sqrt{n}}} \\
    &= 2\sqrt{\epsilon}+2\frac{\sqrt{\beta}}{\alpha}+\frac{\beta}{\alpha^2}+2^{-\alpha \sqrt{n}+1}+2\sqrt{2\sqrt{\epsilon}+2\frac{\sqrt{\beta}}{\alpha}+\frac{\beta}{\alpha^2}+2^{-\alpha \sqrt{n}}},
\end{align*}
where the first line is due to triangle inequality. The second, third and fourth lines are 
due to Eqs. (\ref{eq: Tr of omega'}) and (\ref{eq:thrown away sum}), 
and Lemma \ref{Gentle Operator Lemma}, respectively.
\end{proof}

\begin{proof-of}[{of Theorem \ref{Asymptotic equivalence theorem}}]
\zk{We first sketch the proof in this paragraph and later provide rigorous steps of the proof.    
The approximate microcanonical (a.m.c.) subspace for charges $A_j$ and average values $v_j$ which is basically a \textit{common} subspace for the spectral projectors of $A_j^{(n)}$ with corresponding values close to $n v_j$; that is, a subspace onto which a state projects with high probability if and only if it projects onto the spectral projectors of the charges with high probability. We show in Theorem~\ref{thm:symmetric-micro-exists} that for a large enough $n$ such a subspace exits.
An interesting property of an a.m.c. subspace is that any unitary acting on this subspace 
is an almost commuting unitary with charges $A_j^{(n)}$.

In Corollary~\ref{corollary: a.m.c. projection}, we show that assuming $\frac{1}{n}\tr (\rho^n A^{(n)}_j)\approx \frac{1}{n}\tr (\rho^n A^{(n)}_j) \approx v_j$ the states $\rho^n$ and
$\sigma^n$ project onto the a.m.c. subspace with high probability. 
Hence, in Lemma~\ref{lemma: timmed state}, we show that one can find states $\widetilde{\rho}$ and $\widetilde{\sigma}$ with support inside the a.m.c. subspace which are very close to the original states in trace norm, that is, $\widetilde{\rho} \approx \rho^n$ and $\widetilde{\sigma} \approx \sigma^n$, and there are unitaries $V_1$ and $V_2$ that factorizes these states to the tensor product of maximally mixed states $\tau$ and $\tau'$ and some other state of very small dimension:
\begin{align*}
    V_1\widetilde{\rho}V_1^{\dagger}=\tau \otimes \omega \quad \text{and} \quad 
    V_2\widetilde{\sigma}V_2^{\dagger}=\tau' \otimes \omega'.
\end{align*}
Further, assuming that the states $\rho^n$ and $\sigma^n$ have very close entropy rates, i.e. 
$\frac{1}{n}S(\rho^n) \approx \frac{1}{n}S(\sigma^n)$, one can find states $\tau$ and $\tau'$ with the same dimension that is $\tau=\tau'$. Thus, we observe that two states $\widetilde{\rho}\otimes \omega'$ and $\widetilde{\sigma}\otimes \omega$ have exactly the same spectrum, so there is unitary acting on the a.m.c. subspace and the ancillary system taking one state to another. Based on the properties of the a.m.c. subspace, we show that this unitary is an almost commuting unitary with the charges 
$A_j^{(n)}$.}

\medskip
We first prove the \emph{if} part. If there is an almost-commuting unitary $U$ and an ancillary system with the desired properties stated in the theorem, then we obtain
\zk{
\begin{align}
    \frac{1}{n}\abs{S(\rho^n)-S(\sigma^n)}
    &= \frac{1}{n}\abs{S(\rho^n \otimes \omega')-S(\sigma^n\otimes \omega)
    -S(\omega')+S(\omega)} \nonumber \\
    &\leq \frac{1}{n}\abs{S(\rho^n \otimes \omega')-S(\sigma^n\otimes \omega)}
    +\frac{1}{n}\abs{S(\omega')-S(\omega)} \nonumber \\
    &\leq \frac{1}{n}\abs{S(\rho^n\otimes \omega')-S(\sigma^n\otimes \omega)}
    +\frac{2}{n}\log 2^{o(n)} \nonumber \\
    &= \frac{1}{n}\abs{S(U(\rho^n\otimes \omega' )U^{\dagger})-S(\sigma^n\otimes \omega)}
    +o(1) \nonumber \\
    &\leq \frac{1}{n}  o(1) \log (d^{n} \times 2^{ o(n)})
    +\frac{1}{n}h\left(o(1)\right)+o(1)
     = o(1) \nonumber,
\end{align}
where the first line follows from the additivity of the von Neumann on tensor product states and adding and subtracting $S(\omega)$ and $S(\omega')$.} The second line is due to the triangle inequality. The third line is due to the fact that von Neumann entropy of a state is upper bounded by the logarithm of the dimension \zk{ (assuming that the dimension of the ancillary system is bounded by $2^{o(n)}$).
The fourth line follows because unitaries do not change the entropy.
The last line follows because the trace distance between the two states $U(\rho^n\otimes \omega' )U^{\dagger}$ and $\sigma^n\otimes \omega$ converges to zero, therefore we can apply the continuity of von Neumann entropy \cite{Fannes1973,Audenaert2007}  where $h(x) = -x \log x - (1 - x) \log(1 - x)$ is the binary entropy function.}   
Moreover, we obtain
\begin{align}\label{approximate average charge conservation}
    \frac{1}{n}&\abs{\Tr(\rho^n A_j^{(n)})-\Tr(\sigma^n A_j^{(n)})} \nonumber\\
    &=\frac{1}{n}\abs{\Tr\left( \rho^n \otimes \omega'  (A_j^{(n)}+A_j')\right)- \Tr\left( \sigma^n\otimes \omega  (A_j^{(n)}+A_j')\right) } \nonumber \\
    &\leq\frac{1}{n}\abs{\Tr\left( \rho^n \otimes \omega'  (A_j^{(n)}+A_j')\right)- \Tr\left( U\rho^n\otimes \omega'U^{\dagger}  (A_j^{(n)}+A_j')\right) }\nonumber\\
    &\quad \quad +\frac{1}{n}\abs{\Tr\left( U\rho^n\otimes \omega'U^{\dagger}  (A_j^{(n)}+A_j')\right)- \Tr\left( \sigma^n\otimes \omega  (A_j^{(n)}+A_j')\right) }\nonumber\\ 
    &=\frac{1}{n}\abs{\Tr\left( \rho^n \otimes \omega'  \left(A_j^{(n)}+A_j' -U^{\dagger}(A_j^{(n)}+A_j')U\right)\right) }
     +\frac{1}{n}\abs{\Tr\left( \left(U\rho^n\otimes  \omega'U^{\dagger} - \sigma^n\otimes \omega \right)  (A_j^{(n)}+A_j')\right)} \nonumber\\ 
    &\leq  \frac{1}{n} \Tr(\rho^n \otimes \omega') \norm{U(A_j^{(n)}+A_j')U^{\dagger} - (A_j^{(n)}+A_j')}_{\infty} + \frac{1}{n} \norm{U\rho^n\otimes  \omega'U^{\dagger} - \sigma^n\otimes \omega}_1 \norm{A_j^{(n)}+A_j'}_{\infty} \nonumber\\
    & = o(1),  
\end{align}
the second line follows because $A_j'=0$ for all $j$. The third and fifth lines are due to  
triangle inequality and H\"{o}lder's inequality, respectively. 

\medskip
Now we turn to the proof of the \emph{only if} part. 
\zk{That is, assuming that $\rho^n$ and $\sigma^n$ are asymptotically equivalent, we construct the ancillary system and the almost commuting unitaries.
We apply Theorem \ref{thm:symmetric-micro-exists} to construct a non-trivial a.m.c. subspace for $\rho^n$. Since $\sigma^n$ has average entropy and charges values very close to those of $\rho^n$, both $\rho^n$ and $\sigma^n$ project to this a.m.c. subspace with high probability. Then we apply Lemma~\ref{lemma: timmed state} to find states $\widetilde{\rho} \approx \rho^n$ and $\widetilde{\sigma}\approx \sigma^n$ where these states (up to unitaries) are decomposed as the tensor product of a maximally mixed state $\tau$ of very large dimension and another state of very small dimension, i.e. $V_1\widetilde{\rho}V_1^{\dagger}=\tau \otimes \omega$ and $
 V_2\widetilde{\sigma}V_2^{\dagger}=\tau \otimes \omega'$. Now, we consider the states 
 $\underbrace{\tau \otimes \omega}_{\approx \rho^n} \otimes \omega'$ and $\underbrace{\tau \otimes \omega'}_{\approx \sigma^n}\otimes \omega$ which have exactly the same eigenvalues, hence the states $\rho^n \otimes \omega'$ and $\sigma^n \otimes \omega$ have very similar eigenvalues. Therefore,  $\rho^n \otimes \omega'$ and $\sigma^n \otimes \omega$ are approximately equal up to a unitary. In the end, we show that such a unitary, with support inside a.m.c subspace, almost commutes with all charges of the total system.

} 
Namely, assume that for the sates $\rho^n$ and $\sigma^n$ the following holds:
\begin{align*}
   & \frac{1}{n} \abs{S(\rho^n)-S(\sigma^n)} \leq \gamma_n \\ 
   & \frac{1}{n} \abs{ \Tr(A_j^{(n)}\rho^n)-\Tr(A_j^{(n)}\sigma^n)}\leq \gamma'_n  \quad \quad j=1,\ldots,c,
\end{align*}
for vanishing $\gamma_n $ and $\gamma'_n $ as $n$ goes to $\infty$.
According to Theorem \ref{thm:symmetric-micro-exists}, for charges $A_j$, values $ v_j=\frac{1}{n} \Tr(\rho^n A_j^{(n)})$, $\eta'>0$ and any $n>0$, there is an a.m.c. subspace $\mathcal{M}$ of $\mathcal{H}^{\otimes n}$ with projector $P$ and the following parameters:
\begin{align*}
  &\eta=2 \eta',\\
  &\delta'=\frac{c+3}{2}(5n)^{2d^2} e^{-\frac{n \eta'^2}{8c^2(d
  +1)^2}}, \\ 
  &\delta=(c+3)(5n)^{2d^2} e^{-\frac{n \eta'^2}{8c^2(d
  +1)^2}},\\ 
  &\epsilon=2(c+3)(n+1)^{3d^2}(5n)^{2d^2} e^{-\frac{n \eta'^2}{8c^2(d
  +1)^2}}. 
\end{align*}
Choose $\eta'$ as the following such that $\delta$, $\delta'$ and $\epsilon$ vanish for large $n$:
\begin{align*}
\eta'=\left\{
                \begin{array}{ll}
                  \frac{\sqrt{8}c(d+1) }{n^{\frac{1}{4}} \Sigma(A)_{\min}} \quad &\text{if} \quad \gamma'_n \leq \frac{1}{n^{\frac{1}{4}}}\\
                  \frac{\sqrt{8}c(d+1)\gamma'_n}{ \Sigma(A)_{\min}} \quad &\text{if} \quad \gamma'_n > \frac{1}{n^{\frac{1}{4}}}
                \end{array}
              \right.
\end{align*}
where $\Sigma(A)_{\min}$ is the minimum spectral diameter among all spectral diameters of 
charges $\Sigma(A_j)$. Since $\frac{1}{n} \Tr(\rho^n A_j^{(n)})= v_j$ and $\abs{\frac{1}{n}\Tr(\sigma^n A_j^{(n)})- v_j} \leq \frac{1}{2}\eta' \Sigma (A_j)$, Corollary~\ref{corollary: a.m.c. projection} implies that states 
$\rho^n$ and $\sigma^n$ project onto this a.m.c. subspace with probability $\epsilon$:
\begin{align*}
    &\Tr(\rho^n P)\geq 1-\epsilon,\\
    &\Tr(\sigma^n P)\geq 1-\epsilon.
\end{align*}
Moreover, consider the typical projectors $\Pi^n_{\alpha ,\rho^n }$ and $\Pi^n_{\alpha ,\sigma^n }$ of states $\rho^n$ and $\sigma^n$, respectively, with $\alpha=n^{\frac{1}{3}}$. Then point 3 and 4 of Lemma~\ref{lemma: timmed state} implies that there are states $\widetilde{\rho}$ and $\widetilde{\sigma}$ with support inside the a.m.c. subspace $\mathcal{M}$ and unitaries $V_1$ and $V_2$ such that
\begin{align}\label{eq:4 formulas}
 &\norm{\widetilde{\rho}-\rho^n}_1 \leq o(1) , \nonumber\\
 &\norm{\widetilde{\sigma}-\sigma^n}_1 \leq o(1), \nonumber\\
 &V_1\widetilde{\rho}V_1^{\dagger}=\tau \otimes \omega, \nonumber\\
 &V_2\widetilde{\sigma}V_2^{\dagger}=\tau' \otimes \omega',
\end{align}
where $\tau$ and $\tau'$ are maximally mixed states; since $\abs{S(\rho^n)-S(\sigma^n)} \leq n \gamma_n$, one may choose the dimension of them in Eq.~(\ref{eq:L}) to be exactly the same as $L=2^{\floor{\sum_{i=1}^n S(\rho_i)-10 z}}$ with $z=\max \{\alpha \sqrt{n}, n \gamma_n\}$, hence, we obtain $\tau=\tau'$. Then, $\omega$ and $\omega'$ are states with support inside Hilbert space $\mathcal{K}$ of dimension $2^{o(z)}=2^{o(n)}$.
%
%
Then, it is immediate to see that the states $\widetilde{\rho} \otimes \omega'$ and $\widetilde{\sigma} \otimes \omega$ on Hilbert space $\mathcal{M}_t=\mathcal{M} \otimes \mathcal{K}$ have exactly the same spectrum; thus, there is a unitary $\widetilde{U}$ on  subspace $\mathcal{M}_t$ such that 
\begin{align}\label{eq:exact spectrum states}
    \widetilde{U} \widetilde{\rho} \otimes \omega' \widetilde{U}^{\dagger} =\widetilde{\sigma} \otimes \omega.
\end{align}
We extend the unitary $\widetilde{U}$ to $U=\widetilde{U} \oplus \1_{\mathcal{M}_t^{\perp}}$
acting on $\mathcal{H}^{\otimes n} \otimes \mathcal{K}$ and obtain
\begin{align*}
\norm{U \rho^n \otimes \omega' U^{\dagger} - \sigma^n \otimes \omega}_1 &\leq
\norm{U \rho^n \otimes \omega' U^{\dagger} - U \widetilde{\rho} \otimes \omega' U^{\dagger}}_1 +\norm{   \sigma^n \otimes \omega-\widetilde{\sigma}\otimes \omega}_1 
  +\norm{ U \widetilde{\rho} \otimes \omega' U^{\dagger}-\widetilde{\sigma}\otimes \omega}_1 \\
 &=\norm{U \rho^n \otimes \omega' U^{\dagger} - U \widetilde{\rho} \otimes \omega' U^{\dagger}}_1 +\norm{   \sigma^n \otimes \omega-\widetilde{\sigma}\otimes \omega}_1 \\
 &\leq o(1),
\end{align*}
where the second and last lines are due to Eqs. (\ref{eq:exact spectrum states}) and (\ref{eq:4 formulas}), respectively.

As mentioned before, $\mathcal{M}_t=\mathcal{M} \otimes \mathcal{K}$ is a subspace of $\mathcal{H}^{\otimes n} \otimes \mathcal{K}$ with projector $P_t=P \otimes \1_{\mathcal{K}}$ where $P$ is the corresponding projector of a.m.c. subspace. We define total charges $A_j^t=A_j^{(n)}+A_j'$ and let $A_j'=0$ for all $j$ and show that every unitary of the form $U=U_{\mathcal{M}_t} \oplus \1_{\mathcal{M}_t^{\perp}}$  asymptotically commutes with all total charges: 
\begin{align*}
\norm{U A_j^t U^{\dagger}-A_j^t}_{\infty} &= \norm{(P_t+P_t^{\perp})(U A_j^t U^{\dagger}-A_j^t)(P_t+P_t^{\perp})}_{\infty} \\
& \leq \norm{P_t(U A_j^t U^{\dagger}-A_j^t)P_t}_{\infty}+\norm{P_t^{\perp}(U A_j^t U^{\dagger}-A_j^t)P_t}_{\infty} \\
& \quad +\norm{P_t(U A_j^t U^{\dagger}-A_j^t)P_t^{\perp}}_{\infty}+\norm{P_t^{\perp}(U A_j^t U^{\dagger}-A_j^t)P_t^{\perp}}_{\infty}\\
& = \norm{P_t(U A_j^t U^{\dagger}-A_j^t)P_t}_{\infty}+2\norm{P_t^{\perp}(U A_j^t U^{\dagger}-A_j^t)P_t}_{\infty} \\
& \leq 3  \norm{(U A_j^t U^{\dagger}-A_j^t)P_t}_{\infty}\\
& = 3  \norm{(U A_j^t U^{\dagger} -n v_j \1+ nv_j \1-A_j^t)P_t}_{\infty}\\
& \leq 3  \norm{(U A_j^t U^{\dagger} -n v_j \1)P_t}_{\infty}+3  \norm{(A_j^t- nv_j \1)P_t}_{\infty}\\
& = 6   \norm{(A_j^t- nv_j \1)P_t}_{\infty}\\
& = 6 \max_{\ket{v} \in \mathcal{M}_t} \norm{(A_j^t-n v_j \1)\ket{v}}_2\\
& = 6 \max_{\ket{v} \in \mathcal{M}_t} \norm{ (A_j^t-n v_j \1)(\Pi_j^{\eta} \otimes \1_{\mathcal{K}} +\1-\Pi_j^{\eta} \otimes \1_{\mathcal{K}}) \ket{v} 
 }_2  \\
& \leq 6 \max_{\ket{v} \in \mathcal{M}_t} \norm{(A_j^t-n v_j \1) \Pi_j^{\eta} \otimes \1 \ket{v}}_2 +6 \max_{\ket{v} \in \mathcal{M}_t} \norm{(A_j^t-n v_j \1) (\1-\Pi_j^{\eta} \otimes \1) \ket{v}}_2  \\
& \leq  6n \Sigma (A_j) \eta +6 \max_{\ket{v} \in \mathcal{M}_t} \norm{(A_j^t-n v_j \1) (\1-\Pi_j^{\eta} \otimes \1) \ket{v}}_2,
\end{align*}
where the first line is due to the fact that $P_t+P_t^{\perp}=\1_{\mathcal{H}^{\ox n}} \ox \1_{\mathcal{K}}$. The fourth line follows because $U A_j^t U^{\dagger}-A_j^t$ is a Hermitian operator with zero eigenvalues in the subspace $P_t^{\perp}$. 
The fifth line is due to Lemma~\ref{lemma:norm inequality}. The twelfth line is due to the definition of the a.m.c. subspace. Now, bound the second term in the above:
\begin{align*}
6 \max_{\ket{v} \in \mathcal{M}_t} \norm{(A_j^t-n v_j \1) (\1-\Pi_j^{\eta} \otimes \1) \ket{v}}_2 & \leq  6 \max_{\ket{v} \in \mathcal{M}_t} \norm{A_j^t-n v_j \1 }_{\infty} \norm{ (\1-\Pi_j^{\eta} \otimes \1) \ket{v}}_2 \\
& = 6 \norm{A_j^t-n v_j \1}_{\infty}  \max_{\ket{v} \in \mathcal{M}_t} \sqrt{ \Tr ((\1-\Pi_j^{\eta} \otimes \1)\ketbra{v}{v})}\\
& = 6 n\norm{A_j- v_j \1}_{\infty} \max_{v \in \mathcal{M}} \sqrt{ \Tr ((\1-\Pi_j^{\eta} )v)}\\
& \leq 6 n\norm{A_j- v_j \1}_{\infty} \sqrt{ \delta},
\end{align*}
where the first line is due to Lemma~\ref{lemma:norm inequality}. 
The last line is by definition of the a.m.c. subspace. 
Thus, for vanishing $\delta$ and $\eta$ we obtain
\begin{align*}
    \frac{1}{n}\norm{U A_j^t U^{\dagger}-A_j^t}_{\infty} \leq o(1),
\end{align*}
concluding the proof.
\end{proof-of}

\pagebreak 


\begin{thebibliography}{80}%
\makeatletter
\providecommand \@ifxundefined [1]{%
 \@ifx{#1\undefined}
}%
\providecommand \@ifnum [1]{%
 \ifnum #1\expandafter \@firstoftwo
 \else \expandafter \@secondoftwo
 \fi
}%
\providecommand \@ifx [1]{%
 \ifx #1\expandafter \@firstoftwo
 \else \expandafter \@secondoftwo
 \fi
}%
\providecommand \natexlab [1]{#1}%
\providecommand \enquote  [1]{``#1''}%
\providecommand \bibnamefont  [1]{#1}%
\providecommand \bibfnamefont [1]{#1}%
\providecommand \citenamefont [1]{#1}%
\providecommand \href@noop [0]{\@secondoftwo}%
\providecommand \href [0]{\begingroup \@sanitize@url \@href}%
\providecommand \@href[1]{\@@startlink{#1}\@@href}%
\providecommand \@@href[1]{\endgroup#1\@@endlink}%
\providecommand \@sanitize@url [0]{\catcode `\\12\catcode `\$12\catcode
  `\&12\catcode `\#12\catcode `\^12\catcode `\_12\catcode `\%12\relax}%
\providecommand \@@startlink[1]{}%
\providecommand \@@endlink[0]{}%
\providecommand \url  [0]{\begingroup\@sanitize@url \@url }%
\providecommand \@url [1]{\endgroup\@href {#1}{\urlprefix }}%
\providecommand \urlprefix  [0]{URL }%
\providecommand \Eprint [0]{\href }%
\providecommand \doibase [0]{http://dx.doi.org/}%
\providecommand \selectlanguage [0]{\@gobble}%
\providecommand \bibinfo  [0]{\@secondoftwo}%
\providecommand \bibfield  [0]{\@secondoftwo}%
\providecommand \translation [1]{[#1]}%
\providecommand \BibitemOpen [0]{}%
\providecommand \bibitemStop [0]{}%
\providecommand \bibitemNoStop [0]{.\EOS\space}%
\providecommand \EOS [0]{\spacefactor3000\relax}%
\providecommand \BibitemShut  [1]{\csname bibitem#1\endcsname}%
\let\auto@bib@innerbib\@empty
\bibitem [{\citenamefont {Gemmer}\ \emph {et~al.}(2009)\citenamefont {Gemmer},
  \citenamefont {Michel},\ and\ \citenamefont {Mahler}}]{Gemmer2009}%
  \BibitemOpen
  \bibfield  {author} {\bibinfo {author} {\bibfnamefont {Jochen}\ \bibnamefont
  {Gemmer}}, \bibinfo {author} {\bibfnamefont {Mathias}\ \bibnamefont
  {Michel}}, \ and\ \bibinfo {author} {\bibfnamefont {G\"unter}\ \bibnamefont
  {Mahler}},\ }\href {\doibase 10.1007/978-3-540-70510-9} {\emph {\bibinfo
  {title} {Quantum Thermodynamics}}},\ \bibinfo {series} {Lecture Notes in
  Physics}, Vol.\ \bibinfo {volume} {784}\ (\bibinfo  {publisher} {Springer
  Verlag, Berlin Heidelberg New York},\ \bibinfo {year} {2009})\BibitemShut
  {NoStop}%
\bibitem [{\citenamefont {Binder}\ \emph {et~al.}(2018)\citenamefont {Binder},
  \citenamefont {Correa}, \citenamefont {Gogolin}, \citenamefont {Anders},\
  and\ \citenamefont {Adesso}}]{book2}%
  \BibitemOpen
  \bibinfo {editor} {\bibfnamefont {Felix}\ \bibnamefont {Binder}}, \bibinfo
  {editor} {\bibfnamefont {Luis~A.}\ \bibnamefont {Correa}}, \bibinfo {editor}
  {\bibfnamefont {Christian}\ \bibnamefont {Gogolin}}, \bibinfo {editor}
  {\bibfnamefont {Janet}\ \bibnamefont {Anders}}, \ and\ \bibinfo {editor}
  {\bibfnamefont {Gerardo}\ \bibnamefont {Adesso}},\ eds.,\ \href {\doibase
  10.1007/978-3-319-99046-0} {\emph {\bibinfo {title} {Thermodynamics in the
  Quantum Regime: Fundamental Aspects and New Directions}}},\ \bibinfo {series}
  {Series Fundamental Theories of Physics}, Vol.\ \bibinfo {volume} {195}\
  (\bibinfo  {publisher} {Springer Verlag},\ \bibinfo {year}
  {2018})\BibitemShut {NoStop}%
\bibitem [{\citenamefont {Jarzynski}(1997)}]{Jar97}%
  \BibitemOpen
  \bibfield  {author} {\bibinfo {author} {\bibfnamefont {Christopher}\
  \bibnamefont {Jarzynski}},\ }\bibfield  {title} {\enquote {\bibinfo {title}
  {Nonequilibrium equality for free energy differences},}\ }\href {\doibase
  10.1103/PhysRevLett.78.2690} {\bibfield  {journal} {\bibinfo  {journal}
  {Physical Review Letters}\ }\textbf {\bibinfo {volume} {78}},\ \bibinfo
  {pages} {2690--2693} (\bibinfo {year} {1997})}\BibitemShut {NoStop}%
\bibitem [{\citenamefont {Crooks}(1999)}]{ro99}%
  \BibitemOpen
  \bibfield  {author} {\bibinfo {author} {\bibfnamefont {Gavin~E.}\
  \bibnamefont {Crooks}},\ }\bibfield  {title} {\enquote {\bibinfo {title}
  {Entropy production fluctuation theorem and the nonequilibrium work relation
  for free energy differences},}\ }\href {\doibase 10.1103/PhysRevE.60.2721}
  {\bibfield  {journal} {\bibinfo  {journal} {Physical Review E}\ }\textbf
  {\bibinfo {volume} {60}},\ \bibinfo {pages} {2721--2726} (\bibinfo {year}
  {1999})}\BibitemShut {NoStop}%
\bibitem [{\citenamefont {Campisi}\ \emph {et~al.}(2011)\citenamefont
  {Campisi}, \citenamefont {H\"anggi},\ and\ \citenamefont {Talkner}}]{cht11}%
  \BibitemOpen
  \bibfield  {author} {\bibinfo {author} {\bibfnamefont {Michele}\ \bibnamefont
  {Campisi}}, \bibinfo {author} {\bibfnamefont {Peter}\ \bibnamefont
  {H\"anggi}}, \ and\ \bibinfo {author} {\bibfnamefont {Peter}\ \bibnamefont
  {Talkner}},\ }\bibfield  {title} {\enquote {\bibinfo {title} {{Colloquium:
  Quantum fluctuation relations: Foundations and applications}},}\ }\href
  {\doibase 10.1103/RevModPhys.83.771} {\bibfield  {journal} {\bibinfo
  {journal} {Reviews of Modern Physics}\ }\textbf {\bibinfo {volume} {83}},\
  \bibinfo {pages} {771--791} (\bibinfo {year} {2011})}\BibitemShut {NoStop}%
\bibitem [{\citenamefont {Brand\~ao}\ \emph {et~al.}(2013)\citenamefont
  {Brand\~ao}, \citenamefont {Horodecki}, \citenamefont {Oppenheim},
  \citenamefont {Renes},\ and\ \citenamefont {Spekkens}}]{bho13}%
  \BibitemOpen
  \bibfield  {author} {\bibinfo {author} {\bibfnamefont {Fernando G. S.~L.}\
  \bibnamefont {Brand\~ao}}, \bibinfo {author} {\bibfnamefont {Micha{\l}}\
  \bibnamefont {Horodecki}}, \bibinfo {author} {\bibfnamefont {Jonathan}\
  \bibnamefont {Oppenheim}}, \bibinfo {author} {\bibfnamefont {Joseph~M.}\
  \bibnamefont {Renes}}, \ and\ \bibinfo {author} {\bibfnamefont {Robert~W.}\
  \bibnamefont {Spekkens}},\ }\bibfield  {title} {\enquote {\bibinfo {title}
  {{Resource Theory of Quantum States Out of Thermal Equilibrium}},}\ }\href
  {\doibase 10.1103/PhysRevLett.111.250404} {\bibfield  {journal} {\bibinfo
  {journal} {Physical Review Letters}\ }\textbf {\bibinfo {volume} {111}},\
  \bibinfo {pages} {250404} (\bibinfo {year} {2013})}\BibitemShut {NoStop}%
\bibitem [{\citenamefont {Horodecki}\ and\ \citenamefont
  {Oppenheim}(2013)}]{h&p13}%
  \BibitemOpen
  \bibfield  {author} {\bibinfo {author} {\bibfnamefont {Micha{\l}}\
  \bibnamefont {Horodecki}}\ and\ \bibinfo {author} {\bibfnamefont {Jonathan}\
  \bibnamefont {Oppenheim}},\ }\bibfield  {title} {\enquote {\bibinfo {title}
  {Fundamental limitations for quantum and nanoscale thermodynamics},}\ }\href
  {\doibase 10.1038/ncomms3059} {\bibfield  {journal} {\bibinfo  {journal}
  {Nature Communications}\ }\textbf {\bibinfo {volume} {4}},\ \bibinfo {pages}
  {2059} (\bibinfo {year} {2013})}\BibitemShut {NoStop}%
\bibitem [{\citenamefont {Brand\~ao}\ \emph {et~al.}(2015)\citenamefont
  {Brand\~ao}, \citenamefont {Horodecki}, \citenamefont {Ng}, \citenamefont
  {Oppenheim},\ and\ \citenamefont {Wehner}}]{Brandao2015}%
  \BibitemOpen
  \bibfield  {author} {\bibinfo {author} {\bibfnamefont {Fernando G. S.~L.}\
  \bibnamefont {Brand\~ao}}, \bibinfo {author} {\bibfnamefont {Micha{\l}}\
  \bibnamefont {Horodecki}}, \bibinfo {author} {\bibfnamefont {Nelly}\
  \bibnamefont {Ng}}, \bibinfo {author} {\bibfnamefont {Jonathan}\ \bibnamefont
  {Oppenheim}}, \ and\ \bibinfo {author} {\bibfnamefont {Stephanie}\
  \bibnamefont {Wehner}},\ }\bibfield  {title} {\enquote {\bibinfo {title} {The
  second laws of quantum thermodynamics},}\ }\href {\doibase
  10.1073/pnas.1411728112} {\bibfield  {journal} {\bibinfo  {journal}
  {Proceedings of the National Academy of Sciences}\ }\textbf {\bibinfo
  {volume} {112}},\ \bibinfo {pages} {3275--3279} (\bibinfo {year}
  {2015})}\BibitemShut {NoStop}%
\bibitem [{\citenamefont {Skrzypczyk}\ \emph {et~al.}(2014)\citenamefont
  {Skrzypczyk}, \citenamefont {Short},\ and\ \citenamefont {Popescu}}]{ssp14}%
  \BibitemOpen
  \bibfield  {author} {\bibinfo {author} {\bibfnamefont {Paul}\ \bibnamefont
  {Skrzypczyk}}, \bibinfo {author} {\bibfnamefont {Anthony~J.}\ \bibnamefont
  {Short}}, \ and\ \bibinfo {author} {\bibfnamefont {Sandu}\ \bibnamefont
  {Popescu}},\ }\bibfield  {title} {\enquote {\bibinfo {title} {Work extraction
  and thermodynamics for individual quantum systems},}\ }\href {\doibase
  10.1038/ncomms5185} {\bibfield  {journal} {\bibinfo  {journal} {Nature
  Communications}\ }\textbf {\bibinfo {volume} {5}},\ \bibinfo {pages} {4185}
  (\bibinfo {year} {2014})}\BibitemShut {NoStop}%
\bibitem [{\citenamefont {{\r{A}}berg}(2013)}]{abe13}%
  \BibitemOpen
  \bibfield  {author} {\bibinfo {author} {\bibfnamefont {Johan}\ \bibnamefont
  {{\r{A}}berg}},\ }\bibfield  {title} {\enquote {\bibinfo {title} {Truly
  work-like work extraction via a single-shot analysis},}\ }\href {\doibase
  10.1038/ncomms2712} {\bibfield  {journal} {\bibinfo  {journal} {Nature
  Communications}\ }\textbf {\bibinfo {volume} {4}},\ \bibinfo {pages} {1925}
  (\bibinfo {year} {2013})}\BibitemShut {NoStop}%
\bibitem [{\citenamefont {Bera}\ \emph {et~al.}(2017)\citenamefont {Bera},
  \citenamefont {Riera}, \citenamefont {Lewenstein},\ and\ \citenamefont
  {Winter}}]{brl17}%
  \BibitemOpen
  \bibfield  {author} {\bibinfo {author} {\bibfnamefont {Manabendra~Nath}\
  \bibnamefont {Bera}}, \bibinfo {author} {\bibfnamefont {Arnau}\ \bibnamefont
  {Riera}}, \bibinfo {author} {\bibfnamefont {Maciej}\ \bibnamefont
  {Lewenstein}}, \ and\ \bibinfo {author} {\bibfnamefont {Andreas}\
  \bibnamefont {Winter}},\ }\bibfield  {title} {\enquote {\bibinfo {title}
  {Generalized laws of thermodynamics in the presence of correlations},}\
  }\href {\doibase 10.1038/s41467-017-02370-x} {\bibfield  {journal} {\bibinfo
  {journal} {Nature Communications}\ }\textbf {\bibinfo {volume} {8}},\
  \bibinfo {pages} {2180} (\bibinfo {year} {2017})}\BibitemShut {NoStop}%
\bibitem [{\citenamefont {{Nath Bera}}\ \emph {et~al.}(2018)\citenamefont
  {{Nath Bera}}, \citenamefont {{Riera}}, \citenamefont {{Lewenstein}},
  \citenamefont {{Baghali Khanian}},\ and\ \citenamefont
  {{Winter}}}]{Bera2017}%
  \BibitemOpen
  \bibfield  {author} {\bibinfo {author} {\bibfnamefont {Manabendra}\
  \bibnamefont {{Nath Bera}}}, \bibinfo {author} {\bibfnamefont {Arnau}\
  \bibnamefont {{Riera}}}, \bibinfo {author} {\bibfnamefont {Maciej}\
  \bibnamefont {{Lewenstein}}}, \bibinfo {author} {\bibfnamefont {Zahra}\
  \bibnamefont {{Baghali Khanian}}}, \ and\ \bibinfo {author} {\bibfnamefont
  {Andreas}\ \bibnamefont {{Winter}}},\ }\bibfield  {title} {\enquote {\bibinfo
  {title} {{Thermodynamics as a Consequence of Information Conservation}},}\
  }\href {\doibase 10.22331/q-2019-02-14-121} {\bibfield  {journal} {\bibinfo
  {journal} {Quantum}\ }\textbf {\bibinfo {volume} {3}},\ \bibinfo {pages}
  {121} (\bibinfo {year} {2018})},\ \bibinfo {note}
  {arXiv[quant-ph]:1707.01750v3}\BibitemShut {NoStop}%
\bibitem [{\citenamefont {{Bera}}\ \emph {et~al.}(2021)\citenamefont {{Bera}},
  \citenamefont {{Lewenstein}},\ and\ \citenamefont {{Bera}}}]{blb19}%
  \BibitemOpen
  \bibfield  {author} {\bibinfo {author} {\bibfnamefont {Mohit~Lal}\
  \bibnamefont {{Bera}}}, \bibinfo {author} {\bibfnamefont {Maciej}\
  \bibnamefont {{Lewenstein}}}, \ and\ \bibinfo {author} {\bibfnamefont
  {Manabendra~Nath}\ \bibnamefont {{Bera}}},\ }\href {\doibase
  10.1038/s41534-021-00366-6} {\enquote {\bibinfo {title} {{Attaining Carnot
  Efficiency with Quantum and Nanoscale Heat Engines}},}\ } (\bibinfo {year}
  {2021}),\ \bibinfo {note} {arXiv[quant-ph]:1911.07003}\BibitemShut {NoStop}%
\bibitem [{\citenamefont {Baumgratz}\ \emph {et~al.}(2014)\citenamefont
  {Baumgratz}, \citenamefont {Cramer},\ and\ \citenamefont {Plenio}}]{bcp14}%
  \BibitemOpen
  \bibfield  {author} {\bibinfo {author} {\bibfnamefont {Tilo}\ \bibnamefont
  {Baumgratz}}, \bibinfo {author} {\bibfnamefont {Marcus}\ \bibnamefont
  {Cramer}}, \ and\ \bibinfo {author} {\bibfnamefont {Martin~B.}\ \bibnamefont
  {Plenio}},\ }\bibfield  {title} {\enquote {\bibinfo {title} {Quantifying
  coherence},}\ }\href {\doibase 10.1103/PhysRevLett.113.140401} {\bibfield
  {journal} {\bibinfo  {journal} {Physical Review Letters}\ }\textbf {\bibinfo
  {volume} {113}},\ \bibinfo {pages} {140401} (\bibinfo {year}
  {2014})}\BibitemShut {NoStop}%
\bibitem [{\citenamefont {Winter}\ and\ \citenamefont {Yang}(2016)}]{WinterRT}%
  \BibitemOpen
  \bibfield  {author} {\bibinfo {author} {\bibfnamefont {Andreas}\ \bibnamefont
  {Winter}}\ and\ \bibinfo {author} {\bibfnamefont {Dong}\ \bibnamefont
  {Yang}},\ }\bibfield  {title} {\enquote {\bibinfo {title} {{Operational
  Resource Theory of Coherence}},}\ }\href {\doibase
  10.1103/PhysRevLett.116.120404} {\bibfield  {journal} {\bibinfo  {journal}
  {Physical Review Letters}\ }\textbf {\bibinfo {volume} {116}},\ \bibinfo
  {pages} {120404} (\bibinfo {year} {2016})}\BibitemShut {NoStop}%
\bibitem [{\citenamefont {Chitambar}\ and\ \citenamefont {Gour}(2016)}]{c&g16}%
  \BibitemOpen
  \bibfield  {author} {\bibinfo {author} {\bibfnamefont {Eric}\ \bibnamefont
  {Chitambar}}\ and\ \bibinfo {author} {\bibfnamefont {Gilad}\ \bibnamefont
  {Gour}},\ }\bibfield  {title} {\enquote {\bibinfo {title} {Critical
  examination of incoherent operations and a physically consistent resource
  theory of quantum coherence},}\ }\href {\doibase
  10.1103/PhysRevLett.117.030401} {\bibfield  {journal} {\bibinfo  {journal}
  {Physical Review Letters}\ }\textbf {\bibinfo {volume} {117}},\ \bibinfo
  {pages} {030401} (\bibinfo {year} {2016})}\BibitemShut {NoStop}%
\bibitem [{\citenamefont {Marvian}\ and\ \citenamefont
  {Spekkens}(2016)}]{m&s16}%
  \BibitemOpen
  \bibfield  {author} {\bibinfo {author} {\bibfnamefont {Iman}\ \bibnamefont
  {Marvian}}\ and\ \bibinfo {author} {\bibfnamefont {Robert~W.}\ \bibnamefont
  {Spekkens}},\ }\bibfield  {title} {\enquote {\bibinfo {title} {How to
  quantify coherence: Distinguishing speakable and unspeakable notions},}\
  }\href {\doibase 10.1103/PhysRevA.94.052324} {\bibfield  {journal} {\bibinfo
  {journal} {Physical Review A}\ }\textbf {\bibinfo {volume} {94}},\ \bibinfo
  {pages} {052324} (\bibinfo {year} {2016})}\BibitemShut {NoStop}%
\bibitem [{\citenamefont {de~Vicente}\ and\ \citenamefont
  {Streltsov}(2016)}]{V&S17}%
  \BibitemOpen
  \bibfield  {author} {\bibinfo {author} {\bibfnamefont {Julio~I.}\
  \bibnamefont {de~Vicente}}\ and\ \bibinfo {author} {\bibfnamefont
  {Alexander}\ \bibnamefont {Streltsov}},\ }\bibfield  {title} {\enquote
  {\bibinfo {title} {Genuine quantum coherence},}\ }\href {\doibase
  10.1088/1751-8121/50/4/045301} {\bibfield  {journal} {\bibinfo  {journal}
  {Journal of Physics A: Mathematical and Theoretical}\ }\textbf {\bibinfo
  {volume} {50}},\ \bibinfo {pages} {045301} (\bibinfo {year}
  {2016})}\BibitemShut {NoStop}%
\bibitem [{\citenamefont {Marvian}\ \emph {et~al.}(2016)\citenamefont
  {Marvian}, \citenamefont {Spekkens},\ and\ \citenamefont {Zanardi}}]{msz16}%
  \BibitemOpen
  \bibfield  {author} {\bibinfo {author} {\bibfnamefont {Iman}\ \bibnamefont
  {Marvian}}, \bibinfo {author} {\bibfnamefont {Robert~W.}\ \bibnamefont
  {Spekkens}}, \ and\ \bibinfo {author} {\bibfnamefont {Paolo}\ \bibnamefont
  {Zanardi}},\ }\bibfield  {title} {\enquote {\bibinfo {title} {Quantum speed
  limits, coherence, and asymmetry},}\ }\href {\doibase
  10.1103/PhysRevA.93.052331} {\bibfield  {journal} {\bibinfo  {journal}
  {Physical Review A}\ }\textbf {\bibinfo {volume} {93}},\ \bibinfo {pages}
  {052331} (\bibinfo {year} {2016})}\BibitemShut {NoStop}%
\bibitem [{\citenamefont {Streltsov}\ \emph
  {et~al.}(2017{\natexlab{a}})\citenamefont {Streltsov}, \citenamefont
  {Adesso},\ and\ \citenamefont {Plenio}}]{sap16}%
  \BibitemOpen
  \bibfield  {author} {\bibinfo {author} {\bibfnamefont {Alexander}\
  \bibnamefont {Streltsov}}, \bibinfo {author} {\bibfnamefont {Gerardo}\
  \bibnamefont {Adesso}}, \ and\ \bibinfo {author} {\bibfnamefont {Martin~B.}\
  \bibnamefont {Plenio}},\ }\bibfield  {title} {\enquote {\bibinfo {title}
  {Colloquium: Quantum coherence as a resource},}\ }\href {\doibase
  10.1103/RevModPhys.89.041003} {\bibfield  {journal} {\bibinfo  {journal}
  {Reviews in Modern Physics}\ }\textbf {\bibinfo {volume} {89}},\ \bibinfo
  {pages} {041003} (\bibinfo {year} {2017}{\natexlab{a}})}\BibitemShut
  {NoStop}%
\bibitem [{\citenamefont {Streltsov}\ \emph
  {et~al.}(2017{\natexlab{b}})\citenamefont {Streltsov}, \citenamefont {Rana},
  \citenamefont {Bera},\ and\ \citenamefont {Lewenstein}}]{srb17}%
  \BibitemOpen
  \bibfield  {author} {\bibinfo {author} {\bibfnamefont {Alexander}\
  \bibnamefont {Streltsov}}, \bibinfo {author} {\bibfnamefont {Swapan}\
  \bibnamefont {Rana}}, \bibinfo {author} {\bibfnamefont {Manabendra~Nath}\
  \bibnamefont {Bera}}, \ and\ \bibinfo {author} {\bibfnamefont {Maciej}\
  \bibnamefont {Lewenstein}},\ }\bibfield  {title} {\enquote {\bibinfo {title}
  {Towards resource theory of coherence in distributed scenarios},}\ }\href
  {\doibase 10.1103/PhysRevX.7.011024} {\bibfield  {journal} {\bibinfo
  {journal} {Physical Review X}\ }\textbf {\bibinfo {volume} {7}},\ \bibinfo
  {pages} {011024} (\bibinfo {year} {2017}{\natexlab{b}})}\BibitemShut
  {NoStop}%
\bibitem [{\citenamefont {Gour}\ and\ \citenamefont {Winter}(2019)}]{g&w19}%
  \BibitemOpen
  \bibfield  {author} {\bibinfo {author} {\bibfnamefont {Gilad}\ \bibnamefont
  {Gour}}\ and\ \bibinfo {author} {\bibfnamefont {Andreas}\ \bibnamefont
  {Winter}},\ }\bibfield  {title} {\enquote {\bibinfo {title} {How to quantify
  a dynamical quantum resource},}\ }\href {\doibase
  10.1103/PhysRevLett.123.150401} {\bibfield  {journal} {\bibinfo  {journal}
  {Physical Review Letters}\ }\textbf {\bibinfo {volume} {123}},\ \bibinfo
  {pages} {150401} (\bibinfo {year} {2019})}\BibitemShut {NoStop}%
\bibitem [{\citenamefont {Contreras-Tejada}\ \emph {et~al.}(2019)\citenamefont
  {Contreras-Tejada}, \citenamefont {Palazuelos},\ and\ \citenamefont
  {de~Vicente}}]{cpv18}%
  \BibitemOpen
  \bibfield  {author} {\bibinfo {author} {\bibfnamefont {Patricia}\
  \bibnamefont {Contreras-Tejada}}, \bibinfo {author} {\bibfnamefont {Carlos}\
  \bibnamefont {Palazuelos}}, \ and\ \bibinfo {author} {\bibfnamefont
  {Julio~I.}\ \bibnamefont {de~Vicente}},\ }\bibfield  {title} {\enquote
  {\bibinfo {title} {Resource theory of entanglement with a unique multipartite
  maximally entangled state},}\ }\href {\doibase
  10.1103/PhysRevLett.122.120503} {\bibfield  {journal} {\bibinfo  {journal}
  {Physical Review Letters}\ }\textbf {\bibinfo {volume} {122}},\ \bibinfo
  {pages} {120503} (\bibinfo {year} {2019})}\BibitemShut {NoStop}%
\bibitem [{\citenamefont {Shahandeh}(2019)}]{sha19}%
  \BibitemOpen
  \bibfield  {author} {\bibinfo {author} {\bibfnamefont {Farid}\ \bibnamefont
  {Shahandeh}},\ }\href {\doibase 10.1007/978-3-030-24120-9} {\emph {\bibinfo
  {title} {Quantum Correlations: A Modern Augmentation}}},\ Springer Theses\
  (\bibinfo  {publisher} {Springer Verlag},\ \bibinfo {year}
  {2019})\BibitemShut {NoStop}%
\bibitem [{\citenamefont {de~Vicente}(2014)}]{Vic14}%
  \BibitemOpen
  \bibfield  {author} {\bibinfo {author} {\bibfnamefont {Julio~I.}\
  \bibnamefont {de~Vicente}},\ }\bibfield  {title} {\enquote {\bibinfo {title}
  {On nonlocality as a resource theory and nonlocality measures},}\ }\href
  {\doibase 10.1088/1751-8113/47/42/424017} {\bibfield  {journal} {\bibinfo
  {journal} {Journal of Physics A: Mathematical and Theoretical}\ }\textbf
  {\bibinfo {volume} {47}},\ \bibinfo {pages} {424017} (\bibinfo {year}
  {2014})}\BibitemShut {NoStop}%
\bibitem [{\citenamefont {Duarte}\ and\ \citenamefont {Amaral}(2018)}]{d&a18}%
  \BibitemOpen
  \bibfield  {author} {\bibinfo {author} {\bibfnamefont {Cristhiano}\
  \bibnamefont {Duarte}}\ and\ \bibinfo {author} {\bibfnamefont {Barbara}\
  \bibnamefont {Amaral}},\ }\bibfield  {title} {\enquote {\bibinfo {title}
  {Resource theory of contextuality for arbitrary prepare-and-measure
  experiments},}\ }\href {\doibase 10.1063/1.5018582} {\bibfield  {journal}
  {\bibinfo  {journal} {Journal of Mathematical Physics}\ }\textbf {\bibinfo
  {volume} {59}},\ \bibinfo {pages} {062202} (\bibinfo {year}
  {2018})}\BibitemShut {NoStop}%
\bibitem [{\citenamefont {Yunger~Halpern}\ and\ \citenamefont
  {Renes}(2016)}]{thermo_multiple_charge1}%
  \BibitemOpen
  \bibfield  {author} {\bibinfo {author} {\bibfnamefont {Nicole}\ \bibnamefont
  {Yunger~Halpern}}\ and\ \bibinfo {author} {\bibfnamefont {Joseph~M.}\
  \bibnamefont {Renes}},\ }\bibfield  {title} {\enquote {\bibinfo {title}
  {Beyond heat baths: {Generalized} resource theories for small-scale
  thermodynamics},}\ }\href {\doibase 10.1103/PhysRevE.93.022126} {\bibfield
  {journal} {\bibinfo  {journal} {Physical Review E}\ }\textbf {\bibinfo
  {volume} {93}},\ \bibinfo {pages} {022126} (\bibinfo {year}
  {2016})}\BibitemShut {NoStop}%
\bibitem [{\citenamefont {Yunger~Halpern}(2018)}]{thermo_multiple_charge2}%
  \BibitemOpen
  \bibfield  {author} {\bibinfo {author} {\bibfnamefont {Nicole}\ \bibnamefont
  {Yunger~Halpern}},\ }\bibfield  {title} {\enquote {\bibinfo {title} {Beyond
  heat baths {II}: framework for generalized thermodynamic resource
  theories},}\ }\href {\doibase 10.1088/1751-8121/aaa62f} {\bibfield  {journal}
  {\bibinfo  {journal} {Journal of Physics A}\ }\textbf {\bibinfo {volume}
  {51}},\ \bibinfo {pages} {094001} (\bibinfo {year} {2018})}\BibitemShut
  {NoStop}%
\bibitem [{\citenamefont {Gour}\ \emph {et~al.}(2018)\citenamefont {Gour},
  \citenamefont {Jennings}, \citenamefont {Buscemi}, \citenamefont {Duan},\
  and\ \citenamefont {Marvian}}]{thermo_multiple_charge3}%
  \BibitemOpen
  \bibfield  {author} {\bibinfo {author} {\bibfnamefont {Gilad}\ \bibnamefont
  {Gour}}, \bibinfo {author} {\bibfnamefont {David}\ \bibnamefont {Jennings}},
  \bibinfo {author} {\bibfnamefont {Francesco}\ \bibnamefont {Buscemi}},
  \bibinfo {author} {\bibfnamefont {Runyao}\ \bibnamefont {Duan}}, \ and\
  \bibinfo {author} {\bibfnamefont {Iman}\ \bibnamefont {Marvian}},\ }\bibfield
   {title} {\enquote {\bibinfo {title} {Quantum majorization and a complete set
  of entropic conditions for quantum thermodynamics},}\ }\href {\doibase
  https://doi.org/10.1038/s41467-018-06261-7} {\bibfield  {journal} {\bibinfo
  {journal} {Nature Communications}\ }\textbf {\bibinfo {volume} {9}},\
  \bibinfo {pages} {5352} (\bibinfo {year} {2018})}\BibitemShut {NoStop}%
\bibitem [{\citenamefont {Devetak}\ \emph {et~al.}(2008)\citenamefont
  {Devetak}, \citenamefont {Harrow},\ and\ \citenamefont
  {Winter}}]{DevetakHarrowWinter:Shannon}%
  \BibitemOpen
  \bibfield  {author} {\bibinfo {author} {\bibfnamefont {Igor}\ \bibnamefont
  {Devetak}}, \bibinfo {author} {\bibfnamefont {Aram~W.}\ \bibnamefont
  {Harrow}}, \ and\ \bibinfo {author} {\bibfnamefont {Andreas}\ \bibnamefont
  {Winter}},\ }\bibfield  {title} {\enquote {\bibinfo {title} {{A Resource
  Framework for Quantum Shannon Theory}},}\ }\href {\doibase
  10.1109/TIT.2008.928980} {\bibfield  {journal} {\bibinfo  {journal} {IEEE
  Transactions on Information Theory}\ }\textbf {\bibinfo {volume} {54}},\
  \bibinfo {pages} {4587--4618} (\bibinfo {year} {2008})}\BibitemShut {NoStop}%
\bibitem [{\citenamefont {{Liu}}\ and\ \citenamefont {{Winter}}(2019)}]{l&w19}%
  \BibitemOpen
  \bibfield  {author} {\bibinfo {author} {\bibfnamefont {Zi-Wen}\ \bibnamefont
  {{Liu}}}\ and\ \bibinfo {author} {\bibfnamefont {Andreas}\ \bibnamefont
  {{Winter}}},\ }\href {https://arxiv.org/abs/1904.04201} {\enquote {\bibinfo
  {title} {{Resource theories of quantum channels and the universal role of
  resource erasure}},}\ } (\bibinfo {year} {2019}),\ \bibinfo {note}
  {arXiv[quant-ph]:1904.04201}\BibitemShut {NoStop}%
\bibitem [{\citenamefont {Lostaglio}\ \emph {et~al.}(2017)\citenamefont
  {Lostaglio}, \citenamefont {Jennings},\ and\ \citenamefont
  {Rudolph}}]{Lostaglio2017}%
  \BibitemOpen
  \bibfield  {author} {\bibinfo {author} {\bibfnamefont {Matteo}\ \bibnamefont
  {Lostaglio}}, \bibinfo {author} {\bibfnamefont {David}\ \bibnamefont
  {Jennings}}, \ and\ \bibinfo {author} {\bibfnamefont {Terry}\ \bibnamefont
  {Rudolph}},\ }\bibfield  {title} {\enquote {\bibinfo {title} {Thermodynamic
  resource theories, non-commutativity and maximum entropy principles},}\
  }\href {\doibase 10.1088/1367-2630/aa617f} {\bibfield  {journal} {\bibinfo
  {journal} {New Journal of Physics}\ }\textbf {\bibinfo {volume} {19}},\
  \bibinfo {pages} {043008} (\bibinfo {year} {2017})}\BibitemShut {NoStop}%
\bibitem [{\citenamefont {{Sparaciari}}\ \emph {et~al.}(2017)\citenamefont
  {{Sparaciari}}, \citenamefont {{Oppenheim}},\ and\ \citenamefont
  {{Fritz}}}]{Sparaciari2016}%
  \BibitemOpen
  \bibfield  {author} {\bibinfo {author} {\bibfnamefont {Carlo}\ \bibnamefont
  {{Sparaciari}}}, \bibinfo {author} {\bibfnamefont {Jonathan}\ \bibnamefont
  {{Oppenheim}}}, \ and\ \bibinfo {author} {\bibfnamefont {Tobias}\
  \bibnamefont {{Fritz}}},\ }\bibfield  {title} {\enquote {\bibinfo {title} {{A
  Resource Theory for Work and Heat}},}\ }\href {\doibase
  10.1103/PhysRevA.96.052112} {\bibfield  {journal} {\bibinfo  {journal}
  {Physical Review A}\ }\textbf {\bibinfo {volume} {96}},\ \bibinfo {pages}
  {052112} (\bibinfo {year} {2017})},\ \bibinfo {note}
  {arXiv[quant-ph]:1607.01302}\BibitemShut {NoStop}%
\bibitem [{\citenamefont {{Yunger Halpern}}\ \emph {et~al.}(2016)\citenamefont
  {{Yunger Halpern}}, \citenamefont {{Faist}}, \citenamefont {{Oppenheim}},\
  and\ \citenamefont {{Winter}}}]{Halpern2016}%
  \BibitemOpen
  \bibfield  {author} {\bibinfo {author} {\bibfnamefont {Nicole}\ \bibnamefont
  {{Yunger Halpern}}}, \bibinfo {author} {\bibfnamefont {Philippe}\
  \bibnamefont {{Faist}}}, \bibinfo {author} {\bibfnamefont {Jonathan}\
  \bibnamefont {{Oppenheim}}}, \ and\ \bibinfo {author} {\bibfnamefont
  {Andreas}\ \bibnamefont {{Winter}}},\ }\bibfield  {title} {\enquote {\bibinfo
  {title} {Microcanonical and resource-theoretic derivations of the thermal
  state of a quantum system with noncommuting charges},}\ }\href {\doibase
  10.1038/ncomms12051} {\bibfield  {journal} {\bibinfo  {journal} {Nature
  Communications}\ }\textbf {\bibinfo {volume} {7}},\ \bibinfo {pages} {12051}
  (\bibinfo {year} {2016})},\ \bibinfo {note}
  {arXiv[quant-ph]:1512.01189}\BibitemShut {NoStop}%
\bibitem [{\citenamefont {{Liu}}(2007)}]{Liu2007}%
  \BibitemOpen
  \bibfield  {author} {\bibinfo {author} {\bibfnamefont {Yi-Kai}\ \bibnamefont
  {{Liu}}},\ }\emph {\bibinfo {title} {{The Complexity of the Consistency and
  N-Representability Problems for Quantum States}}},\ \href
  {https://arxiv.org/abs/0712.3041} {Ph.D. thesis},\ \bibinfo  {school}
  {Department of Computer Science, University of California, San Diego}
  (\bibinfo {year} {2007}),\ \bibinfo {note}
  {arXiv[quant-ph]:0712.3041}\BibitemShut {NoStop}%
\bibitem [{\citenamefont {{Guryanova}}\ \emph {et~al.}(2016)\citenamefont
  {{Guryanova}}, \citenamefont {{Popescu}}, \citenamefont {{Short}},
  \citenamefont {{Silva}},\ and\ \citenamefont {{Skrzypczyk}}}]{Guryanova2016}%
  \BibitemOpen
  \bibfield  {author} {\bibinfo {author} {\bibfnamefont {Yelena}\ \bibnamefont
  {{Guryanova}}}, \bibinfo {author} {\bibfnamefont {Sandu}\ \bibnamefont
  {{Popescu}}}, \bibinfo {author} {\bibfnamefont {Anthony~J.}\ \bibnamefont
  {{Short}}}, \bibinfo {author} {\bibfnamefont {Ralph}\ \bibnamefont
  {{Silva}}}, \ and\ \bibinfo {author} {\bibfnamefont {Paul}\ \bibnamefont
  {{Skrzypczyk}}},\ }\bibfield  {title} {\enquote {\bibinfo {title}
  {{Thermodynamics of quantum systems with multiple conserved quantities}},}\
  }\href {\doibase 10.1038/ncomms12049} {\bibfield  {journal} {\bibinfo
  {journal} {Nature Communications}\ }\textbf {\bibinfo {volume} {7}},\
  \bibinfo {pages} {12049} (\bibinfo {year} {2016})},\ \bibinfo {note}
  {arXiv[quant-ph]:1512.01190}\BibitemShut {NoStop}%
\bibitem [{\citenamefont {{Popescu}}\ \emph {et~al.}(2018)\citenamefont
  {{Popescu}}, \citenamefont {{Sainz}}, \citenamefont {{Short}},\ and\
  \citenamefont {{Winter}}}]{Popescu2018}%
  \BibitemOpen
  \bibfield  {author} {\bibinfo {author} {\bibfnamefont {Sandu}\ \bibnamefont
  {{Popescu}}}, \bibinfo {author} {\bibfnamefont {Ana~Belen}\ \bibnamefont
  {{Sainz}}}, \bibinfo {author} {\bibfnamefont {Anthony~J.}\ \bibnamefont
  {{Short}}}, \ and\ \bibinfo {author} {\bibfnamefont {Andreas}\ \bibnamefont
  {{Winter}}},\ }\bibfield  {title} {\enquote {\bibinfo {title} {{Quantum
  Reference Frames and Their Applications to Thermodynamics}},}\ }\href
  {\doibase 10.1098/rsta.2018.0111} {\bibfield  {journal} {\bibinfo  {journal}
  {Philosophical Transactions of the Royal Society A}\ }\textbf {\bibinfo
  {volume} {376}},\ \bibinfo {pages} {20180111} (\bibinfo {year} {2018})},\
  \bibinfo {note} {arXiv[quant-ph]:1804.03730}\BibitemShut {NoStop}%
\bibitem [{\citenamefont {Popescu}\ \emph {et~al.}(2020)\citenamefont
  {Popescu}, \citenamefont {Sainz}, \citenamefont {Short},\ and\ \citenamefont
  {Winter}}]{Popescu2019}%
  \BibitemOpen
  \bibfield  {author} {\bibinfo {author} {\bibfnamefont {Sandu}\ \bibnamefont
  {Popescu}}, \bibinfo {author} {\bibfnamefont {{Ana Belen}}\ \bibnamefont
  {Sainz}}, \bibinfo {author} {\bibfnamefont {Anthony~J.}\ \bibnamefont
  {Short}}, \ and\ \bibinfo {author} {\bibfnamefont {Andreas}\ \bibnamefont
  {Winter}},\ }\bibfield  {title} {\enquote {\bibinfo {title} {Reference frames
  which separately store non-commuting conserved quantities},}\ }\href
  {\doibase 10.1103/PhysRevLett.125.090601} {\bibfield  {journal} {\bibinfo
  {journal} {Physical Review Letters}\ }\textbf {\bibinfo {volume} {125}},\
  \bibinfo {pages} {090601} (\bibinfo {year} {2020})},\ \bibinfo {note}
  {arXiv[quant-ph]:1908.02713}\BibitemShut {NoStop}%
\bibitem [{\citenamefont {Fannes}(1973)}]{Fannes1973}%
  \BibitemOpen
  \bibfield  {author} {\bibinfo {author} {\bibfnamefont {Mark}\ \bibnamefont
  {Fannes}},\ }\bibfield  {title} {\enquote {\bibinfo {title} {A continuity
  property of the entropy density for spin lattice systems},}\ }\href {\doibase
  10.1007/BF01646490} {\bibfield  {journal} {\bibinfo  {journal}
  {Communications in Mathematical Physics}\ }\textbf {\bibinfo {volume} {31}},\
  \bibinfo {pages} {291--294} (\bibinfo {year} {1973})}\BibitemShut {NoStop}%
\bibitem [{\citenamefont {Audenaert}(2007)}]{Audenaert2007}%
  \BibitemOpen
  \bibfield  {author} {\bibinfo {author} {\bibfnamefont {Koenraad M.~R.}\
  \bibnamefont {Audenaert}},\ }\bibfield  {title} {\enquote {\bibinfo {title}
  {A sharp continuity estimate for the von {N}eumann entropy},}\ }\href
  {\doibase 10.1088/1751-8113/40/28/S18} {\bibfield  {journal} {\bibinfo
  {journal} {Journal of Physics A: Mathematical and Theoretical}\ }\textbf
  {\bibinfo {volume} {40}},\ \bibinfo {pages} {8127--8131} (\bibinfo {year}
  {2007})}\BibitemShut {NoStop}%
\bibitem [{\citenamefont {Hayden}\ \emph {et~al.}(2006)\citenamefont {Hayden},
  \citenamefont {Leung},\ and\ \citenamefont {Winter}}]{Hayden2006}%
  \BibitemOpen
  \bibfield  {author} {\bibinfo {author} {\bibfnamefont {Patrick}\ \bibnamefont
  {Hayden}}, \bibinfo {author} {\bibfnamefont {Debbie~W.}\ \bibnamefont
  {Leung}}, \ and\ \bibinfo {author} {\bibfnamefont {Andreas}\ \bibnamefont
  {Winter}},\ }\bibfield  {title} {\enquote {\bibinfo {title} {Aspects of
  generic entanglement},}\ }\href {\doibase 10.1007/s00220-006-1535-6}
  {\bibfield  {journal} {\bibinfo  {journal} {Communications in Mathematical
  Physics}\ }\textbf {\bibinfo {volume} {265}},\ \bibinfo {pages} {95--117}
  (\bibinfo {year} {2006})}\BibitemShut {NoStop}%
\bibitem [{\citenamefont {Horodecki}\ \emph {et~al.}(2005)\citenamefont
  {Horodecki}, \citenamefont {Oppenheim},\ and\ \citenamefont
  {Winter}}]{HOW:negative-Nature}%
  \BibitemOpen
  \bibfield  {author} {\bibinfo {author} {\bibfnamefont {Micha{\l}}\
  \bibnamefont {Horodecki}}, \bibinfo {author} {\bibfnamefont {Jonathan}\
  \bibnamefont {Oppenheim}}, \ and\ \bibinfo {author} {\bibfnamefont {Andreas}\
  \bibnamefont {Winter}},\ }\bibfield  {title} {\enquote {\bibinfo {title}
  {Partial quantum information},}\ }\href {\doibase 10.1038/nature03909}
  {\bibfield  {journal} {\bibinfo  {journal} {Nature}\ }\textbf {\bibinfo
  {volume} {436}},\ \bibinfo {pages} {673--676} (\bibinfo {year}
  {2005})}\BibitemShut {NoStop}%
\bibitem [{\citenamefont {Horodecki}\ \emph {et~al.}(2007)\citenamefont
  {Horodecki}, \citenamefont {Oppenheim},\ and\ \citenamefont
  {Winter}}]{HOW:negative-CMP}%
  \BibitemOpen
  \bibfield  {author} {\bibinfo {author} {\bibfnamefont {Micha{\l}}\
  \bibnamefont {Horodecki}}, \bibinfo {author} {\bibfnamefont {Jonathan}\
  \bibnamefont {Oppenheim}}, \ and\ \bibinfo {author} {\bibfnamefont {Andreas}\
  \bibnamefont {Winter}},\ }\bibfield  {title} {\enquote {\bibinfo {title}
  {Quantum state merging and negative information},}\ }\href {\doibase
  10.1007/s00220-006-0118-x} {\bibfield  {journal} {\bibinfo  {journal}
  {Communications in Mathematical Physics}\ }\textbf {\bibinfo {volume}
  {269}},\ \bibinfo {pages} {107--136} (\bibinfo {year} {2007})}\BibitemShut
  {NoStop}%
\bibitem [{\citenamefont {Hayden}\ \emph {et~al.}(2004)\citenamefont {Hayden},
  \citenamefont {Jozsa}, \citenamefont {Petz},\ and\ \citenamefont
  {Winter}}]{SSA-eq}%
  \BibitemOpen
  \bibfield  {author} {\bibinfo {author} {\bibfnamefont {Patrick}\ \bibnamefont
  {Hayden}}, \bibinfo {author} {\bibfnamefont {Richard}\ \bibnamefont {Jozsa}},
  \bibinfo {author} {\bibfnamefont {D{\'e}nes}\ \bibnamefont {Petz}}, \ and\
  \bibinfo {author} {\bibfnamefont {Andreas}\ \bibnamefont {Winter}},\
  }\bibfield  {title} {\enquote {\bibinfo {title} {Structure of states which
  satisfy strong subadditivity of quantum entropy with equality},}\ }\href
  {\doibase 10.1007/s00220-004-1049-z} {\bibfield  {journal} {\bibinfo
  {journal} {Communications in Mathematical Physics}\ }\textbf {\bibinfo
  {volume} {246}},\ \bibinfo {pages} {359--374} (\bibinfo {year}
  {2004})}\BibitemShut {NoStop}%
\bibitem [{\citenamefont {Horodecki}\ and\ \citenamefont
  {Horodecki}(1994)}]{HoroHoro1994}%
  \BibitemOpen
  \bibfield  {author} {\bibinfo {author} {\bibfnamefont {Ryszard}\ \bibnamefont
  {Horodecki}}\ and\ \bibinfo {author} {\bibfnamefont {Pawel}\ \bibnamefont
  {Horodecki}},\ }\bibfield  {title} {\enquote {\bibinfo {title} {Quantum
  redundancies and local realism},}\ }\href {\doibase
  10.1016/0375-9601(94)91275-0} {\bibfield  {journal} {\bibinfo  {journal}
  {Physics Letters A}\ }\textbf {\bibinfo {volume} {194}},\ \bibinfo {pages}
  {147--152} (\bibinfo {year} {1994})}\BibitemShut {NoStop}%
\bibitem [{\citenamefont {Nielsen}\ and\ \citenamefont
  {Kempe}(2001)}]{NielsenKempe}%
  \BibitemOpen
  \bibfield  {author} {\bibinfo {author} {\bibfnamefont {Michael~A.}\
  \bibnamefont {Nielsen}}\ and\ \bibinfo {author} {\bibfnamefont {Julia}\
  \bibnamefont {Kempe}},\ }\bibfield  {title} {\enquote {\bibinfo {title}
  {{Separable States Are More Disordered Globally than Locally}},}\ }\href
  {\doibase 10.1103/PhysRevLett.86.5184} {\bibfield  {journal} {\bibinfo
  {journal} {Physical Review Letters}\ }\textbf {\bibinfo {volume} {86}},\
  \bibinfo {pages} {5184--5187} (\bibinfo {year} {2001})},\ \bibinfo {note}
  {arXiv:quant-ph/0011117}\BibitemShut {NoStop}%
\bibitem [{\citenamefont {Arute}\ \emph {et~al.}(2019)\citenamefont {Arute},
  \citenamefont {Arya}, \citenamefont {Babbush}, \citenamefont {Bacon},
  \citenamefont {Bardin}, \citenamefont {Barends}, \citenamefont {Biswas},
  \citenamefont {Boixo}, \citenamefont {Brandao}, \citenamefont {Buell},
  \citenamefont {Burkett}, \citenamefont {Chen}, \citenamefont {Chen},
  \citenamefont {Chiaro}, \citenamefont {Collins}, \citenamefont {Courtney},
  \citenamefont {Dunsworth}, \citenamefont {Farhi}, \citenamefont {Foxen},
  \citenamefont {Fowler}, \citenamefont {Gidney}, \citenamefont {Giustina},
  \citenamefont {Graff}, \citenamefont {Guerin}, \citenamefont {Habegger},
  \citenamefont {Harrigan}, \citenamefont {Hartmann}, \citenamefont {Ho},
  \citenamefont {Hoffmann}, \citenamefont {Huang}, \citenamefont {Humble},
  \citenamefont {Isakov}, \citenamefont {Jeffrey}, \citenamefont {Jiang},
  \citenamefont {Kafri}, \citenamefont {Kechedzhi}, \citenamefont {Kelly},
  \citenamefont {Klimov}, \citenamefont {Knysh}, \citenamefont {Korotkov},
  \citenamefont {Kostritsa}, \citenamefont {Landhuis}, \citenamefont
  {Lindmark}, \citenamefont {Lucero}, \citenamefont {Lyakh}, \citenamefont
  {Mandr{\`a}}, \citenamefont {McClean}, \citenamefont {McEwen}, \citenamefont
  {Megrant}, \citenamefont {Mi}, \citenamefont {Michielsen}, \citenamefont
  {Mohseni}, \citenamefont {Mutus}, \citenamefont {Naaman}, \citenamefont
  {Neeley}, \citenamefont {Neill}, \citenamefont {Niu}, \citenamefont {Ostby},
  \citenamefont {Petukhov}, \citenamefont {Platt}, \citenamefont {Quintana},
  \citenamefont {Rieffel}, \citenamefont {Roushan}, \citenamefont {Rubin},
  \citenamefont {Sank}, \citenamefont {Satzinger}, \citenamefont {Smelyanskiy},
  \citenamefont {Sung}, \citenamefont {Trevithick}, \citenamefont
  {Vainsencher}, \citenamefont {Villalonga}, \citenamefont {White},
  \citenamefont {Yao}, \citenamefont {Yeh}, \citenamefont {Zalcman},
  \citenamefont {Neven},\ and\ \citenamefont {Martinis}}]{Google}%
  \BibitemOpen
  \bibfield  {author} {\bibinfo {author} {\bibfnamefont {Frank}\ \bibnamefont
  {Arute}}, \bibinfo {author} {\bibfnamefont {Kunal}\ \bibnamefont {Arya}},
  \bibinfo {author} {\bibfnamefont {Ryan}\ \bibnamefont {Babbush}}, \bibinfo
  {author} {\bibfnamefont {Dave}\ \bibnamefont {Bacon}}, \bibinfo {author}
  {\bibfnamefont {Joseph~C.}\ \bibnamefont {Bardin}}, \bibinfo {author}
  {\bibfnamefont {Rami}\ \bibnamefont {Barends}}, \bibinfo {author}
  {\bibfnamefont {Rupak}\ \bibnamefont {Biswas}}, \bibinfo {author}
  {\bibfnamefont {Sergio}\ \bibnamefont {Boixo}}, \bibinfo {author}
  {\bibfnamefont {Fernando G. S.~L.}\ \bibnamefont {Brandao}}, \bibinfo
  {author} {\bibfnamefont {David~A.}\ \bibnamefont {Buell}}, \bibinfo {author}
  {\bibfnamefont {Brian}\ \bibnamefont {Burkett}}, \bibinfo {author}
  {\bibfnamefont {Yu}~\bibnamefont {Chen}}, \bibinfo {author} {\bibfnamefont
  {Zijun}\ \bibnamefont {Chen}}, \bibinfo {author} {\bibfnamefont {Ben}\
  \bibnamefont {Chiaro}}, \bibinfo {author} {\bibfnamefont {Roberto}\
  \bibnamefont {Collins}}, \bibinfo {author} {\bibfnamefont {William}\
  \bibnamefont {Courtney}}, \bibinfo {author} {\bibfnamefont {Andrew}\
  \bibnamefont {Dunsworth}}, \bibinfo {author} {\bibfnamefont {Edward}\
  \bibnamefont {Farhi}}, \bibinfo {author} {\bibfnamefont {Brooks}\
  \bibnamefont {Foxen}}, \bibinfo {author} {\bibfnamefont {Austin}\
  \bibnamefont {Fowler}}, \bibinfo {author} {\bibfnamefont {Craig}\
  \bibnamefont {Gidney}}, \bibinfo {author} {\bibfnamefont {Marissa}\
  \bibnamefont {Giustina}}, \bibinfo {author} {\bibfnamefont {Rob}\
  \bibnamefont {Graff}}, \bibinfo {author} {\bibfnamefont {Keith}\ \bibnamefont
  {Guerin}}, \bibinfo {author} {\bibfnamefont {Steve}\ \bibnamefont
  {Habegger}}, \bibinfo {author} {\bibfnamefont {Matthew~P.}\ \bibnamefont
  {Harrigan}}, \bibinfo {author} {\bibfnamefont {Michael~J.}\ \bibnamefont
  {Hartmann}}, \bibinfo {author} {\bibfnamefont {Alan}\ \bibnamefont {Ho}},
  \bibinfo {author} {\bibfnamefont {Markus}\ \bibnamefont {Hoffmann}}, \bibinfo
  {author} {\bibfnamefont {Trent}\ \bibnamefont {Huang}}, \bibinfo {author}
  {\bibfnamefont {Travis~S.}\ \bibnamefont {Humble}}, \bibinfo {author}
  {\bibfnamefont {Sergei~V.}\ \bibnamefont {Isakov}}, \bibinfo {author}
  {\bibfnamefont {Evan}\ \bibnamefont {Jeffrey}}, \bibinfo {author}
  {\bibfnamefont {Zhang}\ \bibnamefont {Jiang}}, \bibinfo {author}
  {\bibfnamefont {Dvir}\ \bibnamefont {Kafri}}, \bibinfo {author}
  {\bibfnamefont {Kostyantyn}\ \bibnamefont {Kechedzhi}}, \bibinfo {author}
  {\bibfnamefont {Julian}\ \bibnamefont {Kelly}}, \bibinfo {author}
  {\bibfnamefont {Paul~V.}\ \bibnamefont {Klimov}}, \bibinfo {author}
  {\bibfnamefont {Sergey}\ \bibnamefont {Knysh}}, \bibinfo {author}
  {\bibfnamefont {Alexander}\ \bibnamefont {Korotkov}}, \bibinfo {author}
  {\bibfnamefont {Fedor}\ \bibnamefont {Kostritsa}}, \bibinfo {author}
  {\bibfnamefont {David}\ \bibnamefont {Landhuis}}, \bibinfo {author}
  {\bibfnamefont {Mike}\ \bibnamefont {Lindmark}}, \bibinfo {author}
  {\bibfnamefont {Erik}\ \bibnamefont {Lucero}}, \bibinfo {author}
  {\bibfnamefont {Dmitry}\ \bibnamefont {Lyakh}}, \bibinfo {author}
  {\bibfnamefont {Salvatore}\ \bibnamefont {Mandr{\`a}}}, \bibinfo {author}
  {\bibfnamefont {Jarrod~R.}\ \bibnamefont {McClean}}, \bibinfo {author}
  {\bibfnamefont {Matthew}\ \bibnamefont {McEwen}}, \bibinfo {author}
  {\bibfnamefont {Anthony}\ \bibnamefont {Megrant}}, \bibinfo {author}
  {\bibfnamefont {Xiao}\ \bibnamefont {Mi}}, \bibinfo {author} {\bibfnamefont
  {Kristel}\ \bibnamefont {Michielsen}}, \bibinfo {author} {\bibfnamefont
  {Masoud}\ \bibnamefont {Mohseni}}, \bibinfo {author} {\bibfnamefont {Josh}\
  \bibnamefont {Mutus}}, \bibinfo {author} {\bibfnamefont {Ofer}\ \bibnamefont
  {Naaman}}, \bibinfo {author} {\bibfnamefont {Matthew}\ \bibnamefont
  {Neeley}}, \bibinfo {author} {\bibfnamefont {Charles}\ \bibnamefont {Neill}},
  \bibinfo {author} {\bibfnamefont {Murphy~Yuezhen}\ \bibnamefont {Niu}},
  \bibinfo {author} {\bibfnamefont {Eric}\ \bibnamefont {Ostby}}, \bibinfo
  {author} {\bibfnamefont {Andre}\ \bibnamefont {Petukhov}}, \bibinfo {author}
  {\bibfnamefont {John~C.}\ \bibnamefont {Platt}}, \bibinfo {author}
  {\bibfnamefont {Chris}\ \bibnamefont {Quintana}}, \bibinfo {author}
  {\bibfnamefont {Eleanor~G.}\ \bibnamefont {Rieffel}}, \bibinfo {author}
  {\bibfnamefont {Pedram}\ \bibnamefont {Roushan}}, \bibinfo {author}
  {\bibfnamefont {Nicholas~C.}\ \bibnamefont {Rubin}}, \bibinfo {author}
  {\bibfnamefont {Daniel}\ \bibnamefont {Sank}}, \bibinfo {author}
  {\bibfnamefont {Kevin~J.}\ \bibnamefont {Satzinger}}, \bibinfo {author}
  {\bibfnamefont {Vadim}\ \bibnamefont {Smelyanskiy}}, \bibinfo {author}
  {\bibfnamefont {Kevin~J.}\ \bibnamefont {Sung}}, \bibinfo {author}
  {\bibfnamefont {Matthew~D.}\ \bibnamefont {Trevithick}}, \bibinfo {author}
  {\bibfnamefont {Amit}\ \bibnamefont {Vainsencher}}, \bibinfo {author}
  {\bibfnamefont {Benjamin}\ \bibnamefont {Villalonga}}, \bibinfo {author}
  {\bibfnamefont {Theodore}\ \bibnamefont {White}}, \bibinfo {author}
  {\bibfnamefont {Z.~Jamie}\ \bibnamefont {Yao}}, \bibinfo {author}
  {\bibfnamefont {Ping}\ \bibnamefont {Yeh}}, \bibinfo {author} {\bibfnamefont
  {Adam}\ \bibnamefont {Zalcman}}, \bibinfo {author} {\bibfnamefont {Hartmut}\
  \bibnamefont {Neven}}, \ and\ \bibinfo {author} {\bibfnamefont {John~M.}\
  \bibnamefont {Martinis}},\ }\bibfield  {title} {\enquote {\bibinfo {title}
  {Quantum supremacy using a programmable superconducting processor},}\ }\href
  {\doibase 10.1038/s41586-019-1666-5} {\bibfield  {journal} {\bibinfo
  {journal} {Nature}\ }\textbf {\bibinfo {volume} {574}},\ \bibinfo {pages}
  {505--510} (\bibinfo {year} {2019})}\BibitemShut {NoStop}%
\bibitem [{\citenamefont {King}\ \emph {et~al.}(2019)\citenamefont {King},
  \citenamefont {Yarkoni}, \citenamefont {Raymond}, \citenamefont {Ozfidan},
  \citenamefont {King}, \citenamefont {Nevisi}, \citenamefont {Hilton},\ and\
  \citenamefont {McGeoch}}]{D-Wave}%
  \BibitemOpen
  \bibfield  {author} {\bibinfo {author} {\bibfnamefont {James}\ \bibnamefont
  {King}}, \bibinfo {author} {\bibfnamefont {Sheir}\ \bibnamefont {Yarkoni}},
  \bibinfo {author} {\bibfnamefont {Jack}\ \bibnamefont {Raymond}}, \bibinfo
  {author} {\bibfnamefont {Isil}\ \bibnamefont {Ozfidan}}, \bibinfo {author}
  {\bibfnamefont {Andrew~D.}\ \bibnamefont {King}}, \bibinfo {author}
  {\bibfnamefont {Mayssam~Mohammadi}\ \bibnamefont {Nevisi}}, \bibinfo {author}
  {\bibfnamefont {Jeremy~P.}\ \bibnamefont {Hilton}}, \ and\ \bibinfo {author}
  {\bibfnamefont {Catherine~C.}\ \bibnamefont {McGeoch}},\ }\bibfield  {title}
  {\enquote {\bibinfo {title} {Quantum annealing amid local ruggedness and
  global frustration},}\ }\href {\doibase 10.7566/JPSJ.88.061007} {\bibfield
  {journal} {\bibinfo  {journal} {Journal of the Physical Society of Japan}\
  }\textbf {\bibinfo {volume} {88}},\ \bibinfo {pages} {061007} (\bibinfo
  {year} {2019})}\BibitemShut {NoStop}%
\bibitem [{\citenamefont {Paraoanu}(2014)}]{Parao}%
  \BibitemOpen
  \bibfield  {author} {\bibinfo {author} {\bibfnamefont {Gheorghe~S.}\
  \bibnamefont {Paraoanu}},\ }\bibfield  {title} {\enquote {\bibinfo {title}
  {{Recent Progress in Quantum Simulation Using Superconducting Circuits}},}\
  }\href {\doibase 10.1007/s10909-014-1175-8} {\bibfield  {journal} {\bibinfo
  {journal} {Journal of Low Temperature Physics}\ }\textbf {\bibinfo {volume}
  {175}},\ \bibinfo {pages} {633--654} (\bibinfo {year} {2014})}\BibitemShut
  {NoStop}%
\bibitem [{\citenamefont {Blais}\ \emph {et~al.}(2021)\citenamefont {Blais},
  \citenamefont {Grimsmo}, \citenamefont {Girvin},\ and\ \citenamefont
  {Wallraff}}]{Wallraff}%
  \BibitemOpen
  \bibfield  {author} {\bibinfo {author} {\bibfnamefont {Alexandre}\
  \bibnamefont {Blais}}, \bibinfo {author} {\bibfnamefont {Arne~L.}\
  \bibnamefont {Grimsmo}}, \bibinfo {author} {\bibfnamefont {Steven~M.}\
  \bibnamefont {Girvin}}, \ and\ \bibinfo {author} {\bibfnamefont {Andreas}\
  \bibnamefont {Wallraff}},\ }\bibfield  {title} {\enquote {\bibinfo {title}
  {Circuit quantum electrodynamics},}\ }\href {\doibase
  10.1103/RevModPhys.93.025005} {\bibfield  {journal} {\bibinfo  {journal}
  {Reviews of Modern Physics}\ }\textbf {\bibinfo {volume} {93}},\ \bibinfo
  {pages} {025005} (\bibinfo {year} {2021})}\BibitemShut {NoStop}%
\bibitem [{\citenamefont {Lewenstein}\ \emph {et~al.}(2012)\citenamefont
  {Lewenstein}, \citenamefont {Sanpera},\ and\ \citenamefont
  {Ahufinger}}]{LSA2017}%
  \BibitemOpen
  \bibfield  {author} {\bibinfo {author} {\bibfnamefont {Maciej}\ \bibnamefont
  {Lewenstein}}, \bibinfo {author} {\bibfnamefont {Anna}\ \bibnamefont
  {Sanpera}}, \ and\ \bibinfo {author} {\bibfnamefont {Ver{\`o}nica}\
  \bibnamefont {Ahufinger}},\ }\href@noop {} {\emph {\bibinfo {title}
  {Ultracold Atoms in Optical Lattices: Simulating quantum many-body
  systems}}}\ (\bibinfo  {publisher} {Oxford University Press},\ \bibinfo
  {address} {Oxford},\ \bibinfo {year} {2012})\BibitemShut {NoStop}%
\bibitem [{\citenamefont {Zhang}\ \emph {et~al.}(2017)\citenamefont {Zhang},
  \citenamefont {Pagano}, \citenamefont {Hess}, \citenamefont {Kyprianidis},
  \citenamefont {Becker}, \citenamefont {Kaplan}, \citenamefont {Gorshkov},
  \citenamefont {Gong},\ and\ \citenamefont {Monroe}}]{Monroe53}%
  \BibitemOpen
  \bibfield  {author} {\bibinfo {author} {\bibfnamefont {Jiehang}\ \bibnamefont
  {Zhang}}, \bibinfo {author} {\bibfnamefont {Guido}\ \bibnamefont {Pagano}},
  \bibinfo {author} {\bibfnamefont {Paul~W.}\ \bibnamefont {Hess}}, \bibinfo
  {author} {\bibfnamefont {Antonios}\ \bibnamefont {Kyprianidis}}, \bibinfo
  {author} {\bibfnamefont {Patrick}\ \bibnamefont {Becker}}, \bibinfo {author}
  {\bibfnamefont {Harvey}\ \bibnamefont {Kaplan}}, \bibinfo {author}
  {\bibfnamefont {Alexey~V.}\ \bibnamefont {Gorshkov}}, \bibinfo {author}
  {\bibfnamefont {Zhe-Xuan}\ \bibnamefont {Gong}}, \ and\ \bibinfo {author}
  {\bibfnamefont {Chris}\ \bibnamefont {Monroe}},\ }\bibfield  {title}
  {\enquote {\bibinfo {title} {Observation of a many-body dynamical phase
  transition with a 53-qubit quantum simulator},}\ }\href {\doibase
  10.1038/nature24654} {\bibfield  {journal} {\bibinfo  {journal} {Nature}\
  }\textbf {\bibinfo {volume} {551}},\ \bibinfo {pages} {601--604} (\bibinfo
  {year} {2017})}\BibitemShut {NoStop}%
\bibitem [{\citenamefont {Monroe}\ \emph {et~al.}(2021)\citenamefont {Monroe},
  \citenamefont {Campbell}, \citenamefont {Duan}, \citenamefont {Gong},
  \citenamefont {Gorshkov}, \citenamefont {Hess}, \citenamefont {Islam},
  \citenamefont {Kim}, \citenamefont {Linke}, \citenamefont {Pagano},
  \citenamefont {Richerme}, \citenamefont {Senko},\ and\ \citenamefont
  {Yao}}]{MonroeRMP}%
  \BibitemOpen
  \bibfield  {author} {\bibinfo {author} {\bibfnamefont {Chris}\ \bibnamefont
  {Monroe}}, \bibinfo {author} {\bibfnamefont {Wesley~C.}\ \bibnamefont
  {Campbell}}, \bibinfo {author} {\bibfnamefont {Lu-Ming}\ \bibnamefont
  {Duan}}, \bibinfo {author} {\bibfnamefont {Zhe-Xuan}\ \bibnamefont {Gong}},
  \bibinfo {author} {\bibfnamefont {Alexey~V.}\ \bibnamefont {Gorshkov}},
  \bibinfo {author} {\bibfnamefont {Paul~W.}\ \bibnamefont {Hess}}, \bibinfo
  {author} {\bibfnamefont {Rajibul}\ \bibnamefont {Islam}}, \bibinfo {author}
  {\bibfnamefont {Kihwan}\ \bibnamefont {Kim}}, \bibinfo {author}
  {\bibfnamefont {Norbert~M.}\ \bibnamefont {Linke}}, \bibinfo {author}
  {\bibfnamefont {Guido}\ \bibnamefont {Pagano}}, \bibinfo {author}
  {\bibfnamefont {Phil}\ \bibnamefont {Richerme}}, \bibinfo {author}
  {\bibfnamefont {Crystal}\ \bibnamefont {Senko}}, \ and\ \bibinfo {author}
  {\bibfnamefont {Norman~Y.}\ \bibnamefont {Yao}},\ }\bibfield  {title}
  {\enquote {\bibinfo {title} {Programmable quantum simulations of spin systems
  with trapped ions},}\ }\href {\doibase 10.1103/RevModPhys.93.025001}
  {\bibfield  {journal} {\bibinfo  {journal} {Reviews of Modern Physics}\
  }\textbf {\bibinfo {volume} {93}},\ \bibinfo {pages} {025001} (\bibinfo
  {year} {2021})}\BibitemShut {NoStop}%
\bibitem [{\citenamefont {Ringbauer}\ \emph {et~al.}(2022)\citenamefont
  {Ringbauer}, \citenamefont {Meth}, \citenamefont {Postler}, \citenamefont
  {Stricker}, \citenamefont {Blatt}, \citenamefont {Schindler},\ and\
  \citenamefont {Monz}}]{Ringbauer}%
  \BibitemOpen
  \bibfield  {author} {\bibinfo {author} {\bibfnamefont {Martin}\ \bibnamefont
  {Ringbauer}}, \bibinfo {author} {\bibfnamefont {Michael}\ \bibnamefont
  {Meth}}, \bibinfo {author} {\bibfnamefont {Lukas}\ \bibnamefont {Postler}},
  \bibinfo {author} {\bibfnamefont {Roman}\ \bibnamefont {Stricker}}, \bibinfo
  {author} {\bibfnamefont {Rainer}\ \bibnamefont {Blatt}}, \bibinfo {author}
  {\bibfnamefont {Philipp}\ \bibnamefont {Schindler}}, \ and\ \bibinfo {author}
  {\bibfnamefont {Thomas}\ \bibnamefont {Monz}},\ }\bibfield  {title} {\enquote
  {\bibinfo {title} {A universal qudit quantum processor with trapped ions},}\
  }\href {\doibase 10.1038/s41567-022-01658-0} {\bibfield  {journal} {\bibinfo
  {journal} {Nature Physics}\ }\textbf {\bibinfo {volume} {18}},\ \bibinfo
  {pages} {1053--1067} (\bibinfo {year} {2022})},\ \bibinfo {note}
  {arXiv[quant-ph]:2109.06903}\BibitemShut {NoStop}%
\bibitem [{\citenamefont {Bernien}\ \emph {et~al.}(2017)\citenamefont
  {Bernien}, \citenamefont {Schwartz}, \citenamefont {Keesling}, \citenamefont
  {Levine}, \citenamefont {Omran}, \citenamefont {Pichler}, \citenamefont
  {Choi}, \citenamefont {Zibrov}, \citenamefont {Endres}, \citenamefont
  {Greiner}, \citenamefont {Vuletic},\ and\ \citenamefont {Lukin}}]{Lukin51}%
  \BibitemOpen
  \bibfield  {author} {\bibinfo {author} {\bibfnamefont {Hannes}\ \bibnamefont
  {Bernien}}, \bibinfo {author} {\bibfnamefont {Sylvain}\ \bibnamefont
  {Schwartz}}, \bibinfo {author} {\bibfnamefont {Alexander}\ \bibnamefont
  {Keesling}}, \bibinfo {author} {\bibfnamefont {Harry}\ \bibnamefont
  {Levine}}, \bibinfo {author} {\bibfnamefont {Ahmed}\ \bibnamefont {Omran}},
  \bibinfo {author} {\bibfnamefont {Hannes}\ \bibnamefont {Pichler}}, \bibinfo
  {author} {\bibfnamefont {Soonwon}\ \bibnamefont {Choi}}, \bibinfo {author}
  {\bibfnamefont {Alexander~S.}\ \bibnamefont {Zibrov}}, \bibinfo {author}
  {\bibfnamefont {Manuel}\ \bibnamefont {Endres}}, \bibinfo {author}
  {\bibfnamefont {Markus}\ \bibnamefont {Greiner}}, \bibinfo {author}
  {\bibfnamefont {Vladan}\ \bibnamefont {Vuletic}}, \ and\ \bibinfo {author}
  {\bibfnamefont {Mikhail~D.}\ \bibnamefont {Lukin}},\ }\bibfield  {title}
  {\enquote {\bibinfo {title} {Probing many-body dynamics on a 51-atom quantum
  simulator},}\ }\href {http://dx.doi.org/10.1038/nature24622} {\bibfield
  {journal} {\bibinfo  {journal} {Nature}\ }\textbf {\bibinfo {volume} {551}},\
  \bibinfo {pages} {579--} (\bibinfo {year} {2017})}\BibitemShut {NoStop}%
\bibitem [{\citenamefont {Bluvstein}\ \emph {et~al.}(2021)\citenamefont
  {Bluvstein}, \citenamefont {Omran}, \citenamefont {Levine}, \citenamefont
  {Keesling}, \citenamefont {Semeghini}, \citenamefont {Ebadi}, \citenamefont
  {Wang}, \citenamefont {Michailidis}, \citenamefont {Maskara}, \citenamefont
  {Ho}, \citenamefont {Choi}, \citenamefont {Serbyn}, \citenamefont {Greiner},
  \citenamefont {Vuletic},\ and\ \citenamefont {Lukin}}]{LukinScars}%
  \BibitemOpen
  \bibfield  {author} {\bibinfo {author} {\bibfnamefont {Dolev}\ \bibnamefont
  {Bluvstein}}, \bibinfo {author} {\bibfnamefont {Ahmed}\ \bibnamefont
  {Omran}}, \bibinfo {author} {\bibfnamefont {Harry}\ \bibnamefont {Levine}},
  \bibinfo {author} {\bibfnamefont {Alexander}\ \bibnamefont {Keesling}},
  \bibinfo {author} {\bibfnamefont {Giulia}\ \bibnamefont {Semeghini}},
  \bibinfo {author} {\bibfnamefont {Sepehr}\ \bibnamefont {Ebadi}}, \bibinfo
  {author} {\bibfnamefont {Tout~T.}\ \bibnamefont {Wang}}, \bibinfo {author}
  {\bibfnamefont {Alexios~A.}\ \bibnamefont {Michailidis}}, \bibinfo {author}
  {\bibfnamefont {Nishad}\ \bibnamefont {Maskara}}, \bibinfo {author}
  {\bibfnamefont {Wen~W.}\ \bibnamefont {Ho}}, \bibinfo {author} {\bibfnamefont
  {Soonwon}\ \bibnamefont {Choi}}, \bibinfo {author} {\bibfnamefont {Maksym}\
  \bibnamefont {Serbyn}}, \bibinfo {author} {\bibfnamefont {Markus}\
  \bibnamefont {Greiner}}, \bibinfo {author} {\bibfnamefont {Vladan}\
  \bibnamefont {Vuletic}}, \ and\ \bibinfo {author} {\bibfnamefont
  {Mikhail~D.}\ \bibnamefont {Lukin}},\ }\bibfield  {title} {\enquote {\bibinfo
  {title} {Controlling quantum many-body dynamics in driven {Rydberg} atom
  arrays},}\ }\href {\doibase 10.1126/science.abg2530} {\bibfield  {journal}
  {\bibinfo  {journal} {Science}\ }\textbf {\bibinfo {volume} {371}},\ \bibinfo
  {pages} {1355--1359} (\bibinfo {year} {2021})}\BibitemShut {NoStop}%
\bibitem [{\citenamefont {Scholl}\ \emph {et~al.}(2021)\citenamefont {Scholl},
  \citenamefont {Schuler}, \citenamefont {Williams}, \citenamefont
  {Eberharter}, \citenamefont {Barredo}, \citenamefont {Schymik}, \citenamefont
  {Lienhard}, \citenamefont {Henry}, \citenamefont {Lang}, \citenamefont
  {Lahaye}, \citenamefont {L{\"a}uchli},\ and\ \citenamefont
  {Browaeys}}]{Browaeys}%
  \BibitemOpen
  \bibfield  {author} {\bibinfo {author} {\bibfnamefont {Pascal}\ \bibnamefont
  {Scholl}}, \bibinfo {author} {\bibfnamefont {Michael}\ \bibnamefont
  {Schuler}}, \bibinfo {author} {\bibfnamefont {Hannah~J.}\ \bibnamefont
  {Williams}}, \bibinfo {author} {\bibfnamefont {Alexander~A.}\ \bibnamefont
  {Eberharter}}, \bibinfo {author} {\bibfnamefont {Daniel}\ \bibnamefont
  {Barredo}}, \bibinfo {author} {\bibfnamefont {Kai-Niklas}\ \bibnamefont
  {Schymik}}, \bibinfo {author} {\bibfnamefont {Vincent}\ \bibnamefont
  {Lienhard}}, \bibinfo {author} {\bibfnamefont {Louis-Paul}\ \bibnamefont
  {Henry}}, \bibinfo {author} {\bibfnamefont {Thomas~C.}\ \bibnamefont {Lang}},
  \bibinfo {author} {\bibfnamefont {Thierry}\ \bibnamefont {Lahaye}}, \bibinfo
  {author} {\bibfnamefont {Andreas~M.}\ \bibnamefont {L{\"a}uchli}}, \ and\
  \bibinfo {author} {\bibfnamefont {Antoine}\ \bibnamefont {Browaeys}},\
  }\bibfield  {title} {\enquote {\bibinfo {title} {Quantum simulation of 2d
  antiferromagnets with hundreds of {Rydberg} atoms},}\ }\href {\doibase
  10.1038/s41586-021-03585-1} {\bibfield  {journal} {\bibinfo  {journal}
  {Nature}\ }\textbf {\bibinfo {volume} {595}},\ \bibinfo {pages} {233--238}
  (\bibinfo {year} {2021})}\BibitemShut {NoStop}%
\bibitem [{\citenamefont {Schlawin}\ \emph {et~al.}(2022)\citenamefont
  {Schlawin}, \citenamefont {Kennes},\ and\ \citenamefont
  {Sentef}}]{Schlawin_2022}%
  \BibitemOpen
  \bibfield  {author} {\bibinfo {author} {\bibfnamefont {Frank}\ \bibnamefont
  {Schlawin}}, \bibinfo {author} {\bibfnamefont {Dante~M.}\ \bibnamefont
  {Kennes}}, \ and\ \bibinfo {author} {\bibfnamefont {Michael~A.}\ \bibnamefont
  {Sentef}},\ }\bibfield  {title} {\enquote {\bibinfo {title} {Cavity quantum
  materials},}\ }\href {\doibase 10.1063/5.0083825} {\bibfield  {journal}
  {\bibinfo  {journal} {Applied Physics Reviews}\ }\textbf {\bibinfo {volume}
  {9}},\ \bibinfo {pages} {011312} (\bibinfo {year} {2022})}\BibitemShut
  {NoStop}%
\bibitem [{\citenamefont {Clark}\ \emph {et~al.}(2020)\citenamefont {Clark},
  \citenamefont {Schine}, \citenamefont {Baum}, \citenamefont {Jia},\ and\
  \citenamefont {Simon}}]{Clark_2020}%
  \BibitemOpen
  \bibfield  {author} {\bibinfo {author} {\bibfnamefont {Logan~W.}\
  \bibnamefont {Clark}}, \bibinfo {author} {\bibfnamefont {Nathan}\
  \bibnamefont {Schine}}, \bibinfo {author} {\bibfnamefont {Claire}\
  \bibnamefont {Baum}}, \bibinfo {author} {\bibfnamefont {Ningyuan}\
  \bibnamefont {Jia}}, \ and\ \bibinfo {author} {\bibfnamefont {Jonathan}\
  \bibnamefont {Simon}},\ }\bibfield  {title} {\enquote {\bibinfo {title}
  {{Observation of Laughlin states made of light}},}\ }\href {\doibase
  10.1038/s41586-020-2318-5} {\bibfield  {journal} {\bibinfo  {journal}
  {Nature}\ }\textbf {\bibinfo {volume} {582}},\ \bibinfo {pages} {41--45}
  (\bibinfo {year} {2020})}\BibitemShut {NoStop}%
\bibitem [{\citenamefont {Carusotto}\ \emph {et~al.}(2020)\citenamefont
  {Carusotto}, \citenamefont {Houck}, \citenamefont {Koll{\'a}r}, \citenamefont
  {Roushan}, \citenamefont {Schuster},\ and\ \citenamefont
  {Simon}}]{carusotto_2020}%
  \BibitemOpen
  \bibfield  {author} {\bibinfo {author} {\bibfnamefont {Iacopo}\ \bibnamefont
  {Carusotto}}, \bibinfo {author} {\bibfnamefont {Andrew}\ \bibnamefont
  {Houck}}, \bibinfo {author} {\bibfnamefont {Alicia}\ \bibnamefont
  {Koll{\'a}r}}, \bibinfo {author} {\bibfnamefont {Pedram}\ \bibnamefont
  {Roushan}}, \bibinfo {author} {\bibfnamefont {David}\ \bibnamefont
  {Schuster}}, \ and\ \bibinfo {author} {\bibfnamefont {Jonathan}\ \bibnamefont
  {Simon}},\ }\bibfield  {title} {\enquote {\bibinfo {title} {Photonic
  materials in circuit quantum electrodynamics},}\ }\href {\doibase
  10.1038/s41567-020-0815-y} {\bibfield  {journal} {\bibinfo  {journal} {Nature
  Physics}\ }\textbf {\bibinfo {volume} {16}},\ \bibinfo {pages} {268--279}
  (\bibinfo {year} {2020})}\BibitemShut {NoStop}%
\bibitem [{\citenamefont {Ma}\ \emph {et~al.}(2019)\citenamefont {Ma},
  \citenamefont {Saxberg}, \citenamefont {Owens}, \citenamefont {Leung},
  \citenamefont {Lu}, \citenamefont {Simon},\ and\ \citenamefont
  {Schuster}}]{Ma_2019}%
  \BibitemOpen
  \bibfield  {author} {\bibinfo {author} {\bibfnamefont {Ruichao}\ \bibnamefont
  {Ma}}, \bibinfo {author} {\bibfnamefont {Brendan}\ \bibnamefont {Saxberg}},
  \bibinfo {author} {\bibfnamefont {Clai}\ \bibnamefont {Owens}}, \bibinfo
  {author} {\bibfnamefont {Nelson}\ \bibnamefont {Leung}}, \bibinfo {author}
  {\bibfnamefont {Yao}\ \bibnamefont {Lu}}, \bibinfo {author} {\bibfnamefont
  {Jonathan}\ \bibnamefont {Simon}}, \ and\ \bibinfo {author} {\bibfnamefont
  {David~I.}\ \bibnamefont {Schuster}},\ }\bibfield  {title} {\enquote
  {\bibinfo {title} {A dissipatively stabilized mott insulator of photons},}\
  }\href {\doibase 10.1038/s41586-019-0897-9} {\bibfield  {journal} {\bibinfo
  {journal} {Nature}\ }\textbf {\bibinfo {volume} {566}},\ \bibinfo {pages}
  {51--57} (\bibinfo {year} {2019})}\BibitemShut {NoStop}%
\bibitem [{\citenamefont {Schine}\ \emph {et~al.}(2016)\citenamefont {Schine},
  \citenamefont {Ryou}, \citenamefont {Gromov}, \citenamefont {Sommer},\ and\
  \citenamefont {Simon}}]{Schine_2016}%
  \BibitemOpen
  \bibfield  {author} {\bibinfo {author} {\bibfnamefont {Nathan}\ \bibnamefont
  {Schine}}, \bibinfo {author} {\bibfnamefont {Albert}\ \bibnamefont {Ryou}},
  \bibinfo {author} {\bibfnamefont {Andrey}\ \bibnamefont {Gromov}}, \bibinfo
  {author} {\bibfnamefont {Ariel}\ \bibnamefont {Sommer}}, \ and\ \bibinfo
  {author} {\bibfnamefont {Jonathan}\ \bibnamefont {Simon}},\ }\bibfield
  {title} {\enquote {\bibinfo {title} {Synthetic {Landau} levels for
  photons},}\ }\href {\doibase 10.1038/nature17943} {\bibfield  {journal}
  {\bibinfo  {journal} {Nature}\ }\textbf {\bibinfo {volume} {534}},\ \bibinfo
  {pages} {671--675} (\bibinfo {year} {2016})}\BibitemShut {NoStop}%
\bibitem [{\citenamefont {Cao}\ \emph {et~al.}(2020)\citenamefont {Cao},
  \citenamefont {Rodan-Legrain}, \citenamefont {Rubies-Bigorda}, \citenamefont
  {Park}, \citenamefont {Watanabe}, \citenamefont {Taniguchi},\ and\
  \citenamefont {Jarillo-Herrero}}]{Pablo}%
  \BibitemOpen
  \bibfield  {author} {\bibinfo {author} {\bibfnamefont {Yuan}\ \bibnamefont
  {Cao}}, \bibinfo {author} {\bibfnamefont {Daniel}\ \bibnamefont
  {Rodan-Legrain}}, \bibinfo {author} {\bibfnamefont {Oriol}\ \bibnamefont
  {Rubies-Bigorda}}, \bibinfo {author} {\bibfnamefont {Jeong~Min}\ \bibnamefont
  {Park}}, \bibinfo {author} {\bibfnamefont {Kenji}\ \bibnamefont {Watanabe}},
  \bibinfo {author} {\bibfnamefont {Takashi}\ \bibnamefont {Taniguchi}}, \ and\
  \bibinfo {author} {\bibfnamefont {Pablo}\ \bibnamefont {Jarillo-Herrero}},\
  }\bibfield  {title} {\enquote {\bibinfo {title} {Tunable correlated states
  and spin-polarized phases in twisted bilayer{\textendash}bilayer graphene},}\
  }\href {\doibase 10.1038/s41586-020-2260-6} {\bibfield  {journal} {\bibinfo
  {journal} {Nature}\ }\textbf {\bibinfo {volume} {583}},\ \bibinfo {pages}
  {215--220} (\bibinfo {year} {2020})}\BibitemShut {NoStop}%
\bibitem [{\citenamefont {Stepanov}\ \emph {et~al.}(2020)\citenamefont
  {Stepanov}, \citenamefont {Das}, \citenamefont {Lu}, \citenamefont
  {Fahimniya}, \citenamefont {Watanabe}, \citenamefont {Taniguchi},
  \citenamefont {Koppens}, \citenamefont {Lischner}, \citenamefont {Levitov},\
  and\ \citenamefont {Efetov}}]{Dima}%
  \BibitemOpen
  \bibfield  {author} {\bibinfo {author} {\bibfnamefont {Petr}\ \bibnamefont
  {Stepanov}}, \bibinfo {author} {\bibfnamefont {Ipsita}\ \bibnamefont {Das}},
  \bibinfo {author} {\bibfnamefont {Xiaobo}\ \bibnamefont {Lu}}, \bibinfo
  {author} {\bibfnamefont {Ali}\ \bibnamefont {Fahimniya}}, \bibinfo {author}
  {\bibfnamefont {Kenji}\ \bibnamefont {Watanabe}}, \bibinfo {author}
  {\bibfnamefont {Takashi}\ \bibnamefont {Taniguchi}}, \bibinfo {author}
  {\bibfnamefont {Frank H.~L.}\ \bibnamefont {Koppens}}, \bibinfo {author}
  {\bibfnamefont {Johannes}\ \bibnamefont {Lischner}}, \bibinfo {author}
  {\bibfnamefont {Leonid}\ \bibnamefont {Levitov}}, \ and\ \bibinfo {author}
  {\bibfnamefont {Dmitri~K.}\ \bibnamefont {Efetov}},\ }\bibfield  {title}
  {\enquote {\bibinfo {title} {Untying the insulating and superconducting
  orders in magic-angle graphene},}\ }\href {\doibase
  10.1038/s41586-020-2459-6} {\bibfield  {journal} {\bibinfo  {journal}
  {Nature}\ }\textbf {\bibinfo {volume} {583}},\ \bibinfo {pages} {375--378}
  (\bibinfo {year} {2020})}\BibitemShut {NoStop}%
\bibitem [{\citenamefont {Kennes}\ \emph {et~al.}(2021)\citenamefont {Kennes},
  \citenamefont {Claassen}, \citenamefont {Xian}, \citenamefont {Georges},
  \citenamefont {Millis}, \citenamefont {Hone}, \citenamefont {Dean},
  \citenamefont {Basov}, \citenamefont {Pasupathy},\ and\ \citenamefont
  {Rubio}}]{Kennes21}%
  \BibitemOpen
  \bibfield  {author} {\bibinfo {author} {\bibfnamefont {Dante}\ \bibnamefont
  {Kennes}}, \bibinfo {author} {\bibfnamefont {Martin}\ \bibnamefont
  {Claassen}}, \bibinfo {author} {\bibfnamefont {Lede}\ \bibnamefont {Xian}},
  \bibinfo {author} {\bibfnamefont {Antoine}\ \bibnamefont {Georges}}, \bibinfo
  {author} {\bibfnamefont {Andrew}\ \bibnamefont {Millis}}, \bibinfo {author}
  {\bibfnamefont {James}\ \bibnamefont {Hone}}, \bibinfo {author}
  {\bibfnamefont {Cory}\ \bibnamefont {Dean}}, \bibinfo {author} {\bibfnamefont
  {D.}~\bibnamefont {Basov}}, \bibinfo {author} {\bibfnamefont {Abhay}\
  \bibnamefont {Pasupathy}}, \ and\ \bibinfo {author} {\bibfnamefont {Angel}\
  \bibnamefont {Rubio}},\ }\bibfield  {title} {\enquote {\bibinfo {title}
  {Moir{\'e} heterostructures as a condensed-matter quantum simulator},}\ }\href
  {\doibase 10.1038/s41567-020-01154-3} {\bibfield  {journal} {\bibinfo
  {journal} {Nature Physics}\ }\textbf {\bibinfo {volume} {17}},\ \bibinfo
  {pages} {1-9} (\bibinfo {year} {2021})}\BibitemShut {NoStop}%
\bibitem [{\citenamefont {Salamon}\ \emph {et~al.}(2020)\citenamefont
  {Salamon}, \citenamefont {Celi}, \citenamefont {Chhajlany}, \citenamefont
  {Fr\'erot}, \citenamefont {Lewenstein}, \citenamefont {Tarruell},\ and\
  \citenamefont {Rakshit}}]{tymek}%
  \BibitemOpen
  \bibfield  {author} {\bibinfo {author} {\bibfnamefont {Tymoteusz}\
  \bibnamefont {Salamon}}, \bibinfo {author} {\bibfnamefont {Alessio}\
  \bibnamefont {Celi}}, \bibinfo {author} {\bibfnamefont {Ravindra~W.}\
  \bibnamefont {Chhajlany}}, \bibinfo {author} {\bibfnamefont {Ir\'en\'ee}\
  \bibnamefont {Fr\'erot}}, \bibinfo {author} {\bibfnamefont {Maciej}\
  \bibnamefont {Lewenstein}}, \bibinfo {author} {\bibfnamefont {Leticia}\
  \bibnamefont {Tarruell}}, \ and\ \bibinfo {author} {\bibfnamefont {Debraj}\
  \bibnamefont {Rakshit}},\ }\bibfield  {title} {\enquote {\bibinfo {title}
  {Simulating twistronics without a twist},}\ }\href {\doibase
  10.1103/PhysRevLett.125.030504} {\bibfield  {journal} {\bibinfo  {journal}
  {Physical Review Letters}\ }\textbf {\bibinfo {volume} {125}},\ \bibinfo
  {pages} {030504} (\bibinfo {year} {2020})}\BibitemShut {NoStop}%
\bibitem [{\citenamefont {Basov}\ \emph {et~al.}(2021)\citenamefont {Basov},
  \citenamefont {Asenjo-Garcia}, \citenamefont {Schuck}, \citenamefont {Zhu},\
  and\ \citenamefont {Rubio}}]{Basov2021}%
  \BibitemOpen
  \bibfield  {author} {\bibinfo {author} {\bibfnamefont {Dmitri~N.}\
  \bibnamefont {Basov}}, \bibinfo {author} {\bibfnamefont {Ana}\ \bibnamefont
  {Asenjo-Garcia}}, \bibinfo {author} {\bibfnamefont {P.~James}\ \bibnamefont
  {Schuck}}, \bibinfo {author} {\bibfnamefont {Xiaoyang}\ \bibnamefont {Zhu}},
  \ and\ \bibinfo {author} {\bibfnamefont {Angel}\ \bibnamefont {Rubio}},\
  }\bibfield  {title} {\enquote {\bibinfo {title} {Polariton panorama},}\
  }\href {\doibase doi:10.1515/nanoph-2020-0449} {\bibfield  {journal}
  {\bibinfo  {journal} {Nanophotonics}\ }\textbf {\bibinfo {volume} {10}},\
  \bibinfo {pages} {549--577} (\bibinfo {year} {2021})}\BibitemShut {NoStop}%
\bibitem [{\citenamefont {H{\"u}bener}\ \emph {et~al.}(2021)\citenamefont
  {H{\"u}bener}, \citenamefont {De~Giovannini}, \citenamefont {Sch{\"a}fer},
  \citenamefont {Andberger}, \citenamefont {Ruggenthaler}, \citenamefont
  {Faist},\ and\ \citenamefont {Rubio}}]{hubener2021}%
  \BibitemOpen
  \bibfield  {author} {\bibinfo {author} {\bibfnamefont {Hannes}\ \bibnamefont
  {H{\"u}bener}}, \bibinfo {author} {\bibfnamefont {Umberto}\ \bibnamefont
  {De~Giovannini}}, \bibinfo {author} {\bibfnamefont {Christian}\ \bibnamefont
  {Sch{\"a}fer}}, \bibinfo {author} {\bibfnamefont {Johan}\ \bibnamefont
  {Andberger}}, \bibinfo {author} {\bibfnamefont {Michael}\ \bibnamefont
  {Ruggenthaler}}, \bibinfo {author} {\bibfnamefont {Jerome}\ \bibnamefont
  {Faist}}, \ and\ \bibinfo {author} {\bibfnamefont {Angel}\ \bibnamefont
  {Rubio}},\ }\bibfield  {title} {\enquote {\bibinfo {title} {Engineering
  quantum materials with chiral optical cavities},}\ }\href {\doibase
  10.1038/s41563-020-00801-7} {\bibfield  {journal} {\bibinfo  {journal}
  {Nature Materials}\ }\textbf {\bibinfo {volume} {20}},\ \bibinfo {pages}
  {438--442} (\bibinfo {year} {2021})}\BibitemShut {NoStop}%
\bibitem [{\citenamefont {Boulier}\ \emph {et~al.}(2020)\citenamefont
  {Boulier}, \citenamefont {Jacquet}, \citenamefont {Ma{\^{i}}tre}, \citenamefont
  {Lerario}, \citenamefont {Claude}, \citenamefont {Pigeon}, \citenamefont
  {Glorieux}, \citenamefont {Amo}, \citenamefont {Bloch}, \citenamefont
  {Bramati},\ and\ \citenamefont {Giacobino}}]{JBloch}%
  \BibitemOpen
  \bibfield  {author} {\bibinfo {author} {\bibfnamefont {Thomas}\ \bibnamefont
  {Boulier}}, \bibinfo {author} {\bibfnamefont {Maxime~J.}\ \bibnamefont
  {Jacquet}}, \bibinfo {author} {\bibfnamefont {Anne}\ \bibnamefont {Ma{\^{i}}tre}},
  \bibinfo {author} {\bibfnamefont {Giovanni}\ \bibnamefont {Lerario}},
  \bibinfo {author} {\bibfnamefont {Ferdinand}\ \bibnamefont {Claude}},
  \bibinfo {author} {\bibfnamefont {Simon}\ \bibnamefont {Pigeon}}, \bibinfo
  {author} {\bibfnamefont {Quentin}\ \bibnamefont {Glorieux}}, \bibinfo
  {author} {\bibfnamefont {Alberto}\ \bibnamefont {Amo}}, \bibinfo {author}
  {\bibfnamefont {Jacqueline}\ \bibnamefont {Bloch}}, \bibinfo {author}
  {\bibfnamefont {Alberto}\ \bibnamefont {Bramati}}, \ and\ \bibinfo {author}
  {\bibfnamefont {Elisabeth}\ \bibnamefont {Giacobino}},\ }\bibfield  {title}
  {\enquote {\bibinfo {title} {Microcavity polaritons for quantum
  simulation},}\ }\href {\doibase https://doi.org/10.1002/qute.202000052}
  {\bibfield  {journal} {\bibinfo  {journal} {Advanced Quantum Technologies}\
  }\textbf {\bibinfo {volume} {3}},\ \bibinfo {pages} {2000052} (\bibinfo
  {year} {2020})}\BibitemShut {NoStop}%
\bibitem [{\citenamefont {Chubb}\ \emph {et~al.}(2018)\citenamefont {Chubb},
  \citenamefont {Tomamichel},\ and\ \citenamefont
  {Korzekwa}}]{Korzekwa:heatcap}%
  \BibitemOpen
  \bibfield  {author} {\bibinfo {author} {\bibfnamefont {Christopher~T.}\
  \bibnamefont {Chubb}}, \bibinfo {author} {\bibfnamefont {Marco}\ \bibnamefont
  {Tomamichel}}, \ and\ \bibinfo {author} {\bibfnamefont {Kamil}\ \bibnamefont
  {Korzekwa}},\ }\bibfield  {title} {\enquote {\bibinfo {title} {Beyond the
  thermodynamic limit: finite-size corrections to state interconversion
  rates},}\ }\href {\doibase 10.22331/q-2018-11-27-108} {\bibfield  {journal}
  {\bibinfo  {journal} {Quantum}\ }\textbf {\bibinfo {volume} {2}},\ \bibinfo
  {pages} {108} (\bibinfo {year} {2018})},\ \bibinfo {note}
  {arXiv[quant-ph]:1711.01193v5}\BibitemShut {NoStop}%
\bibitem [{\citenamefont {Manzano}\ \emph {et~al.}(2022)\citenamefont
  {Manzano}, \citenamefont {Parrondo},\ and\ \citenamefont
  {Landi}}]{Experiment1}%
  \BibitemOpen
  \bibfield  {author} {\bibinfo {author} {\bibfnamefont {Gonzalo}\ \bibnamefont
  {Manzano}}, \bibinfo {author} {\bibfnamefont {Juan~M.R.}\ \bibnamefont
  {Parrondo}}, \ and\ \bibinfo {author} {\bibfnamefont {Gabriel~T.}\
  \bibnamefont {Landi}},\ }\bibfield  {title} {\enquote {\bibinfo {title}
  {{Non-Abelian Quantum Transport and Thermosqueezing Effects}},}\ }\href
  {\doibase 10.1103/PRXQuantum.3.010304} {\bibfield  {journal} {\bibinfo
  {journal} {Physical Review X Quantum}\ }\textbf {\bibinfo {volume} {3}},\
  \bibinfo {pages} {010304} (\bibinfo {year} {2022})}\BibitemShut {NoStop}%
\bibitem [{\citenamefont {Yunger~Halpern}\ \emph {et~al.}(2020)\citenamefont
  {Yunger~Halpern}, \citenamefont {Beverland},\ and\ \citenamefont
  {Kalev}}]{Experiment2}%
  \BibitemOpen
  \bibfield  {author} {\bibinfo {author} {\bibfnamefont {Nicole}\ \bibnamefont
  {Yunger~Halpern}}, \bibinfo {author} {\bibfnamefont {Michael~E.}\
  \bibnamefont {Beverland}}, \ and\ \bibinfo {author} {\bibfnamefont {Amir}\
  \bibnamefont {Kalev}},\ }\bibfield  {title} {\enquote {\bibinfo {title}
  {Noncommuting conserved charges in quantum many-body thermalization},}\
  }\href {\doibase 10.1103/PhysRevE.101.042117} {\bibfield  {journal} {\bibinfo
   {journal} {Physical Review E}\ }\textbf {\bibinfo {volume} {101}},\ \bibinfo
  {pages} {042117} (\bibinfo {year} {2020})}\BibitemShut {NoStop}%
\bibitem [{\citenamefont {Csisz\'ar}\ and\ \citenamefont
  {K\"orner}(2011)}]{csiszar_korner_2011}%
  \BibitemOpen
  \bibfield  {author} {\bibinfo {author} {\bibfnamefont {Imre}\ \bibnamefont
  {Csisz\'ar}}\ and\ \bibinfo {author} {\bibfnamefont {J\'anos}\ \bibnamefont
  {K\"orner}},\ }\href {\doibase 10.1017/CBO9780511921889} {\emph {\bibinfo
  {title} {Information Theory: Coding Theorems for Discrete Memoryless
  Systems}}},\ \bibinfo {edition} {2nd}\ ed.\ (\bibinfo  {publisher} {Cambridge
  University Press},\ \bibinfo {year} {2011})\BibitemShut {NoStop}%
\bibitem [{\citenamefont {{Winter}}(1999)}]{winter1999_2}%
  \BibitemOpen
  \bibfield  {author} {\bibinfo {author} {\bibfnamefont {Andreas}\ \bibnamefont
  {{Winter}}},\ }\bibfield  {title} {\enquote {\bibinfo {title} {Coding theorem
  and strong converse for quantum channels},}\ }\href {\doibase
  10.1109/18.796385} {\bibfield  {journal} {\bibinfo  {journal} {IEEE
  Transactions on Information Theory}\ }\textbf {\bibinfo {volume} {45}},\
  \bibinfo {pages} {2481--2485} (\bibinfo {year} {1999})}\BibitemShut {NoStop}%
\bibitem [{\citenamefont {{Ogawa}}\ and\ \citenamefont
  {{Nagaoka}}(2007)}]{Ogawa2007}%
  \BibitemOpen
  \bibfield  {author} {\bibinfo {author} {\bibfnamefont {Tomohiro}\
  \bibnamefont {{Ogawa}}}\ and\ \bibinfo {author} {\bibfnamefont {Hiroshi}\
  \bibnamefont {{Nagaoka}}},\ }\bibfield  {title} {\enquote {\bibinfo {title}
  {{Making Good Codes for Classical-Quantum Channel Coding via Quantum
  Hypothesis Testing}},}\ }\href {\doibase 10.1109/TIT.2007.896874} {\bibfield
  {journal} {\bibinfo  {journal} {IEEE Transactions on Information Theory}\
  }\textbf {\bibinfo {volume} {53}},\ \bibinfo {pages} {2261--2266} (\bibinfo
  {year} {2007})}\BibitemShut {NoStop}%
\bibitem [{\citenamefont {Wilde}(2013)}]{wilde_2013}%
  \BibitemOpen
  \bibfield  {author} {\bibinfo {author} {\bibfnamefont {Mark~M.}\ \bibnamefont
  {Wilde}},\ }\href {\doibase 10.1017/CBO9781139525343} {\emph {\bibinfo
  {title} {Quantum Information Theory}}}\ (\bibinfo  {publisher} {Cambridge
  University Press},\ \bibinfo {year} {2013})\BibitemShut {NoStop}%
\bibitem [{\citenamefont {Bhatia}(1997)}]{bhatia97}%
  \BibitemOpen
  \bibfield  {author} {\bibinfo {author} {\bibfnamefont {Rajendra}\
  \bibnamefont {Bhatia}},\ }\href {\doibase 10.1007/978-1-4612-0653-8} {\emph
  {\bibinfo {title} {Matrix Analysis}}},\ \bibinfo {series} {Graduate Texts in
  Mathematics}, Vol.\ \bibinfo {volume} {169}\ (\bibinfo  {publisher} {Springer
  Verlag},\ \bibinfo {year} {1997})\BibitemShut {NoStop}%
\bibitem [{\citenamefont {Dembo}\ and\ \citenamefont
  {Zeitouni}(2010)}]{DemboZeitouni}%
  \BibitemOpen
  \bibfield  {author} {\bibinfo {author} {\bibfnamefont {Amit}\ \bibnamefont
  {Dembo}}\ and\ \bibinfo {author} {\bibfnamefont {Ofer}\ \bibnamefont
  {Zeitouni}},\ }\href {\doibase 10.1007/978-3-642-03311-7} {\emph {\bibinfo
  {title} {Large Deviations: Techniques and Applications}}},\ \bibinfo
  {edition} {2nd}\ ed.,\ \bibinfo {series} {Stochastic Modelling and Applied
  Probability}, Vol.~\bibinfo {volume} {38}\ (\bibinfo  {publisher} {Springer
  Verlag},\ \bibinfo {year} {2010})\BibitemShut {NoStop}%
\bibitem [{\citenamefont {Ogata}(2013)}]{Ogata}%
  \BibitemOpen
  \bibfield  {author} {\bibinfo {author} {\bibfnamefont {Yoshiko}\ \bibnamefont
  {Ogata}},\ }\bibfield  {title} {\enquote {\bibinfo {title} {Approximating
  macroscopic observables in quantum spin systems with commuting matrices},}\
  }\href {\doibase 10.1016/j.jfa.2013.01.021} {\bibfield  {journal} {\bibinfo
  {journal} {Journal of Functional Analysis}\ }\textbf {\bibinfo {volume}
  {264}},\ \bibinfo {pages} {2005--2033} (\bibinfo {year} {2013})}\BibitemShut
  {NoStop}%
\bibitem [{\citenamefont {{Duan}}\ \emph {et~al.}(2016)\citenamefont {{Duan}},
  \citenamefont {{Severini}},\ and\ \citenamefont {{Winter}}}]{Duan2016}%
  \BibitemOpen
  \bibfield  {author} {\bibinfo {author} {\bibfnamefont {Runyao}\ \bibnamefont
  {{Duan}}}, \bibinfo {author} {\bibfnamefont {Simone}\ \bibnamefont
  {{Severini}}}, \ and\ \bibinfo {author} {\bibfnamefont {Andreas}\
  \bibnamefont {{Winter}}},\ }\bibfield  {title} {\enquote {\bibinfo {title}
  {{On Zero-Error Communication via Quantum Channels in the Presence of
  Noiseless Feedback}},}\ }\href {\doibase 10.1109/TIT.2016.2562580} {\bibfield
   {journal} {\bibinfo  {journal} {IEEE Transactions on Information Theory}\
  }\textbf {\bibinfo {volume} {62}},\ \bibinfo {pages} {5260--5277} (\bibinfo
  {year} {2016})}\BibitemShut {NoStop}%
\end{thebibliography}
\end{document}